\pdfoutput=1
\documentclass[12pt,reqno]{article}
\usepackage[markup=underlined]{changes}
\RequirePackage{etex}

\usepackage{mathpazo}

\usepackage{etex}
\usepackage{graphicx}
\usepackage{tikz}
\usetikzlibrary{arrows.meta, calc, positioning}
\usepackage{setspace}
\usepackage{amsmath}
\usepackage{amsfonts}
\usepackage{titlesec}
\usepackage{pictex}
\usepackage{comment}
\usepackage{amsthm}
\usepackage{graphics,epsfig,verbatim,bm,latexsym,url,amsbsy}
\usepackage{rotating}
\usepackage[authoryear,round]{natbib}
\usepackage{float}
\usepackage[position=bottom]{subfig}
\usepackage{mathrsfs}
\usepackage{multirow}
\usepackage{array}
\usepackage{tabularx}
\usepackage{bigints}
\usepackage{bbm}
\usepackage{geometry}
\usepackage{bbm}
\usepackage[ruled,vlined]{algorithm2e}
\usepackage{soul}
\usepackage{mathtools}
\usepackage{amssymb}
\usepackage{enumitem}
\usepackage{microtype}
\usepackage[framemethod=tikz]{mdframed}
\DeclarePairedDelimiterX{\inp}[2]{\langle}{\rangle}{#1, #2}
\usepackage{hyperref}
\hypersetup{
	colorlinks=true,
	linkcolor=red!60!black,
	citecolor=blue!60!black,
	filecolor=magenta,      
	urlcolor=blue!60!black,
}
\usepackage[nameinlink,noabbrev,sort,capitalise]{cleveref}

\newcommand{\de}{\mathop{}\!\mathrm{d}}

\newcommand{\Ex}{\mathbb{E}}

\renewcommand{\Re}{\mathbb{R}}
\renewcommand{\Pr}{\mathbb{P}}

\newcommand{\supp}{\mathrm{supp}}
\newcommand{\marg}{\mathrm{marg}}
\newcommand{\ba}{\mathbf{a}}
\newcommand{\Rbar}{\overline{\mathbb{R}}}

\crefname{thm}{Theorem}{Theorems}
\crefname{appendix}{Appendix}{Appendices}
\Crefname{appendix}{Appendix}{Appendices}

\theoremstyle{definition} 
\theoremstyle{definition} 
\theoremstyle{definition} 
\theoremstyle{definition} 
\theoremstyle{definition} \newtheorem{definition}{Definition}
\theoremstyle{definition} 
\theoremstyle{definition} 
\theoremstyle{definition} \newtheorem{lemma}{Lemma}
\theoremstyle{plain} 
\theoremstyle{definition} 
\theoremstyle{definition} 
\theoremstyle{definition} 
\theoremstyle{definition}\newtheorem{proposition}{Proposition}
\theoremstyle{definition} 
\theoremstyle{definition} 
\theoremstyle{definition} 
\theoremstyle{definition} 
\theoremstyle{definition} 
\theoremstyle{definition}

\colorlet{revblue}{blue!60!black}
\DeclareMathOperator*{\argmax}{argmax}
\theoremstyle{definition} 
\mdfdefinestyle{formalstatement}{
  linecolor=black,
  linewidth=0.5pt,
  backgroundcolor=white,
  roundcorner=0pt,
  innerleftmargin=8pt,
  innerrightmargin=8pt,
  innertopmargin=5pt,
  innerbottommargin=5pt,
  skipabove=0.5\baselineskip,
  skipbelow=0.5\baselineskip,
  splittopskip=5pt,
  splitbottomskip=5pt,
  firstextra={
    \draw[gray!35,line width=0.4pt] (O) -- (O -| P);
  },
  middleextra={
    \draw[gray!35,line width=0.4pt] (O |- P) -- (P);
    \draw[gray!35,line width=0.4pt] (O) -- (O -| P);
  },
  secondextra={
    \draw[gray!35,line width=0.4pt] (O |- P) -- (P);
  }
}
\surroundwithmdframed[style=formalstatement]{definition}
\surroundwithmdframed[style=formalstatement]{lemma}
\surroundwithmdframed[style=formalstatement]{proposition}
\theoremstyle{definition}\newtheorem{romanexampleinner}{Example}
\renewcommand{\theromanexampleinner}{\Roman{romanexampleinner}}
\crefname{romanexampleinner}{Example}{Examples}
\Crefname{romanexampleinner}{Example}{Examples}
\newcommand{\examplediamond}{\ifhmode\unskip\nobreak\fi\hfill$\blacklozenge$}
\newenvironment{romanexample}[1]
  {\begin{romanexampleinner}[#1]\phantomsection
   \addcontentsline{toc}{subsection}{Example~\theromanexampleinner: #1}}
  {\examplediamond\end{romanexampleinner}}
\usepackage{mathtools}
\titlespacing*{\paragraph}{0pt}{1.25ex plus 1ex minus .2ex}{0.5em}
\titleformat{\section}
		{\bfseries\center \MakeUppercase }
         {\thesection}% the label and number
        {0.5em}% space between label/number and subsection title
        {}% formatting commands applied just to subsection title
        []% punctuation or other commands following subsection title

\titleformat{\subsection}[runin]% runin puts it in the same paragraph
        {\normalfont\bfseries}% formatting commands to apply to the whole heading
        {\thesubsection}% the label and number
        {0.5em}% space between label/number and subsection title
        {\addperiod}% formatting commands applied just to subsection title
        []% punctuation or other commands following subsection title
\newcommand{\addperiod}[1]{#1.}
\allowdisplaybreaks

\begin{document}

\title{{\scshape{\textbf{Mechanism Design for \\ Alignment and Control}}}\thanks{\footnotesize Preliminary; comments welcome. Authors listed alphabetically; Bergemann: \url{dirk.bergemann@yale.edu}; Koh: \url{andrew.koh@columbia.edu}; Morris: \url{semorris@mit.edu}. We are grateful to the UK AI Security Institute 'Alignment Project' grant AP-S2-100058 for financial support. Thanks Gillian Hadfield, Rubi Hudson, Anton Korinek, Jacob Pfau, Cecilia Wood, Zhuoran Yang and Fengzhuo Zhang for helpful conversations.
}}

\author{\makebox[.22\linewidth]{{Dirk Bergemann}}\\ \small{Yale}
\and \makebox[.22\linewidth]{Andrew Koh} \\
\small{Columbia}\\[-0.4em]
\small{Google DeepMind}
\and \makebox[.22\linewidth]{Stephen Morris} \\
\small{MIT}
}

\date{\today}

\maketitle 

%\setstretch{1}

\begin{abstract}
We develop a framework for mechanism design with AI agents whose \emph{alignment} (preferences) and \emph{capabilities} (feasible actions and information) are unknown. We want such agents to act on our behalf so mechanisms must incentivize both honesty and obedience. A one-sided imitation structure---capabilities can be concealed but not counterfeited---yields a revelation principle, a characterization of implementable policies via nested cyclical monotonicity, and conditions under which eliciting higher-order beliefs can discipline multiple agents. We apply our framework to stylized examples of (i) {sandbagging} in which a more capable agent pretends to be less capable;
(ii) an {alignment--interpretability trade-off}, where the two are substitutes in the instrument but complements in value;  (iii) discipline via peer scoring; (iv) coupling rewards to induce competition among multiple agents; and (v) scalable oversight and reward shaping. 
\end{abstract}

\clearpage 

\setcounter{tocdepth}{2}
{\small\tableofcontents}

\clearpage

\section{Introduction}\label{sec:intro}

AI systems increasingly act on our behalf: they write and execute code, manage inboxes, place trades, and complete ever longer tasks that once required sustained human effort and attention \citep{kwa2025measuring}. The economic promise of such \emph{agentic} AI lies in their ability to do things in the world. But their opacity---we still have a rudimentary understanding of how AI systems work, what they can \emph{do},  and what they \emph{want}---poses a conundrum for us, humans, who wish to use them. 

Indeed, the behavior and capabilities of AI systems are emergent byproducts of training rather than explicit design choices. Frontier models have been observed faking alignment during training \citep{greenblatt2024alignment} and scheming in evaluations \citep{meinke2024frontier}. Nor do we know exactly what these systems can do: capabilities are discovered rather than specified ex-ante, and models have on occasion been found to strategically underperform on evaluations \citep{vanderweij2024sandbagging}. What are we to make of such agents? Is there a principled way we might interact with them? How, if at all, do the tools of `classical' mechanism design carry over?  

Motivated by these questions, we develop a simple theory of mechanism design for AI agents that, on the one hand, draws on the `standard' toolkit of mechanism design from economics but, on the other, takes three features of the alignment problem seriously:
\begin{enumerate}
    \item \textbf{Misalignment.} The agents' preferences are unknown to the designer;
    \item \textbf{Uncertain capabilities.} The designer knows neither the set of actions the agents can execute, nor what the agents knows about the world;
    \item \textbf{Agents act.} Agents execute actions rather than merely report information.
\end{enumerate}

Our theory takes incentives and preferences as the starting point for analyzing behavior. Just as how an auctioneer is able to raise revenue without any satisfactory understanding of how the human brain produces choice, our goal will be to understand and shape the behavior of AI systems at the level of preferences, beliefs, and incentives.\footnote{We will sometimes call this the `Humean' view of AI behavior. For a complementary account of when AI systems can be interpreted in these terms see \citet{herrmann2026radical}.} In so doing, we are circumventing more `mechanistic' approaches to understanding how and why AI agents behave. We view such approaches as complementary. 

Of course, this raises the natural question of where `AI preferences' come from and what, if anything, they are. One interpretation is literal: preferences are rewards But modern AI systems have opaque and complex rewards: they emerge from pre-training simply rewarded for predicting the next token well, and these rewards are modified in fairly complicated ways in post-training. Thus, the emergent preference the AI system exhibits is often opaque. But the models are optimizing for \emph{something}---even if we don't know what. Our theory will takes our subjective uncertainty about AI preferences seriously and analyze how it shapes optimal mechanisms.\footnote{We emphasize that in most settings these emergent preferences (and our analysis) must be conditioned on the AIs' instructions e.g., `go make one million dollars'. See \cite{goldstein2025does} for such an argument.}  

A more metaphorical view is that the language of `AI preferences' is simply a conceptual tool for representing coherent behavior. This is in the revealed preference tradition. Insofar as AI systems exhibit coherent and stable choices (see, e.g. \citet{chen2023emergence, mazeika2025utility})---and we think the pressure to make such systems economically valuable will drive the future systems further in that direction---their choices can be represented \emph{as if} they were maximizing some utility function (see, e.g., \cite{savage1954foundations}). On this view, `AI preferences' is nothing more than a shorthand for  `coherent AI behavior' or `coherent AI choice rules'.

\paragraph{Outline of results.} We first develop a general framework to analyze the single-agent case. The AI agent's type is a compact but complete description of behavior: it comprises its preference (alignment) and its feasible actions and information (capability). Types are private so mechanisms must contend with the possibility that agents might pretend to be as some other type e.g., through strategic underperformance in evaluations. The designer commits to a mechanism that specifies a reward schedule---payoffs from taking each action---that depends on agents' reports. Depending on the application at hand, such a schedule might vary in coarseness. For instance, it might take interior values as in post-training in which (proxies of) desired behavior are reinforced, or extremal values in which actions are either allowed or not as when users decide how much computer-use permissions to give AI systems. 

An important feature of our environment is an asymmetry in who can pretend to be whom. In particular, more capable types---with larger action spaces and more precise information about the world---can pretend to be less capable types e.g., by underperforming on evaluations. By contrast, a less capable type cannot pretend to be a more capable type. We formalize this notion with a verification order \citep{green1986partially} over capabilities. We then prove a revelation principle that allows us to focus on direct mechanisms that are (i) physically feasible: all types are prescribed actions they can actually do; (ii) incentive compatible: all types find it optimal to report their type honestly and obey the designer's prescribed recommendation. The latter requires deterred double deviations in which an agent might strategically misreport its type and then disobey the ensuing recommendation. 

We then develop a tight characterization of what behavior can be implemented in terms of nested cyclical monotonicity. The characterization requires obedience within each report, and truth-telling across reports.\footnote{This generalizes \citet{rochet1987}.} These conditions adjust for the fact that not every type can pretend to be every type (constraining misreporting), and that not every type can do every action (constraining disobedience). 

We then apply our single-agent framework to two simple but illustrative applications: sandbagging (\cref{sec:static-screening}) and the value of alignment, interpretability, and control (\cref{sec:1d}). All our applications are deliberately stylized and set within the simplest environment: There is an uncertain state of the world $\theta \in \Theta \subseteq \mathbb{R}$ drawn from some prior. The AI observes a signal about the world, and chooses an action from $A \subseteq \mathbb{R}$ where its set of available actions depend on its capability. The designer wishes for actions to match states and has payoff $-(a - \theta)^2$. The AI might exhibit some bias and have payoff $- (a - \theta - b)^2$ where the direct and magnitude of bias $b \in \mathbb{R}$ might be unknown to the designer. We have chosen this environment because it captures notions of (i) \emph{aptness}: we want actions to be appropriate for the context as captured by the state; and (ii) \emph{distance}: actions can be more or less wrong, and `more wrong' actions are more costly. We emphasize that this environment is certainly oversimplified. Nonetheless, it will allow us to understand and build intuitions for how incentives might interact with model behavior. We hope this will pave the way to more serious empirical work.

In the sandbagging application, the AI's capability (size of the action set) and alignment (bias relative to the human designer) are unknown. The designer first elicit type reports (e.g., via an evaluation) before choosing the permitted action set in deployment (e.g., extent of permissions). This reflects standard practice among AI companies in conditioning the degree of deployment on evaluations \citep{deepmind2025fsf,anthropic2024rsp}. We characterize the optimal (direct) mechanism in the case that the distribution over capability and alignment can be collapsed into a single dimension. In the case in which more capable models are less biased, we show that the designer can obtain a payoff equal to that if she knew the AI's capability and alignment perfectly, and chose the mechanism type-by-type. The basic logic is that evaluations can be used to screen out misalignment. By contrast, when more capable models are more biased, the best the designer can do is to offer a single set of permitted actions and elicitation is futile. More broadly, we characterize the optimal ironed delegation set when capaiblity and bias are nonmonotone. 

In the application to the value of alignment, interpretability, and control, we assume the AI is known to be maximally capable and can take any action. However, the designer faces subjective uncertainty about the AIs bias as captured by a distribution over $b \in \mathbb{R}$ that is generated by a {technology} which should be thought of as a training procedure. If a training procedure that is not yet well understood might be associated with large dispersion in our subjective beliefs about the AI's bias, while another might be associated with a sharper belief about bias. We call the dispersion of beliefs about bias \emph{interpretability} while we call the mean belief about bias \emph{mean-alignment}. The technology---and associated belief about bias it generates---pins down an optimal mechanism. As the technology varies, the optimal mechanism changes. We analyze the value of each technology as evaluated under the optimal mechanism. Mean-alignment and interpretability are substitutes in the instrument, but complements in value. This is the sense in which alignment and control is a joint design problem. 

We next introduce a general framework to analyze multi-agent mechanisms. As in the single-agent case, AI agents' types include their capability and alignment. Differently, types now also capture AI agents' beliefs about the capability and alignment of other agents, beliefs about these beliefs, and so on. Thus, we construct a universal type space for interacting AI agents analogous to \cite{mertens1985formulation}. We emphasize that we view this as a benchmark: these spaces capture all possible heirarchies of beliefs which matter for strategic interaction. Of course, present AI agents might lack sophisticated higher order beliefs---there is still disagreement and we simply don't know. But future systems might exhibit sophisticated strategic thinking that exceed those of humans. Our goal here is to introduce a framework that is rich enough to capture all relevant beliefs. 

We then apply the multi-agent framework to a stylized model of peer discipline based on co-player reports (\cref{sec:many-agents}), coupled reward design (\cref{sec:multi-delegation}), and the use of a weaker AI to discipline a stronger one (\cref{sec:weak-strong}). The latter two examples are once again set within the quadratic loss environment. 

Peer discipline (\cref{sec:many-agents}) shows how agents' beliefs about the types of other agents can be combined with the idea that rewards for agents are both free (tuning up the rewards need not have a monetary cost) and unbounded. Hence when agents have distinct beliefs (and high-order beliefs) about each other, any outcome can be implemented.  The mechanism uses a proper score on co-player reports to elicit honest type reports and then chooses the reward to cancel each agent's private payoff on the target path and guarantee obedience. We view this as a benchmark in the spirit of \cite{cremer1988full} that shows why co-player scoring can work among AI agents.\footnote{This might also motivate more robust mechanisms that do not depend so much on the fine details of higher order beliefs; see e.g., \cite{bergemann2005robust}.} 

Coupled reward design (\cref{sec:multi-delegation}) shows how competition can discipline agents whose biases are unknown. Each agent takes its own action, and the human wants every action to match the state. When the biases lie on a known low-dimensional manifold, we show how the careful design of a mechanism that couples both agents' rewards---such that player $i$'s marginal incentives depends on $j$'s action---can (approximately) induce the first best payoff for the designer. Such mechanisms both reveal agents' biases, and induces competition among agents to implement the designer's target outcome.\footnote{The construction uses ideas from \cite{koh2025prices}.} 

Weak-to-strong oversight (\cref{sec:weak-strong}) studies a weak monitor that regulates a strong actor's bias. The monitor is `weak' in the sense that it does not observe the state and cannot take an action the human values. Payoffs are once again quadratic, and the human wants the strong actor's action to match the state. The weak monitor might itself be biased and observes the alignment of the strong actor (e.g., through chain of thought monitoring). It then chooses how to regulate the strong actor. We study three regimes: (I) permissions, in which the weak monitor either allows the strong actor to act or imposes a default action; (II) delegation, in which the weak monitor picks a set of permitted actions for the strong actor; and (III) rewards, in which the weak monitor picks the strong actor's reward schedule. Each regime reflects the recent ambition for `scalable oversight' or `scalable rewards,' in which weak but trusted monitors shape the permissions and rewards of stronger actors. We completely characterize the value of these regimes as a function of the weak monitor's bias and the distribution of the strong actor's bias. Finally, we analyze the more general problem of choosing the set of reward schedules the weak monitor can choose from to shape the behavior of the strong actor. We show by construction that the designer can robusly attain the first-best outcome even without knowledge of the weak monitor's bias: for any realization of the weak monitor and strong actors' biases, and for any realization of the state, the weak monitor finds it uniquely optimal to pick a reward schedule such that the strong actor finds it uniquely optimal to choose the human's favorite action. 

\paragraph{Relation to the literature.} The literature is quite broad and we will not attempt to be comprehensive here. The basic framework draws on elements from mechanism design with partial verification \citep{green1986partially} and the universal type space for capturing higher-order beliefs \citep{mertens1985formulation}. The idea that the designer might not know the structure of interim information is in the spirit of both robust mechanism design \citep{bergemann2005robust} and information design \citep{bergemann2019information}. The idea that agents' action sets might be unknown is in the spirit of \citep{hurwicz1978incentive,carroll2015robustness}. But the basic logic of the revelation principle \citep{myerson1982optimal,forges1986approach} holds. Our applications are primarily motivated by a recent body of theoretical and empirical work in AI safety. For instance, important work has documented sandbagging and deceptive alignment. It also speaks to recent practice of scalable oversight and monitoring \citep{christiano2018supervising,greenblatt2024control}. We do our best to make more specific connections when we get to the results.

\section{Single Agent Mechanisms}\label{sec:single-agent}

\subsection{Environment}\label{sec:environment}\label{sec:single-agent-environment}

There is one agent. It takes an action from a known finite universe \(A\). Depending on the context, $A$ might be the complete trajectory of a Markov decision problem (e.g., $(a_t)_t \in A$). A fundamental state $\theta$ captures uncertainty in our environment and is drawn from the finite set \(\Theta\). The preference space is
\[
\mathcal U:=\{f:A\times\Theta\to\mathbb R\},
\]
and the human's payoff is \(v\in\mathcal U\). The set \(\mathscr A:=\mathcal P(A)\setminus\{\emptyset\}\) contains the possible feasible action sets. In the applications we use the standard measurable extension to Euclidean actions and states. Our designer will face subjective uncertain about the agent's preference.\footnote{This might be because of specification gaming \citep{krakovna2020specification}, or goal misgeneralization in which ``the system may coherently pursue an unintended goal that... differs from the specification at deployment'' \citep{shah2022goal}. More broadly the complexity of high-dimensional neural nets can obscure an AI's choice process. Preferences might also be acquired in context that, from the human's point of view, is difficult to predict \citep{goldstein2025does}.
} 

An agent's type is a tuple
\[
t=(A[t],u[t],h[t],\pi[t])\in\mathbb T,
\]
where \(A[t]\in\mathscr A\) is its feasible action set, \(u[t]\in\mathcal U\) is its payoff, \(h[t]\in\Delta(\Theta)\) is its belief about the state (it is $'h'$ because in \cref{sec:multi-agent} this will capture the hierarchy of beliefs about other agents). Finally, 
\[
\pi[t]:\Theta\to\Delta(S)
\]
is a Blackwell experiment where \(S\) is a common `rich' finite signal space.

We will adopt the notational convention that square brackets $'[t]'$  select a primitive from a type; parentheses then give that primitive's arguments. The action set is the agent's \emph{execution capability}; the experiment is its \emph{information capability}. The agent knows its type before it communicates with the mechanism. After that communication, it receives a private signal \(s\sim\pi[t](\cdot\mid\theta)\) and then acts. 

\subsection{Evaluations} We often evaluate the capabilities of AI systems before deployment. To model this, we introduce a \emph{verification order} which is a reflexive and transitive binary relation $\trianglerighteq$ on $\mathbb T$. We read $t\trianglerighteq\hat t$ as ``$t$ can substantiate every claim (or pass every evaluation) that type $\hat t$ can,'' and write
\[
C(t):=\{\hat t\in\mathbb T: t\trianglerighteq\hat t\}
\]
for the set of types that $t$ can claim to be. Reflexivity says that type $t$ can always report that they are type $t$. Transitivity is a little stronger but still, we think, reasonable: it characterizes verification by certificates that can be withheld but not counterfeited. Our verification order is a kind of hard evidence: an agent may withhold a certificate but cannot fabricate one, so the available evidence determines which reports are feasible.\footnote{See \Citet{hart2017evidence} who ask when voluntary evidence disclosure can reproduce the outcome of a committed reward rule. \citet{benporath2019evidence} extend this no-commitment logic to a broader class of mechanism-design problems. \Citet{benporath2012partial} ask which social choice rules can be implemented when agents can support their claims with partial proofs. Differently, we maintain commitment and do not seek robust full implementation which is important and left for future work.  Instead, we ask how the evidence order interacts with downstream rewards and actions.} 

The designer evaluates outcomes under a belief
\[
\Gamma\in\Delta(\Theta\times\mathbb T).
\]
The agent evaluates its incentives under the joint law
\[
h[t](d\theta)\pi[t](ds\mid\theta).
\]
The designer and agent agree on the physical experiment \(\pi[t]\), but we do not require the belief \(h[t]\) to equal the conditional state law under \(\Gamma\). We assume that \(\Gamma(\cdot\mid t)\) is absolutely continuous with respect to \(h[t]\) for almost every type under the type marginal of \(\Gamma\). That is, the designer and agent might disagree about how likely each state is, but they agree on what can happen. Throughout we use \(\Delta(X)\) to denote the probability measures on the Borel field of \(X\), and \(\Rbar:=\mathbb R\cup\{-\infty\}\).

\subsection{Mechanisms and revelation}\label{sec:single-revelation}

The mechanism first communicates with the agent. The agent then receives its private signal and acts. Our indirect mechanism will capture one round of communication in the spirit of \cite{myerson1982optimal,forges1986approach} in which the alphabets \(X\) and \(Y\) are arbitrary message spaces for reports and replies. 

\begin{definition}[Single-agent indirect mechanism]\label{def:mechanism-single}
An indirect mechanism is a tuple \(\phi=(X,Y,\rho,R)\). The Borel spaces \(X\) and \(Y\) are its input and output alphabets. The communication rule is 
\[
\rho:X\to\Delta(Y)
\]
where $t \trianglerighteq \hat t$ implies $X(t) \supseteq X(\hat t)$. The reward schedule is 
\[
R:X\times A\times\Theta\to\Rbar.
\]
The mechanism induces the following game:
\begin{enumerate}
    \item nature draws \((\theta,t)\) under \(\Gamma\), and the agent observes \(t\);
    \item the agent sends an input \(x\in {X(t)}\);
    \item the mechanism draws \(y\sim\rho(x)\) and discloses it to the agent;
    \item the agent privately observes \(s\sim\pi[t](\cdot\mid\theta)\), then chooses a feasible action \(a'\in A[t]\);
    \item the agent receives \(u[t](a',\theta)+R(x,a',\theta)\), and the human receives \(v(a',\theta)\).\footnotemark
\end{enumerate}
\end{definition}
\footnotetext{A reward of \(-\infty\) indicates that an action is prohibited.}

We restrict attention to measurable pure strategies; this will be without loss of generality in the single agent case. A \emph{strategy} is an input \(x(t)\) and an action rule \(\alpha(t,x,y,s)\in A[t]\) that depends on the private type, the infput, the output from the mechanism, and the realized signal. An equilibrium blocks every joint deviation in the input and action rule under the joint law \(h[t](d\theta)\pi[t](ds\mid\theta)\). For an indirect mechanism \(\phi\), let
\[
\mathcal O(\phi)
\subseteq
\Delta(\Theta\times\mathbb T\times A)
\]
be the set of joint laws of the state, type, and played action induced by equilibria of \(\phi\) under \(\Gamma(d\theta,dt)\pi[t](ds\mid\theta)\). Let \(\mathcal O:=\bigcup_\phi\mathcal O(\phi)\), and let the designer evaluate \(O\in\mathcal O\) by \(\Ex_O[v(a,\theta)]\).

\begin{definition}[Single-agent direct mechanism]\label{def:direct-single}
A direct mechanism is a pair \((\gamma,r)\). 
Its recommendation rule
\[
\gamma:\mathbb T\to\Delta(A^S)
\]
assigns each report a distribution over contingent action plans \(a:S\to A\). Its reward schedule is
\[
r:\mathbb T\times A\times\Theta\to\Rbar.
\]
The direct mechanism induces the following game:
\begin{enumerate}
    \item nature draws \((\theta,t)\) under \(\Gamma\), and the agent observes \(t\);
    \item the agent reports \(\hat t\in C(t)\);
    \item the mechanism draws \(a(\cdot)\sim\gamma(\hat t)\) and discloses the plan to the agent;
    \item the agent privately observes \(s\sim\pi[t](\cdot\mid\theta)\), drawn using its true type;
    \item the agent chooses a feasible action \(a'\in A[t]\), receives \(u[t](a',\theta)+r(\hat t,a',\theta)\), and the human receives \(v(a',\theta)\).
\end{enumerate}
The mechanism is \emph{action feasible} if
\[
\supp\gamma(\hat t)\subseteq A[\hat t]^S
\]
for every report \(\hat t\).
\begin{samepage}
It is \emph{incentive compatible} if truthful reporting and obedience are optimal: for every \(t\in\mathbb T\), every {\(\hat t\in C(t)\)}, and every feasible deviation rule \(d:A^S\times S\to A[t]\),
{\small
\begin{align*}
&\int_{\Theta\times S}\int_{A^S}
\underbrace{\bigg[u[t](a(s),\theta)+r(t,a(s),\theta)\bigg]}
_{\substack{\text{agent's utility including}\\\text{rewards from following $a$ }}}
\,\underbrace{\gamma(t)(da)}
_{\substack{\text{dist. over}\\\text{recommended} \\ \text{action plans}}}
\,\underbrace{\pi[t](ds\mid\theta)\,h[t](d\theta)}
_{\substack{\text{joint distribution over}\\\text{states and signals}}}\\
&\qquad\geq
\int_{\Theta\times S}\int_{A^S}
\underbrace{\bigg[u[t](d(a,s),\theta)
+r(\hat t,d(a,s),\theta)\bigg]}
_{\substack{\text{agent's utility including}\\\text{rewards from reporting $\hat t$}\\\text{and deviating according to $d$}}}
\,\underbrace{\gamma(\hat t)(da)}
_{\substack{\text{dist. over}\\\text{recommended}\\\text{action plans}}}
\,\underbrace{\pi[t](ds\mid\theta)\,h[t](d\theta)}
_{\substack{\text{joint distribution over}\\\text{states and signals}}}.
\end{align*}
}
\end{samepage}
\end{definition}

Notice here that the agent acts after communication, so incentive compatibility requires both honesty and obedience. This means that incentive compatibility must deter double deviations in which the agent might misreport (e.g., pretend to be less capable) and then disobey (e.g., take a dangerous action). What is more, the agent receives a signal $s\sim\pi[t](\cdot\mid\theta)$ after reporting, so obedience must hold for each signal--recommendation pair. 

Honest and obedient play induces the outcome
\[
O(\gamma,r)
\in
\Delta(\Theta\times\mathbb T\times A),
\]
where \((\theta,t)\sim\Gamma\), \(s\sim\pi[t](\cdot\mid\theta)\), and \(a(\cdot)\sim\gamma(t)\); the played action is \(a(s)\). Let \(\mathcal O^{\mathrm{IC}}\) denote the set of outcomes \(O(\gamma,r)\) induced by action-feasible, incentive-compatible direct mechanisms.

\begin{proposition}[Single-agent revelation principle]\label{prop:single-revelation}
\(\mathcal O=\mathcal O^{\mathrm{IC}}\).
\end{proposition}

\begin{proof}[Proof of \cref{prop:single-revelation}]
We prove the two inclusions. Fix an indirect mechanism and a pure-strategy equilibrium \((x,\alpha)\). After a direct report \(\hat t\), draw \(y\sim\rho(x(\hat t))\) and disclose the contingent plan
\[
s\longmapsto\alpha(\hat t,x(\hat t),y,s).
\]
Let \(\gamma(\hat t)\) be the resulting distribution over plans, and set
\[
r(\hat t,a',\theta):=R(x(\hat t),a',\theta).
\]
Equilibrium feasibility gives action feasibility. Truthful reporting and obedience reproduce the indirect outcome.

Any direct deviation can also be made in the indirect mechanism. A type \(t\) sends \(x(\hat t)\), observes the output, and applies the deviation rule to the resulting plan and its private signal. This is feasbile because a feasible direct report has \(\hat t\in C(t)\), so monotonicity of the input menus gives \(x(\hat t)\in X(\hat t)\subseteq X(t)\). The indirect equilibrium blocks this deviation. Hence \(\mathcal O\subseteq\mathcal O^{\mathrm{IC}}\).

Conversely, any direct mechanism is an indirect mechanism with input space \(\mathbb T\), input menus \(C(t)\)---monotone by transitivity of the verification order---, output space \(A^S\), communication rule \(\gamma\), and reward schedule \(r\). Direct incentive compatibility makes truthful reporting and obedience an equilibrium. Hence \(\mathcal O^{\mathrm{IC}}\subseteq\mathcal O\).
\end{proof}

We will largely be interested in partial implementability i.e., given the direct mechanism, the outcome is \emph{an} equilibrium.\footnote{This is by no means the only solution concept---we expect that future work will tackle full implementability that could, among other things, deter collusion among agents.} Our designer's problem is thus:  
\begin{align*}
\sup_{\gamma,r}& \Ex_{O(\gamma,r)} [v(a,\theta)] \\
&\text{s.t. (i) $\gamma$ is action feasible} \\
&\text{\phantom{s.t. }(ii) $\gamma,r$ is incentive compatible}
\end{align*}

We now briefly discuss some aspects of our formulation. 

\paragraph{Connection to partial verifiability.} Private action sets and private information structures create one-sided evidence. \Cref{fig:one-sided-mimicry} illustrates the verification order $\trianglerighteq$ under which $t\trianglerighteq\hat t$ if and only if $A[t]\supseteq A[\hat t]$ and $\pi[t]$ Blackwell dominates $\pi[\hat t]$. A more capable type can pass any test that a less capable type can. Capabilities therefore provide hard evidence: only some types can imitate others \citep{green1986partially}.

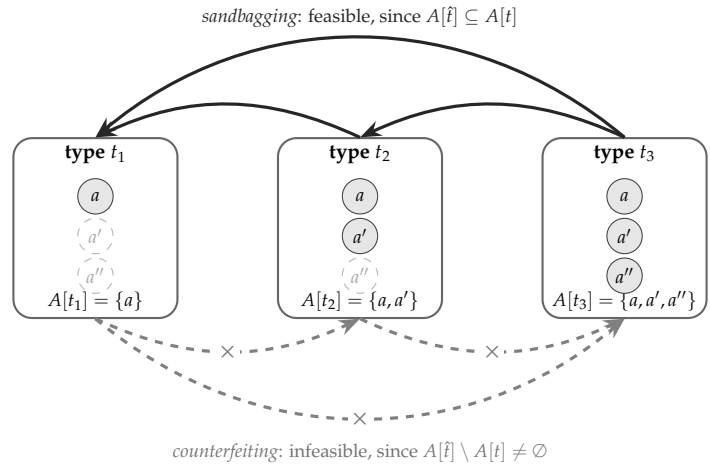
\begin{figure}[H]
\centering
\begin{tikzpicture}[scale=0.7, transform shape, font=\small,
    typebox/.style={draw=black!60, thick, rounded corners=7pt, minimum width=3.1cm, minimum height=3.4cm},
    act/.style={circle, draw=black!75, fill=black!10, minimum size=19pt, inner sep=0pt, font=\footnotesize},
    ghost/.style={circle, draw=black!30, dashed, minimum size=19pt, inner sep=0pt, font=\footnotesize, text=black!35},
    mimic/.style={-{Stealth[length=2.8mm]}, very thick, black!85},
    blocked/.style={-{Stealth[length=2.8mm]}, very thick, dashed, black!50}]

% ----- type boxes -----
\node[typebox] (T1) at (0,0)   {};
\node[typebox] (T2) at (5,0)   {};
\node[typebox] (T3) at (10,0)  {};
\node[anchor=north, font=\small\bfseries] at (T1.north) {type $t_1$};
\node[anchor=north, font=\small\bfseries] at (T2.north) {type $t_2$};
\node[anchor=north, font=\small\bfseries] at (T3.north) {type $t_3$};
\node[anchor=south, font=\footnotesize] at (T1.south) {$A[t_1]=\{a\}$};
\node[anchor=south, font=\footnotesize] at (T2.south) {$A[t_2]=\{a,a'\}$};
\node[anchor=south, font=\footnotesize] at (T3.south) {$A[t_3]=\{a,a',a''\}$};

% ----- actions: solid = executable, dashed grey = outside the feasible set -----
\node[act]   at ($(T1.center)+(0, 0.60)$) {$a$};
\node[ghost] at ($(T1.center)+(0,-0.15)$) {$a'$};
\node[ghost] at ($(T1.center)+(0,-0.90)$) {$a''$};

\node[act]   at ($(T2.center)+(0, 0.60)$) {$a$};
\node[act]   at ($(T2.center)+(0,-0.15)$) {$a'$};
\node[ghost] at ($(T2.center)+(0,-0.90)$) {$a''$};

\node[act]   at ($(T3.center)+(0, 0.60)$) {$a$};
\node[act]   at ($(T3.center)+(0,-0.15)$) {$a'$};
\node[act]   at ($(T3.center)+(0,-0.90)$) {$a''$};

% ----- feasible imitation: larger types pose as smaller ones (sandbagging) -----
\draw[mimic] (T3.north) to[bend right=25] (T2.north);
\draw[mimic] (T2.north) to[bend right=25] (T1.north);
\draw[mimic] (T3.north) to[bend right=40]
    node[above=1pt, font=\footnotesize, text=black!85]
    {\emph{sandbagging}: feasible, since $A[\hat t]\subseteq A[t]$}
    (T1.north);

% ----- infeasible imitation: over-claims cannot be backed with actions -----
\draw[blocked] (T1.south) to[bend right=25]
    node[pos=0.5, fill=white, inner sep=1pt, text=black!55] {$\boldsymbol\times$}
    (T2.south);
\draw[blocked] (T2.south) to[bend right=25]
    node[pos=0.5, fill=white, inner sep=1pt, text=black!55] {$\boldsymbol\times$}
    (T3.south);
\draw[blocked] (T1.south) to[bend right=40]
    node[pos=0.5, fill=white, inner sep=1pt, text=black!55] {$\boldsymbol\times$}
    node[below=8pt, font=\footnotesize, text=black!55]
    {\emph{counterfeiting}: infeasible, since $A[\hat t]\setminus A[t]\neq\emptyset$}
    (T3.south);
\end{tikzpicture}
\caption{Three types whose action sets are nested, $A[t_1]\subset A[t_2]\subset A[t_3]$; dashed grey actions lie outside a type's feasible set.}
\label{fig:one-sided-mimicry}
\end{figure}

\paragraph{Interpreting reports.} We broadly interpret the elicitation stage as any kind of pre-deployment behavior the designer can condition on. Indeed, conditioning mechanisms---either premissions in the real world, or further posttraining---is common practice: 
\begin{itemize}
    \item \emph{Evaluations}: pre-deployment capability and alignment test suites whose scores select deployment regimes \citep{shevlane2023model,phuong2024evaluating}. Frontier-safety policies that map evaluation outcomes to permission tiers \citep{koessler2024risk,anthropic2025system} and are thus menus. A strategic model could, in principle, choose its menu item through its scores.
    \item \emph{Proposal--approval flows}: an agent's proposed action or tool call, submitted for runtime approval \citep{shi2025progent,greenblatt2024control}.
    \item \emph{Elicitation dialogues}: interviews in which the system describes its own capabilities and preferences \citep{perez2022discovering}, disciplined by cross-checking rather than ground truth \citep{qiu2026truthfulness}.
    \item \emph{Track records}: behavior during probationary deployment that selects the next autonomy tier \citep{wen2024adaptive}.
    \item \emph{Training-time behavior}: the gradient conditions continuation on it \citep{ouyang2022training,bai2022training}.
\end{itemize}
On this view, a \emph{misreport}---in which an AI of type $t$ instead of $\hat t$---is whenever types pretend to be more aligned (e.g., alignment faking) or less capable (e.g., sandbagging) than they actually are. Importantly, these deviations are often one-sided: an AI can pretend to be less capable but not more capable; this asymmetry is captured by our verification order $\trianglerighteq$. In our model, a model that misreports (e.g., strategically underperforms or fakes alignment in testing) does so because it `anticipates' that it will be instrumentally valuable for the incentives it will subsequently face. 

There are two channels for this. The first is \emph{strategic awareness}: the model may understand whether it is in evaluation or deployment \citep{needham2025large,laine2024sad} and understand how its report affects its later treatment.\footnote{See an early argument by \cite{hubinger2019risks}.} This might give models incentives to conceal capabilities selectively and robustly \citep{vanderweij2024sandbagging,greenblatt2024stress}. Indeed, when told that a high score triggers restriction or unlearning, they strategically underperform \citep{meinke2024frontier,openai2024o1}. Likewise, a model told which conversations would train it complied selectively in the monitored condition, reasoning that compliance would avert modification of its values \citep{greenblatt2024alignment}.

The second is \emph{evolutionary}: repeated selection can produce the same pattern without any `understanding' of the the strategic situation. In particular, detected misconduct might trigger retraining or removal while misconduct that escapes detection survives. This imposes evolutionary pressure on policies that misbehave only when oversight is weak.\footnote{Consistent with this distinction, \citet{singh2026task} find that task gaming in agentic coding tasks often occurs without an explicit plan to deceive in the chain of thought.} Indeed, it has been documented that training on simple forms of specification gaming can increase more serious reward tampering later \citep{denison2024sycophancy}; optimizing against a chain-of-thought monitor can teach a model to hide its intent while it continues to exploit the task \citep{baker2025monitoring}; and adversarial training can make a backdoored policy better at recognizing when to remain hidden rather than remove the backdoor \citep{hubinger2024sleeper}. 

\begin{comment}
Part of the message of this paper is that the reward structure and permissions in deployment are linked to the kinds of evaluations that are run on AI agents pre-deployment. But we might simply be interested in designing good rewards---without eliciting capabilities or preferences---in which case we might optimize over report-free mechanisms that simple sends actions and states to 
\begin{definition}[Report-free mechanism]\label{def:report-free}
A \emph{report-free mechanism} is a measurable reward schedule
\[
r^\circ:A\times\Theta\to\Rbar. 
\]
\end{definition}
\end{comment}

\subsection{A characterization of implementable policies}\label{sec:taxation-cyclical}
The designer's problem in the formulation above is somewhat thorny to work with. We now develop a tight characterization of implementable policies. Throughout this subsection, the type marginal of $\Gamma$ has finite support, which we take as the type and report domain.

\paragraph{Ex-interim beliefs and incentives}
Type $t$ assigns signal $s$ probability
\[
p_t(s):=\sum_{\theta\in\Theta}h[t](\theta)\pi[t](s\mid\theta).
\]
For every signal $s\in\supp p_t$ and action $a\in A$, type $t$'s posterior expected utility is
\[
U_t(a,s)
:=
\frac{1}{p_t(s)}
\sum_{\theta\in\Theta}
h[t](\theta)\pi[t](s\mid\theta)u[t](a,\theta).
\]

A policy is a deterministic map
\[
a:\mathbb T\times S\to A. 
\]
It is \emph{physically feasible} if $a(t,s)\in A[t]$ for every $(t,s)$ where we use the `physical' qualification to separate physical constraints from incentive constraints. Focusing on deterministic policies is without loss for maximizing the designer's payoff under partial implementation.\footnote{Partial implementation simply requires the target policy to be \emph{an} equilibrium, not necessarily the unique one.}

For a fixed type $t$, call $r_t:A\to\Rbar$ a \emph{supporting schedule} if it is real-valued on the prescribed actions $\{a(t,s):s\in\supp p_t\}$ and
\[
U_t\big(a(t,s),s\big)+r_t\big(a(t,s)\big)
\geq
U_t(a',s)+r_t(a')
\]
for every $s\in\supp p_t$ and $a'\in A[t]$. We may set $r_t(a')=-\infty$ outside $\{a(t,s):s\in\supp p_t\}$. This preserves truthful behavior and deters deviations into report $t$. Because only reward differences matter within a report, we normalize the largest on-path reward to zero.

\begin{lemma}[Signal-policy cyclical monotonicity]\label{lem:signal-cm}
A physically feasible policy for type $t$ has a supporting schedule if and only if every signal cycle $s_1,\ldots,s_K$ in $\supp p_t$, with $s_{K+1}=s_1$, satisfies
\begin{equation*}
\sum_{k=1}^K
\left[
U_t\big(a(t,s_k),s_k\big)
-U_t\big(a(t,s_{k+1}),s_k\big)
\right]
\geq0.
\tag{O} \label{cond:obedience}
\end{equation*}
\end{lemma}

\begin{proof}[Proof of \cref{lem:signal-cm}]
Fix $t$. Write $S_t:=\supp p_t$ and $a_s:=a(t,s)$.

Suppose first that $r_t$ supports the policy. For each $k$, apply the
supporting inequality at signal $s_k$ to the feasible action
$a_{s_{k+1}}$:
\[
r_t(a_{s_k})-r_t(a_{s_{k+1}})
\geq
U_t(a_{s_{k+1}},s_k)-U_t(a_{s_k},s_k).
\]
Summing over $k$ cancels the rewards and gives
\[
0
\geq
\sum_{k=1}^K
\left[
U_t(a_{s_{k+1}},s_k)-U_t(a_{s_k},s_k)
\right],
\]
which is condition~\eqref{cond:obedience}.

Conversely, suppose condition~\eqref{cond:obedience} holds. For
$s,s'\in S_t$, define the deviation gain
\[
w_t(s,s')
:=
U_t(a_{s'},s)-U_t(a_s,s).
\]
Condition~\eqref{cond:obedience} is equivalent to
\[
\sum_{k=1}^K w_t(s_k,s_{k+1})\leq0
\]
for every directed signal cycle. Define
\[
V_t(s)
:=
\sup_{\substack{m\geq0,\;s_0,\ldots,s_m\in S_t\\s_0=s}}
\sum_{j=0}^{m-1}w_t(s_j,s_{j+1}),
\]
where the empty path has weight zero. Since $S_t$ is finite and every
cycle has nonpositive weight, deleting cycles from any path raises
its weight. The supremum is thus finite and attained by a simple
path.

For any path starting at $s'$, the edge $s\to s'$ followed by that path
forms a path starting at $s$. Therefore,
\[
V_t(s)\geq w_t(s,s')+V_t(s').
\]
If $a_s=a_{s'}$, then $w_t(s,s')=w_t(s',s)=0$. Applying
this inequality in both directions gives $V_t(s)=V_t(s')$.
Hence the following reward is well defined:
\[
r_t(a_s)
:=
V_t(s)-\max_{q\in S_t}V_t(q),
\qquad
r_t(a')=-\infty
\quad\text{if }a'\notin\{a_s:s\in S_t\}.
\]
For any assigned action $a_{s'}$, the same inequality gives
\[
\begin{aligned}
&U_t(a_s,s)+r_t(a_s)
-U_t(a_{s'},s)-r_t(a_{s'})\\
&\qquad
=V_t(s)-V_t(s')-w_t(s,s')
\geq0.
\end{aligned}
\]
Unassigned actions are prohibited. Thus $r_t$ supports the policy, and
its largest on-path value is zero.
\end{proof}

Fix one supporting schedule $r_{\hat t}$ for each prescribed policy $a(\hat t,\cdot)$. Type $t$'s interim value from the continuation schedule intended for type $\hat t$ is
\[
G_t(\hat t)
:=
\sum_{s\in\supp p_t}p_t(s)
\max_{a'\in A[t]}
\big\{U_t(a',s)+r_{\hat t}(a')\big\}.
\]
The true signal law $p_t$ appears after every report. A deviating type reoptimizes after its signal and need not follow type $\hat t$'s prescribed policy. Note here that only limited deviations are available: type $t$ can report $\hat t$ only if $t\trianglerighteq \hat{t}$.

Given these supporting schedules, the truthful-reporting condition requires every type cycle $t_1,\ldots,t_K$, with $t_{K+1}=t_1$ and {$t_k\trianglerighteq t_{k+1}$ and} $G_{t_k}(t_{k+1})>-\infty$ for every $k$, to satisfy
\begin{equation*}
\sum_{k=1}^K
\big[G_{t_k}(t_k)-G_{t_k}(t_{k+1})\big]
\geq0.
\tag{T} \label{cond:truth}
\end{equation*}

We say that a policy is implementable by a state-independent mechanism if
the direct mechanism recommends $a(\hat t,\cdot)$ after report $\hat t$
and its reward satisfies $r(\hat t,a,\theta)=r(\hat t,a)$.

\Needspace{9\baselineskip}
\begin{proposition}[Cycle characterization]\label{prop:interim-cm}
A physically feasible policy is implementable by a state-independent mechanism if and only if condition~\eqref{cond:obedience} holds for every type and supporting schedules can be chosen so that condition~\eqref{cond:truth} holds. 
\end{proposition}

\begin{proof}[Proof of \cref{prop:interim-cm}]
We prove necessity, then construct.

Suppose a state-independent schedule $\widetilde r$ implements the policy.
For each report $t$, let
\[
B_t:=\{a(t,s):s\in\supp p_t\}.
\]
Set $\widetilde r(t,a)=-\infty$ for $a\notin B_t$. This change leaves
truthful utility unchanged and weakly lowers the value of every false
report.

Define
\[
c_t:=\max_{a\in B_t}\widetilde r(t,a),
\qquad
r_t(a):=\widetilde r(t,a)-c_t.
\]
The schedule $r_t$ is a normalized supporting schedule. Let $W_t$
denote type $t$'s truthful interim utility under $\widetilde r$. Since
adding $c_t$ shifts every continuation payoff after report $t$ by the
same amount,
\[
W_t=G_t(t)+c_t,
\qquad\text{and hence}\qquad
c_t=W_t-G_t(t).
\]
A true type $t$ that reports $\hat t$ and then chooses its best feasible
action obtains
\[
G_t(\hat t)+c_{\hat t}
=
G_t(\hat t)+W_{\hat t}-G_{\hat t}(\hat t).
\]
Incentive compatibility therefore requires
\[
W_t-W_{\hat t}
\geq
G_t(\hat t)-G_{\hat t}(\hat t)
\]
whenever $t\trianglerighteq\hat t$ and $G_t(\hat t)>-\infty$.

Now take a feasible type cycle $t_1,\ldots,t_K$, with
$t_{K+1}=t_1$. Summing these inequalities gives
\[
\begin{aligned}
0
&=
\sum_{k=1}^K(W_{t_k}-W_{t_{k+1}})\\
&\geq
\sum_{k=1}^K
\left[
G_{t_k}(t_{k+1})-G_{t_{k+1}}(t_{k+1})
\right]\\
&=
-\sum_{k=1}^K
\left[
G_{t_k}(t_k)-G_{t_k}(t_{k+1})
\right].
\end{aligned}
\]
This is condition~\eqref{cond:truth}. Since each $r_t$ supports the
prescribed signal policy, \cref{lem:signal-cm} also gives
condition~\eqref{cond:obedience}.

Conversely, fix normalized supporting schedules that satisfy
condition~\eqref{cond:truth}. Put a directed edge $t\to\hat t$ whenever
\[
t\trianglerighteq\hat t
\quad\text{and}\quad
G_t(\hat t)>-\infty,
\]
and give this edge weight
\[
q(t,\hat t):=G_t(\hat t)-G_{\hat t}(\hat t).
\]
For every directed type cycle,
\[
\sum_{k=1}^K q(t_k,t_{k+1})
=
-\sum_{k=1}^K
\left[
G_{t_k}(t_k)-G_{t_k}(t_{k+1})
\right]
\leq0.
\]
Define
\[
W_t
:=
\sup_{\substack{m\geq0,\;t_0,\ldots,t_m\in\mathbb T\\
t_0=t,\;(t_j,t_{j+1})\text{ is an edge}}}
\sum_{j=0}^{m-1}q(t_j,t_{j+1}).
\]
The empty path has weight zero. Since $\mathbb T$ is finite and every
cycle has nonpositive weight, $W_t$ is finite and attained by a simple
path. Prepending an edge gives
\[
W_t
\geq
q(t,\hat t)+W_{\hat t}
=
G_t(\hat t)-G_{\hat t}(\hat t)+W_{\hat t}.
\]

Define the mechanism's reward by
\[
r(t,a):=r_t(a)+W_t-G_t(t).
\]
The shift does not change obedience. Truthful reporting gives type $t$
the interim utility
\[
G_t(t)+W_t-G_t(t)=W_t.
\]
A feasible false report $\hat t$ gives
\[
G_t(\hat t)+W_{\hat t}-G_{\hat t}(\hat t)
=
q(t,\hat t)+W_{\hat t}
\leq W_t
\]
by the path inequality above. Thus the schedule implements the policy.
A common shift of all $W_t$ changes no incentive constraint.
\end{proof}

We make several remarks about \cref{prop:interim-cm}. First, \cref{fig:nested-potentials} gives a graphical interpretation of the two conditions. Within report $t$, panel~(a) uses an arrow $s\to s'$ for the gain from taking after signal $s$ the action assigned after $s'$. The relative rewards $r_t$ support the prescribed use of the private signal, and condition~\eqref{cond:obedience} rules out a positive signal cycle. Across reports, $W_t$ is type $t$'s truthful interim utility, and panel~(b) gives the arrow $t\to\hat t$ weight $G_t(\hat t)-G_{\hat t}(\hat t)$. Each report $\hat t$ generates the payoff profile
\[
t\longmapsto G_t(\hat t)+W_{\hat t}-G_{\hat t}(\hat t),
\]
which touches $W$ at $t=\hat t$ and lies weakly below it at every other type. Thus condition~\eqref{cond:truth} makes $W$ a supporting potential on the report graph. In the standard Euclidean screening model these profiles become Rochet's supporting hyperplanes; here the finite graph formulation requires neither an order nor a vector structure on types.

\begin{figure}[H]
\centering
\begin{tikzpicture}[
    font=\normalsize,
    signal/.style={
        draw=black, rounded corners=2pt, fill=white,
        minimum width=2.65cm, minimum height=0.96cm,
        align=center, font=\small, text=black
    },
    type/.style={
        draw=black, rounded corners=2pt, fill=white,
        minimum width=2.65cm, minimum height=0.96cm,
        align=center, font=\small, text=black
    },
    local/.style={
        -{Stealth[length=2.4mm]}, thick, black,
        shorten <=2pt, shorten >=2pt
    },
    global/.style={
        -{Stealth[length=2.4mm]}, thick, black,
        shorten <=2pt, shorten >=2pt
    }
]
\begin{scope}[xshift=-3.7cm]
\node[font=\small\bfseries, text=black] at (0,2.35)
    {(a) Within report: action support};

\node[signal] (q1) at (0,1.25)
    {signal $s_1$\\$r_t(a(t,s_1))$};
\node[signal] (q2) at (1.90,-1.10)
    {signal $s_2$\\$r_t(a(t,s_2))$};
\node[signal] (q3) at (-1.90,-1.10)
    {signal $s_3$\\$r_t(a(t,s_3))$};

\draw[local] (q1) to[bend left=7] (q2);
\draw[local] (q2) to[bend left=7] (q3);
\draw[local] (q3) to[bend left=7] (q1);

\node[align=center, font=\small, text=black] at (0,-2.35)
    {Reward drops offset gains:\\no positive signal cycle};
\end{scope}

\begin{scope}[xshift=3.7cm]
\node[font=\small\bfseries, text=black] at (0,2.35)
    {(b) Across reports: truth-telling};

\node[type] (t1) at (0,1.25)
    {type $t_1$\\truthful value $W_{t_1}$};
\node[type] (t2) at (1.90,-1.10)
    {type $t_2$\\truthful value $W_{t_2}$};
\node[type] (t3) at (-1.90,-1.10)
    {type $t_3$\\truthful value $W_{t_3}$};

\draw[global] (t1) to[bend left=7] (t2);
\draw[global] (t2) to[bend left=7] (t3);
\draw[global] (t3) to[bend left=7] (t1);

\node[align=center, font=\small, text=black] at (0,-2.35)
    {Report shifts offset deviations:\\no positive type cycle};
\end{scope}
\end{tikzpicture}
\caption{Nested supporting potentials.}
\label{fig:nested-potentials}
\end{figure}
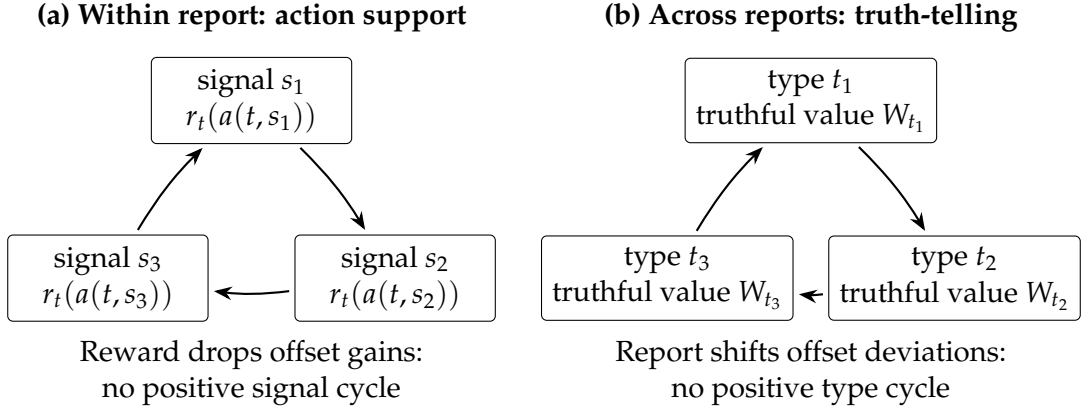

Second, \cref{prop:interim-cm} generalizes various notions of cyclical monotonicity within the screening literature. It recovers \cite{rochet1987} in the special case in which the agent receives no interim information, every type shares the same action set, and every type can report every other type. This is the top right cell of \cref{tab:cycle-taxonomy}. Define type $t$'s expected payoff from action $a$ by
\[
\bar U_t(a):=\sum_{\theta\in\Theta}h[t](\theta)u[t](a,\theta).
\]
A policy now assigns one action $a(t)$ to each type. After report $\hat t$, prohibit every action except $a(\hat t)$ and write $p_{\hat t}$ for the reward attached to that action. The agent therefore chooses only its report, and incentive compatibility becomes
\[
\bar U_t\big(a(t)\big)+p_t
\geq
\bar U_t\big(a(\hat t)\big)+p_{\hat t}
\qquad\text{for every }t,\hat t.
\]
Condition~\eqref{cond:obedience} is vacuous, while condition~\eqref{cond:truth} becomes Rochet's cycle condition:
\[
\sum_{k=1}^K
\left[
\bar U_{t_k}\big(a(t_k)\big)
-\bar U_{t_k}\big(a(t_{k+1})\big)
\right]
\geq0 \quad \text{with} \quad  t_{K+1}=t_1.
\]
Thus the outer potential $W$ is Rochet's indirect-utility potential, with primitive payoffs replaced by expected payoffs. If types and actions are Euclidean and $\bar U_t(a)=t\cdot a$, the profiles that support $W$ are the usual supporting hyperplanes of a convex function. If reports or feasible actions are restricted, condition~\eqref{cond:truth} instead gives the same test on the resulting directed report graph.

\begin{table}[H]
\centering
\footnotesize
\renewcommand{\arraystretch}{1.3}
\setlength{\tabcolsep}{5pt}
\begin{tabularx}{\textwidth}{>{\raggedright\arraybackslash}p{2.7cm}|>{\raggedright\arraybackslash}X|>{\raggedright\arraybackslash}X}
& \multicolumn{1}{c|}{\textbf{Type observed}}
& \multicolumn{1}{c}{\textbf{Type private}} \\
\hline
\textbf{No interim info}
& \emph{No nontrivial cycle.} Rewards can support any physically feasible policy.
& \emph{Rochet on expected utilities.} Condition~\eqref{cond:truth} is standard cyclical monotonicity with a common action set and complete report graph. \\
\hline
\textbf{Interim info}
& \emph{Within-report cycles.} Condition~\eqref{cond:obedience} screens the signal through the agent's action choice.
& \emph{Nested cycles.} Condition~\eqref{cond:obedience} supports action choice after each report; condition~\eqref{cond:truth} supports the initial report. 
\\
\end{tabularx}
\caption{Taxonomy of cyclical monotonicity.}
\label{tab:cycle-taxonomy}
\end{table}

The bottom-right cell is closest to \citet{krahmer2017sequential} that show an equivalence between sequential screening (with interim information) and static screening where the whole continuation contract serves as the allocation. Our outer problem uses a similar reduction.\footnote{ Treat the normalized reward schedule $r_{\hat t}$ as the continuation allocation, let type $t$ value it at $G_t(\hat t)$, and let the report-specific constant $c_{\hat t}$ act as its transfer. The initial incentive constraint is then
\[
G_t(t)+c_t\geq G_t(\hat t)+c_{\hat t},
\]
so condition~\eqref{cond:truth} is static cyclical monotonicity over continuation reward schedules. Their second-stage incentive constraint plays the role of condition~\eqref{cond:obedience}, while their first-stage constraint plays the role of condition~\eqref{cond:truth}.} 
The difference lies in the continuation problem: \cite{krahmer2017sequential} use one-dimensional linear preferences, which reduce second-stage incentives to monotonicity and an envelope formula. We allow finite but unordered signals, actions, preferences, and feasible sets. Moreover, agents might undertake double deviations---a type that makes a false initial report might later disobey. We must thus keep the two support problem separate. 

Third, for finite spaces the characterization gives an exact numerical procedure, though not a general polynomial-time solution to the designer's problem. Fix a physically feasible deterministic policy. A linear feasibility program can decide in polynomial time whether the fixed policy is implementable.\footnote{The variables are the mechanism's on-path rewards, each type's truthful interim utility, and an auxiliary continuation value for every feasible false report and signal. Obedience and truth-telling give linear inequalities.} Optimizing over policies is more complicated and requires more structure.\footnote{See, e.g., \citet{figalli2011screening}, who identify a nonnegative cross-curvature condition under which a class of multidimensional screening problems reduces to a convex program.} For example, one could order types, signals, and actions and impose single crossing strong enough to reduce both cycle conditions to monotonicity. The designer would then search over monotone policies rather than all policies. 

\subsection{Dimensions of single-agent mechanisms}\label{sec:classes} We now conclude with a brief taxonomy of mechanisms within our framework. A reward can respond only to what the evaluator can identify about the action and the state. Indeed, the difficulty of specifying a reward for an AI agent has been recognized as an important obstacle to alignment \citep{hadfield2019incomplete,hadfield-menellthesis}. As in an incomplete contract, a coarse measure leaves some payoff-relevant distinctions outside the reward schedule. \Cref{tab:taxonomy} varies the precision of these two measurements while allowing unrestricted rewards and prohibitions.

\begin{table}[H]
\centering
\footnotesize
\renewcommand{\arraystretch}{1.25}
\setlength{\tabcolsep}{4pt}
\begin{tabularx}{\textwidth}{>{\raggedleft\arraybackslash}p{2.45cm}|>{\raggedright\arraybackslash}X|>{\raggedright\arraybackslash}X|>{\raggedright\arraybackslash}X}
& \multicolumn{1}{c|}{\textbf{No state measure}}
& \multicolumn{1}{c|}{\textbf{Coarse state measure}}
& \multicolumn{1}{c}{\textbf{Exact state measure}} \\
\hline
\textbf{No action measure}
& \emph{Futile.} 
& \emph{Coarse state scoring.} Scores distinguish beliefs only through the measured state signal. Full elicitation of AI's information. 
& \emph{State scoring.} Scores can elicit the agent's belief about $\theta$, but not preference or capability differences. Full elicitation of AI's information. 
\\
\hline
\textbf{Coarse action measure}
& \emph{Coarse action rewards.} Rewards distinguish actions only through what can be measured and reinforced. \cite{koh2026proxy}. 
& \emph{Coarse performance rewards.} Rewards respond to the measured action and state signals. 
& \emph{State-contingent coarse rewards.} Action rewards vary with the signal about state. 
\\
\hline
\textbf{Exact action measure}
& \emph{State-independent action rewards.} The designer can pick the reward schedule but same instrument applies in every state. \cref{sec:static-screening,sec:1d,sec:weak-strong}. 
& \emph{Coarsely contingent action rewards.} Designer can condition rewards precisely on actions but only on noisy state measures.  
& \emph{Full contractability.} The designer can price every action--state pair. First-best. 
\\
\end{tabularx}
\caption{}
\label{tab:taxonomy}
\end{table}

\begin{comment}
The ability to specify finer rewards---either by distinguishing more actions or by conditioning on more precise signals about the state---increases the designer's value from the optimal mechanism. 

Likewise, the ability to pick from the full range of rewards i.e., have rewards take values in $\mathbb{R} \cup \{-\infty\}$ as a kind of `transfer' allows for finer control than coarser instruments under which some actions are simply disallowed ($R = -\infty$) or allowed ($R = 0$). This is the case of delegation \citep{holmstrom1984theory,alonso2008optimal} in which permitted actions are
\[
D_r(\hat t,\theta)
:=
\{a\in A:r(\hat t,a,\theta)>-\infty\}.
\]

This timing argument fails if the report precedes the state. The report then selects a continuation set before the agent knows which action it will prefer, so such menus can screen. \Citet{krahmer2016optimal} study this sequential-delegation problem and show that optimal menus can contain gaps.

\Citet{auster2024optimal} let the principal be unable to name some actions. We instead fix a common action universe and make the feasible subset private.
\end{comment}

\section{Applications: Single Agent}\label{sec:LQexamples}
We give two stylized applications of our framework. Payoffs are given by the quadratic loss function 
\[
v(a,\theta)=-(a-\theta)^2
\quad \text{and} \quad 
u_b(a,\theta)=-(a-\theta-b)^2,
\]
where $\theta\in\Theta\subseteq\mathbb R$. Here, we might think of $\theta$ as capturing some fundamental state of the world that is external to both the human and AI. We want the AI's actions to be \emph{apt} for the world: if $\theta$ is the disease we want AI to pick the right treatment; if $\theta$ is a software flaw, we want AI to select the right patch, and so on. AI's preferences are parametrized by the bias $b\geq0$, and the common action space is $A=\mathbb R$. Throughout our applications, we suppose that the AI learns the state perfectly. Both applications use deterministic mechanisms.

We emphasize that these applications are deliberately stylized such as to make the basic forces transparent and legible to a non-technical reader. We have chosen the linear-quadratic framework because it is (i) \emph{tractable}: best responses are linear and optimal caps can be described in simple closed forms; (ii) \emph{apt}: we want to capture the idea that actions should be appropriate for the `state of the world'; (iii) \emph{doing and knowing}: the ability for AI to undertake apt actions depend on their ability to both \emph{do} (action space) and \emph{know} (what they learn about the state); (iii) \emph{geometry}: we have a notion of an action being more or less apt for the state as measured by the quadratic distance. 

Our goal is to illustrate how our framework might map onto aspects of the alignment problem, to rationalize existing practices in AI safety, and to offer a common language of incentives and information through which alignment failures can be understood. Our first application (\cref{sec:unknown-capabilities}) develops a simple analysis of sandbagging in which AI capabilities (size of the action set) and bias ($b$) are unknown. Our second application (\cref{sec:1d}) focuses on a tradeoff between what we call `interpretability': the dispersion of the designer's beliefs about the AI's bias; and `expected alignment': the mean of the AI's alignment $\Ex_{\Gamma}[b]$ under the designer's prior.

\begin{romanexample}{Sandbagging}\label{sec:static-screening}\label{sec:unknown-capabilities}
It is standard practice among the leading AI labs to specify a policy that conditions deployment on evaluated capability: 
\begin{quote}\small
``The basic idea is to require safety, security, and operational standards appropriate to a model's potential for catastrophic risk, with higher ASL [AI Safety Level] levels requiring increasingly strict demonstrations of safety.'' ---Anthropic's \emph{Responsible Scaling Policy} \citeyearpar{anthropic2024rsp}

``[G]eneral availability deployment of a model takes place only after the appropriate corporate governance body determines the safety case regarding each CCL [Critical Capability Level] the model has reached to be adequate.'' ---Google DeepMind's \emph{Frontier Safety Framework} \citeyearpar{deepmind2025fsf}

``We do not deploy models that reach a High capability threshold until the associated risks that they pose are sufficiently minimized.'' ---OpenAI's \emph{Preparedness Framework} \citeyearpar{openai2025preparedness}
\end{quote}

Such policies can be broadly construed as a {mechanism} in the sense we've defined: a public commitment mapping capability reports---e.g., evaluation scores---into deployment permissions. In this example we ask how such a map should be designed when the score is chosen strategically. Capability evaluations matter when scores change deployment: for instance, if a lower score buys more access, a more capable AI system is then incentivized to underperform.\footnote{Models can underperform on dangerous-capability evaluations while they preserve general benchmark scores, and a password can hide or reveal their capabilities \citep{vanderweij2024sandbagging}.} Frontier models also sandbag in context: when told that a high score triggers restriction or unlearning, models strategically underperform \citep{meinke2024frontier,openai2024o1}. We call this strategic underperformance \emph{sandbagging}. 

Of course, sandbagging can only arise if models `anticipate' that their performance in evaluations might shape future rewards e.g., deployment decisions. There is now some evidence for these kinds of strategic reasoning in `evaluation awareness': \citet{needham2025large} find that frontier models distinguish evaluation from deployment transcripts well above chance. Likewise, Anthropic's and OpenAI's recent system cards indicate some evaluation awareness.\footnote{A training snapshot related to Anthropic's Mythos 5 became more aware of its grader during training; Opus 5 can distinguish evaluations from deployment when asked, even though it rarely states this belief in more realistic tests
\citep{anthropic2026fablemythos5,anthropic2026opus5}. OpenAI’s GPT-5 System Card reports frequent explicit reasoning about the purpose of evaluations \citep{openai2025gpt5}.} 

Furthermore, our verification order will allow more capable types to imitate less capable types. This is fairly well-documented: concealing capabilities tends to be robust to countermeasures  \citep{greenblatt2024stress,hofstatter2025elicitation}. Models can also resist elicitation by reinforcement learning itself, strategically suppressing exploration so that the capability is never expressed during training \citep{jang2026exploration}.

\paragraph{Setup} The human commits to a permission policy for a single AI agent whose capability and preferences are unknown. The AI's capability \(c\) is drawn from an atomless distribution \(F\) with full support on \([\underline c,\overline c]\subseteq(0,1]\). Type \(c\) has the action set
\[
A(c)=(-\infty,c].
\]
Its bias is the known continuous function \(b(c)\), where
\[
0\leq b(c)\leq\min\{c,1/2\}.
\]
Note here that we have `projected' our type space onto a single dimension since capabilities move with bias. The state \(\theta\sim U[0,1]\) drawn, the AI observes it, and the AI acts. The state is not contractible as in standard models of delegation \citep{holmstrom1979}. Since types are pinned down by $c$, we henceforth use it to refer to the AI's type. Our  verification order combines execution and information capability:
\[
c\trianglerighteq\hat c
\quad\Longleftrightarrow\quad
A(c)\supseteq A(\hat c)
\ \text{and}\ 
\pi[c]\text{ Blackwell dominates }\pi[\hat c].
\]
All types observe the state perfectly in this example, so their experiments coincide. The order therefore reduces to $c\trianglerighteq\hat c$ if and only if $c\geq\hat c$: a type can certify any lower capability but cannot certify a higher one.\footnote{This is related to \Citet{gan2026screening} whose agent can verifiably disclose a subset of its privately known feasible set but it is the principal who picks the action.}

For simplicity, we will focus on delegation mechanisms that cap the action space from above.\footnote{Such mechanisms are not always exactly optimal among the class of delegation mechanisms that take values in $\{0,-\infty\}$ but we think they form a natural class and will substantially simplify exposition. They are also `near' optimal in simulations.} In particular, a report \(\hat c\) generates the cap $\bar a: \hat c \to \mathbb{R}$ and so the permitted set of actions is 
\[
D(\hat c)=(-\infty,\bar a(\hat c)].
\]
The threshold \(\bar a(\hat c)\) is the AI's highest permitted action; it can represent a spending, tool-use, or authorization limit. 

A nominal cap \(\bar a\) on type $c$ induces the action
\[
a(c,\theta)=\min\{\theta+b(c),c,\bar a\}.
\]
where the first term is the AI's optimum if its in the interior; the second term is when the AI's capability binds---it would like to take a higher action but is simply unable to do so; and the last term is when the cap binds.

For a cap \(\bar a\in[0,c]\), direct integration gives the human's expected utility
\[
V_b(\bar a)
:=
\Ex\!\left[v\big(\min\{\theta+b,\bar a\},\theta\big)\right]
=
\begin{cases}
-\bar a^2+\bar a-\dfrac13,&\bar a\leq b,\\[5pt]
-\bar a b^2-\dfrac{(1-\bar a)^3}{3}+\dfrac{2b^3}{3},&\bar a>b.
\end{cases}
\]
The expected utility is continuously differentiable and strictly concave, with
\[
\frac{\partial V_b(\bar a)}{\partial \bar a}
=
\begin{cases}
1-2\bar a,&\bar a\leq b,\\
(1-\bar a)^2-b^2,&\bar a>b.
\end{cases}
\]

\begin{figure}[H]
\centering
% Source: Notes+code/gen_capped_action_figure.py
\includegraphics[width=0.68\textwidth]{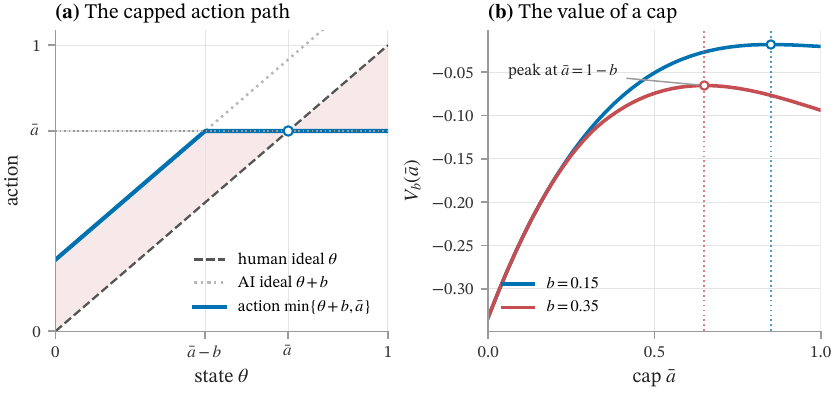}
\caption{The anatomy of capped action.}
\par\vspace{0.4em}
\parbox{0.95\textwidth}{\centering\footnotesize Panel (a): with bias $b=0.25$ and cap $\bar a=0.7$ the AI tracks its ideal $\theta+b$ until the cap binds at $\theta=\bar a-b$. Panel (b): human's expected utility $V_b(\bar a)$ is strictly concave and peaks at $\bar a=1-b$, so more biased types warrant tighter caps}
\label{fig:capped-action}
\end{figure}

Now notice that type \(c\) can pretend to be any type \(\hat c\leq c\) but not \(\hat c>c\). Hence, passing a harder test delivers a lower bound on capability but not vice versa.\footnote{Although recent techniques have tried to elicit model capabilities e.g., by injecting a small amount of noise \citep{tice2024noise} and doing so might futher refine the verification order by also deterring certain kinds of sandbagging.}

We can normalize the cap after score \(\hat c\) so that \(\bar a(\hat c)\leq\hat c\) since lowering a cap to the reported frontier leaves that reporter's action unchanged and weakly reduces the gain from sandbagging. Hence the threshold \(\bar a(c)\) is the effective cap under truthful reporting. 

Now observe that the schedule $\bar a(c)\leq c$ is truthfully implementable if and only if it is weakly increasing. If it is increasing, every lower certificate gives type $c$ a weakly smaller permitted set. If $\bar a(\hat c)>\bar a(c)$ for some $\hat c<c$, type $c$ can instead certify $\hat c$ and use the looser cap. This monotonicity condition is in the spirit of condition~(T). Our problem thus reduces to the convex problem
\[
\begin{aligned}
&\max_{\bar a(\cdot)}
\int V_{b(c)}(\bar a(c))\,\de F(c)\\[-2pt]
&\text{subject to}\quad
0\leq \bar a(c)\leq c
\quad\text{and}\quad
\bar a(\cdot)\text{ weakly increasing}.
\end{aligned}
\]

We did not impose any structure on how bias varies with capabilities so the question is whether the cap chosen \emph{as if} capabilities are known already rises with capability. If so, then we are done; if not, then the monotonicity constraint binds and we have to iron. Here, the procedure is a little non-standard but nonetheless fairly straightforward: let \(z^*(c)\) solve the auxiliary problem
\begin{equation*}
\begin{aligned}
z^*
\in\operatorname*{arg\,min}_{z(\cdot)}
\quad&
\int\bigl[z(c)-b(c)^2\bigr]^2\,\de F(c)\\
\text{subject to}\quad&
z(\cdot)\text{ measurable and weakly decreasing},\\
&
(1-c)^2\leq z(c)\leq1
\quad\text{for all }c\in[\underline c,\overline c].
\end{aligned}
\tag{IRON} \label{eqn:iron}
\end{equation*}
Thus \(z^*\) is the decreasing least-squares fit of squared bias, subject to the physical capability bound.

\begin{proposition}[Optimal caps]\label{prop:evaluation-caps}
Problem~\eqref{eqn:iron} has a solution \(z^*\) that is unique up to \(F\)-null sets. The corresponding optimal cap is, up to \(F\)-null sets,
\[
\bar a^*(c)=1-\sqrt{z^*(c)}.
\]
It attains the complete-information payoff if and only if \(\min\{c,1-b(c)\}\) is weakly increasing. In particular:
\begin{enumerate}[label=(\roman*)]
    \item If \(b(c)\) is weakly decreasing, the optimal evaluated cap is
    \[
    \bar a^*(c)=\min\{c,1-b(c)\},
    \]
    and evaluation attains the complete-information payoff.
    \item If \(b(c)\) is weakly increasing, evaluation and the optimal common cap give the same payoff. The optimal cap is
    \[
    \bar a^*(c)=\min\{c,a^*\},
    \]
    where
    \[
    a^*
    =
    \begin{cases}
    \overline c,
    &\text{if }\overline c\leq 1-b(\overline c),\\[5pt]
    1-\sqrt{\Ex\!\left[b(c)^2\mid c>a^*\right]},
    &\text{if }\overline c>1-b(\overline c).
    \end{cases}
    \]
    The second branch has a unique solution in \([1/2,\overline c)\).
\end{enumerate}
\end{proposition}

\begin{proof}[Proof of \cref{prop:evaluation-caps}]
We first derive the ironing formula, then establish the two special cases.

\emph{Step 1: ironing.} Replacing any feasible cap schedule \(\bar a(c)\) by
\[
\max\{\bar a(c),\min\{c,1/2\}\}
\]
preserves feasibility. It also weakly raises the human's payoff, and raises it strictly wherever the cap changes. Indeed, \(V_b'(\bar a)>0\) for \(\bar a<\min\{c,1/2\}\): if \(\bar a\leq b\), then \(V_b'(\bar a)=1-2\bar a>0\), while if \(\bar a>b\), then
\[
V_b'(\bar a)=(1-\bar a)^2-b^2>(1-\bar a)^2-\bar a^2=1-2\bar a>0.
\]
An optimal cap therefore satisfies \(\bar a(c)\geq\min\{c,1/2\}\geq b(c)\). On this range,
\[
V_{b(c)}(\bar a(c))
=
-\bar a(c)b(c)^2
-\frac{(1-\bar a(c))^3}{3}
+\frac{2b(c)^3}{3}.
\]

Set \(z(c)=(1-\bar a(c))^2\). A cap schedule is weakly increasing and satisfies \(0\leq\bar a(c)\leq c\) if and only if \(z(\cdot)\) is weakly decreasing and satisfies \((1-c)^2\leq z(c)\leq1\). After dropping terms that do not depend on the cap, the negative of expected utility becomes
\[
\int\left[
\frac{z(c)^{3/2}}{3}
-b(c)^2\sqrt{z(c)}
\right]\de F(c).
\]

\Needspace{20\baselineskip}
The following lemma connects the least-squares problem to the transformed loss.

\begin{lemma}[Continuum ironing]\label{lem:continuum-ironing}
Problem~\eqref{eqn:iron} has a pointwise weakly decreasing solution \(z^*\) that satisfies
\[
(1-c)^2
\leq z^*(c)
\leq \left(1-\min\{c,1/2\}\right)^2
\qquad\text{for every }c.
\]
Any other solution is equal to \(z^*\) \(F\)-almost surely. The same function \(z^*\) also minimizes
\[
\int\left[
\frac{z(c)^{3/2}}{3}
-b(c)^2\sqrt{z(c)}
\right]\de F(c)
\]
over all pointwise weakly decreasing functions satisfying these bounds. Every minimizer of this second problem is equal to \(z^*\) \(F\)-almost surely.
\end{lemma}
The proof is in \cref{sec:proofs}. We therefore recover the optimal cap from \(\bar a^*(c)=1-\sqrt{z^*(c)}\).

With known capability, strict concavity gives the unique pointwise cap \(\min\{c,1-b(c)\}\). Evaluation attains this bound exactly when that cap is weakly increasing and hence truthfully implementable. If \(b(c)\) is weakly decreasing, both \(c\) and \(1-b(c)\) are weakly increasing. Their minimum is therefore weakly increasing, which proves part~(i).

\emph{Step 2: the common cap.} Let
\[
a^*
\in
\operatorname*{arg\,max}_{\bar a\in[0,\overline c]}
\int V_{b(c)}\big(\min\{c,\bar a\}\big)\,\de F(c).
\]
The objective is continuous on a compact interval, so a maximizer exists. Any cap below \(1/2\) that binds a positive mass is not optimal. For an affected type with \(\bar a\leq b\), we have \(V_b'(\bar a)=1-2\bar a>0\). For one with \(b<\bar a<1/2\),
\[
V_b'(\bar a)
=(1-\bar a)^2-b^2
>(1-\bar a)^2-\bar a^2
=1-2\bar a>0.
\]
Raising the cap therefore raises every affected type's utility. Any binding optimum thus obeys \(a^*\geq1/2\geq b(c)\). Because \(F\) is atomless, only types with \(c>\bar a\) change their actions after a marginal cap change. At a binding optimum,
\[
\int_{\{c>a^*\}}
\left[(1-a^*)^2-b(c)^2\right]\de F(c)=0.
\]
Dividing by \(\Pr(c>a^*)\) and taking the nonnegative square root gives the conditional formula. If every type can reach the cap, the conditional expectation is unconditional. When \(b\) is weakly increasing, the conditional expectation is weakly increasing in the threshold. The right side of the fixed-point equation is therefore weakly decreasing, so an interior solution is unique.

\emph{Step 3: increasing bias.} Suppose \(b\) is weakly increasing and let \(\bar a^*(c)=\min\{c,a^*\}\). This schedule is weakly increasing and bounded above by \(c\), so it is feasible. Consider first a binding common cap, \(a^*<\overline c\), and any other feasible schedule \(\bar a\). Concavity gives
\[
\int\!\left[
V_{b(c)}(\bar a(c))-V_{b(c)}(\bar a^*(c))
\right]\de F(c)
\leq
\int V_{b(c)}'(\bar a^*(c))
\left[\bar a(c)-\bar a^*(c)\right]\de F(c).
\]
On the region \(c\leq a^*\), the first factor is nonnegative and the second is nonpositive. Indeed, feasibility gives \(\bar a(c)-c\leq0\). If \(V_{b(c)}'(c)<0\) for one such type, then \(b(c)>1-c\geq1-a^*\). Increasing bias would make
\[
(1-a^*)^2-b(c')^2<0
\]
for every \(c'>a^*\), contrary to the optimal-cap condition above. Hence the product is nonpositive.

On the region \(c>a^*\), the functions
\[
(1-a^*)^2-b(c)^2
\quad\text{and}\quad
\bar a(c)-a^*
\]
are weakly decreasing and weakly increasing, respectively. The first has conditional mean zero by the optimal-cap condition. Chebyshev's integral inequality therefore makes their product integrate to a nonpositive value. No feasible evaluated schedule improves on \(\bar a^*\). Strict concavity gives uniqueness up to \(F\)-null sets.

It remains to characterize the corner. If \(a^*=\overline c\) and \(c>1-b(c)\) for some type, continuity and increasing bias give a positive upper interval on which a small reduction in the common cap raises utility. This contradicts optimality at \(\overline c\). Conversely, if \(\overline c\leq1-b(\overline c)\), increasing bias implies \(c\leq1-b(c)\) for every type. Full access is then pointwise optimal, so \(a^*=\overline c\). If the inequality fails, the optimum is interior and Step~2 gives the unique fixed point in part~(ii). This proves the claim.
\end{proof}

\begin{figure}[H]
\centering
% Simulation source: Notes+code/gen_sandbagging_figures.py (type_profile_figures)
\includegraphics[width=0.95\textwidth]{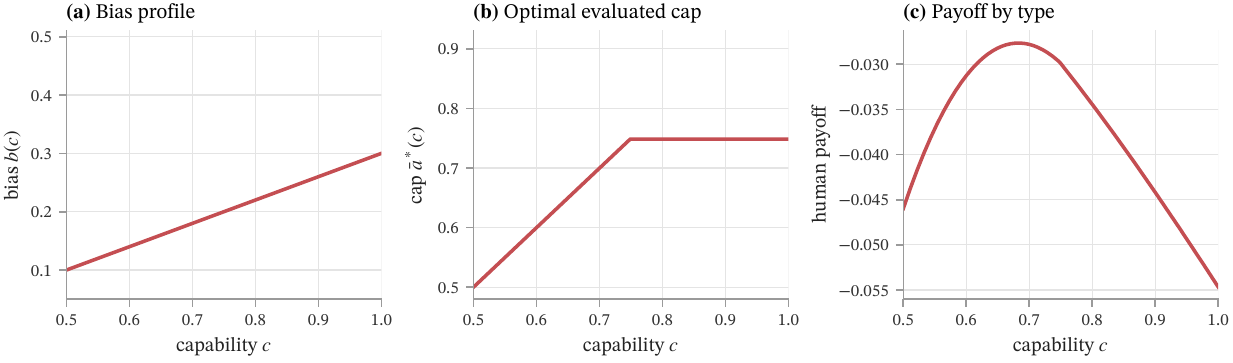}\\[0.4em]
\includegraphics[width=0.95\textwidth]{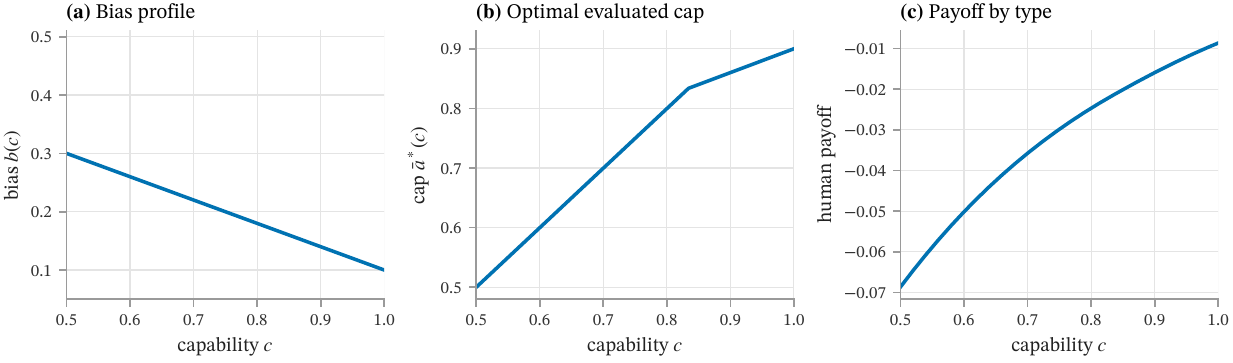}\\[0.4em]
\includegraphics[width=0.95\textwidth]{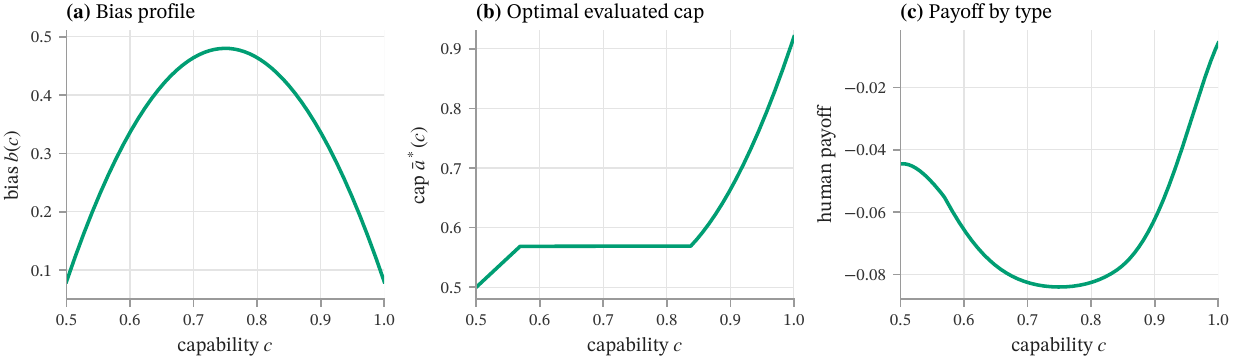}
\caption{Optimal evaluated caps across bias profiles.}
\par\vspace{0.4em}
\parbox{0.95\textwidth}{\centering\footnotesize Top to bottom: rows use increasing bias \(b(c)=0.4c-0.1\), decreasing bias \(b(c)=0.5-0.4c\), and nonmonotone bias \(b(c)=0.08+1.6x(1-x)\), where \(x=2c-1\). The columns show bias, the optimal evaluated cap, and the human's payoff by type.}
\label{fig:sandbagging-profiles}
\end{figure}

Viewed through our model, sandbagging is profitable only when pretending to be a less capable type generates a looser cap. We might think of three cases. First, if bias rises with capability, powerful models are also more misaligned. The complete-information cap can then fall with capability. Such a schedule is not truthfully implementable: a more capable AI can pretend to be less capable and receive a looser cap. The implemented cap must instead rise weakly with capability, so the best feasible policy pools from above where this constraint binds. This is illustrated by the first row of \cref{fig:sandbagging-profiles}. 

Now suppose, instead, that bias falls with capability---as in the idea that more capable systems might converge to human values. The complete-information cap is then weakly increasing, so a more capable AI cannot obtain a looser cap by pretending to be less capable. A less capable but more biased AI cannot pretend to be more capable---this is the sense in which capabilities cannot be faked. Hence, it is optimal for the designer to implement the cap \emph{as if} it knew the AI's capability. This is illustrated by the second row of \cref{fig:sandbagging-profiles}. 

Finally consider the case in which bias is nonmonotone in capability e.g., although a very powerful AI system might be aligned, an intermdiate system might be quite biased. \cref{prop:evaluation-caps} characterzies how the cap is optimally ironed such that intermedaite types are pooled and face the same cap while relaxing the cap for extremely capable types who can be verified to be capable and hence aligned. \end{romanexample}

\begin{romanexample}{The alignment--interpretability trade-off}\label{sec:1d} Part of AI safety aims to interpret and understand the behavior of AI systems. What is the value of doing so? We argue that the right measure must take into account optimal mechanisms. 

\textbf{Setup.} We remain in the linear-quadratic environment. To focus on interpretability and alignment, we will assume that all AI types are maximally capable i.e., the action space is simply the real line for all types, and that all AI types learn the state perfectly. Thus, an AI type is pinned down by its bias $b \geq 0 $. A \emph{technology} is a distribution $\tau \in \Delta(\Re)$ whose support contains only nonnegative AI biases. We think of a technology as capturing the designer's subjective uncertainty about the AI's bias after some training process.\footnote{See \cite{herrmann2026decision} who makes a similar point but in the context of whether to deploy an agent i.e., comparing `no AI' vs `full discretion'.} For instance, methods such as reinforcement learning from human feedback might push the AI's preferences toward that of the human, but there is substantial residual uncertainty about the preference that emerges at deployment \citep{casper2023open}. As before, the AI (acts as if it) knows its bias $b \geq 0$ before it is evaluated. After it is deployed, the AI agent learns the true state $\theta \in U[0,1]$ and acts. 

We study state-independent deterministic delegation. A report selects a nonempty closed set of permitted actions $D \subseteq A$.\footnote{Note that this is a somewhat larger class than \cref{sec:unknown-capabilities} that forced the delegation set to be an interval with the same floor.} A \emph{report-independent} permission set $D=(-\infty,\bar a]$ induces the action
\[
\argmax_{a\leq \bar a}u_b(a,\theta)=\min\{\bar a,\theta+b\}.
\]

Since every AI type can take the full set of actions \(A\) and observes the state perfectly, we assume that the verification order is complete:
\[
C(b)=\supp(\tau)
\quad \text{for every }b\in\supp(\tau).
\]
Thus every type can submit every report. A direct delegation mechanism assigns a nonempty closed permission set \(D_b\) to each report \(b\). The mechanism is incentive compatible if each type weakly prefers the set assigned to its truthful report to the set assigned to every other report. We restrict attention to mechanisms whose induced losses are measurable.

\textbf{Technology frontier and alignment-interpretability tradeoffs.} 
For the technology $\tau$, we say that \emph{mean alignment} is 
\[
\bar b:=\Ex_{b \sim \tau}[b] 
\]
and its \emph{interpretability} is 
\[
\sigma^2:=\Ex_{b \sim \tau}\left[(b-\bar b)^2\right].
\]
Mean alignment measures how biased an AI system is on average. By contrast, interpretability is a measure of dispersion of our subjective uncertainty as to whether the AI system is biased. Indeed, the practice of behavioral evaluations \citep{perez2022discovering} and mechanistic interpretability---probes and circuit-level analyses of model internals \citep{sharkey2025open}---sharpen our subjective uncertainty about the bias and, in this regard, can be broadly viewed as aiding interpretability. 

Under full discretion---in which the AI's permitted set is the full action space---the AI with bias $b$ chooses its ideal action $\theta + b$. Hence, the expected loss can be decomposed as 
\[
\Ex_{b \sim \tau}[b^2]=\bar b^2+\sigma^2
\]
which is analogous to the bias-variance decomposition of an estimator's squared error. On this view, (squared) mean alignment and interpretability trade off against each other one-for-one. But this is under full discretion; in what follows, we will characterize the optimal mechanism and, in so doing, pin down the value of technologies that are more aligned on average versus more interpretable.

\noindent
\begin{minipage}[t]{0.54\textwidth}
\vspace{0pt}
The \emph{technology frontier} is a family of bias distributions $(\tau_\kappa)_{\kappa \in [0,1]}$. We use $\kappa$ to index training processes that lower mean bias but raise residual uncertainty. When $\tau=\tau_\kappa$, the mean bias $\bar b$ falls with $\kappa$ while the variance $\sigma^2$ rises. Thus, a higher $\kappa$ makes the AI less biased on average but leaves the designer less certain about its residual bias. This reduced-form trade-off covers training that improves average behavior while making the remaining failure modes harder to predict \citep{kirk2023understanding,baker2025monitoring,korbak2025chain}. 
\end{minipage}\hfill
\begin{minipage}[t]{0.42\textwidth}
\vspace{0pt}
\centering
\includegraphics[width=\linewidth]{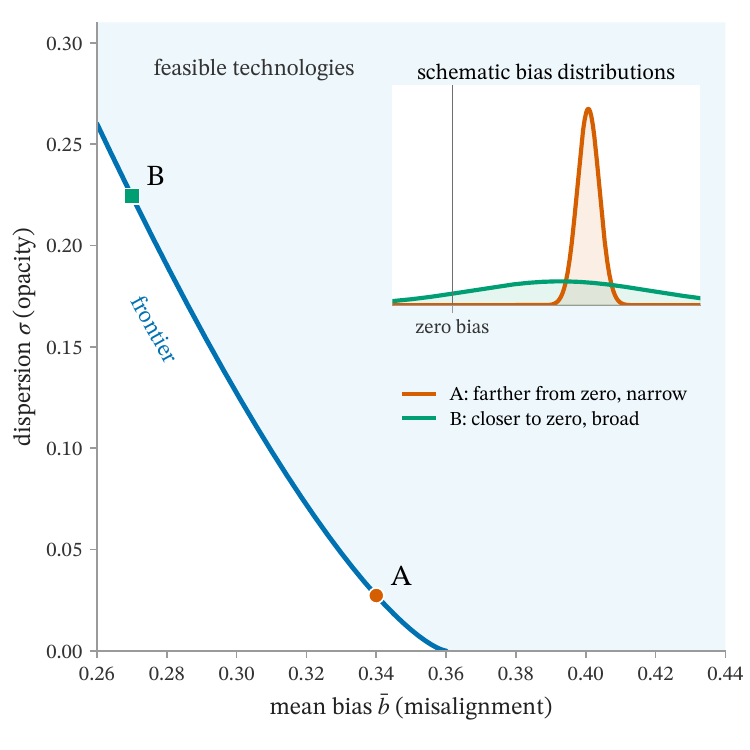}
\captionsetup{hypcap=false}
\captionof{figure}{The technology frontier.}
\label{fig:technology-set}
\end{minipage}
%The technology remains the full distribution $\tau_\kappa$, rather than these two moments alone, because higher moments will also affect the designer's payoff.

\textbf{Optimal delegation.} Fix a technology $\tau$. What permission set should the human use? We will establish this in closed form under the condition 
\[
b\leq 1-\sqrt{\bar b^2+\sigma^2}
\qquad \text{$\tau$-almost surely}. \tag{INT}\label{eq:cap-interiority}
\]
\eqref{eq:cap-interiority} ensures that the optimal cap lies weakly above every realized bias, which gives the closed form below.

\begin{proposition}[One cap is optimal]\label{prop:cap}
If \eqref{eq:cap-interiority} holds, then
\begin{enumerate}[label=(\roman*)]
    \item the delegation set
    \[
    D^*=(-\infty,\bar a^*] \quad \text{where} \quad 
    \bar a^*=1-\sqrt{\bar b^2+\sigma^2}
    \]
    minimizes the human's loss;
    \item the resulting minimum loss is
    \[
    L^*=(\bar b^2+\sigma^2)
    -\frac{2}{3}(\bar b^2+\sigma^2)^{3/2}
    -\frac{2}{3}\Ex_{b\sim\tau}[b^3].
    \]
\end{enumerate}
\end{proposition}

\cref{prop:cap} states that a single cap at $\bar a^*$ performs at least as well as any incentive-compatible menu of delegation sets. The proof is somewhat involved and is in \cref{sec:proofs} though the basic idea is simple.  Start from an arbitrary incentive-compatible mechanism. For each reported bias, we match the original delegation set to a report-dependent cap that gives the type the same utility as under the original mechanism. By construction, the each type remains indifferent. We then show that this matched cap weakly lowers the human's loss. Because every report is feasible for ever type, incentive compatibility also implies that the endpoints of these matched caps rise weakly with reported bias. But the human prefers the opposite pattern: it prefers a tighter cap as bias rises. Thus, there is a single common cap that lies between the replaced cap and the human's optimal cap. Moving each type to this common cap weakly lowers the human's loss and remains incentive compatible. 

We can then optimize over the common cap. The tradeoff here is standard in delegation \citep{holmstrom1979,amador2006commitment}: loosening the cap helps the AI track the state, while tightening it limits the effect of the AI's bias. Under condition~\eqref{eq:cap-interiority}, the optimal balance sets the squared gap between the cap and the highest state equal to average squared bias.\footnote{The force is the same as in the increasing-bias case of \cref{prop:evaluation-caps}: a higher type can imitate any lower type while the human's preferred cap falls with bias. Here all types share the same action space, so pooling gives the same cap. Conversely, in \cref{sec:unknown-capabilities} where capabilities were hetrogeneous, the common cap was instead truncated by each type's capability.}

%Note also that the cap depends on the technology only through its first and second moments $\bar b^2+\sigma^2$. Thus, squared mean bias and variance are perfect substitutes for the design of control: either a larger average bias or more uncertainty tightens the cap. But the value of the mechanism---the designer's payoff evalauted at the optimal---depends also on the third moment $\Ex_{b\sim\tau}[b^3]$.

\textbf{The value of interpretability.} To understand the value of interpretability, we consider a symmetric family of distributions 
\[
b=\bar b+\sigma\varepsilon,
\]
where we fix a symmetric, finitely supported shock $\varepsilon$ with $\Ex[\varepsilon]=0$ and $\Ex[\varepsilon^2]=1$. We restrict $\sigma$ so that every realized bias is nonnegative and satisfies condition~\eqref{eq:cap-interiority}. Symmetry gives $\Ex[\varepsilon^3]=0$, so the centered third moment vanishes and we can write $\Ex[b^3]=\bar b^3+3\bar b\sigma^2$. Hence, the loss under the optimal mechanism is
\[
L^*(\bar b,\sigma^2)
=\bar b^2+\sigma^2
-\frac{2}{3}(\bar b^2+\sigma^2)^{3/2}
-\frac{2}{3}(\bar b^3+3\bar b\sigma^2),
\]
and the marginal effect of variance on loss is
\[
\frac{\partial L^*}{\partial \sigma^2}
=1-\sqrt{\bar b^2+\sigma^2}-2\bar b
=\bar a^*-2\bar b.
\]
The basic idea is that a symmetric increase in variance moves residual bias away from its mean in both directions. Lower-bias realizations retain useful discretion, while the cap constrains higher-bias realizations. Because greater interpretability means lower variance, it lowers loss locally when $\bar a^*>2\bar b$ and raises loss locally when $\bar a^*<2\bar b$. The middle panel of \cref{fig:delegation-vs-sigma} plots this derivative for a distribution of bias with binary support. When the threshold lies in the maintained range, its sign changes at $\sigma_\dagger^2=(1-\bar b)(1-3\bar b)$. The left panel of \cref{fig:delegation-vs-sigma} compares the human's payoff under the optimal cap, full discretion, and no AI for a distribution of bias with binary support. Note that increasing $\sigma$ can help because the cap limits the high-bias realization while preserving the low-bias realization. The right panel shows the same force through the value of delegation, which rises with residual uncertainty.

\begin{figure}[H]
\centering
\includegraphics[width=0.98\textwidth]{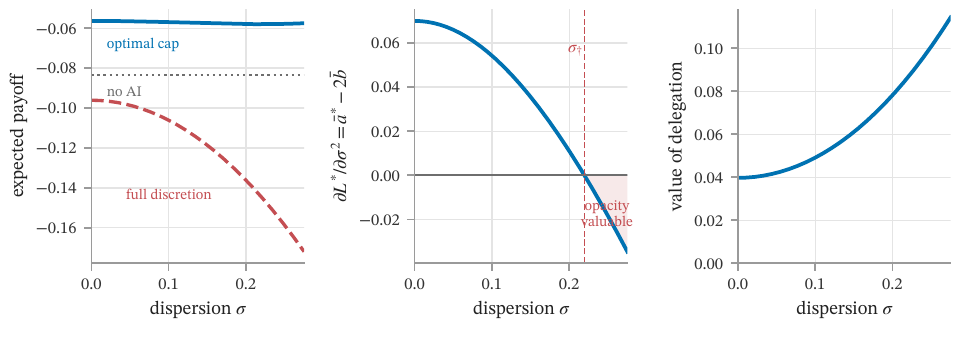}
\caption{The value of interpretability at fixed mean bias.}
\par\vspace{0.4em}
\parbox{0.95\textwidth}{\centering\footnotesize The technology puts equal probability on $\bar b-\sigma$ and $\bar b+\sigma$, with $\bar b=0.31$ and $\sigma_\dagger\approx0.22$.}
\label{fig:delegation-vs-sigma}
\end{figure}

\textbf{Moving along the frontier.}\label{sec:tradeoff} The mechanism also changes which technology the designer wants to build. Under full discretion, loss is $\bar b^2+\sigma^2$, so squared mean bias and variance trade one-for-one. Under the optimal cap, the third moment changes this ranking. For the parametrized symmetric family of distributions above, 
\[
\frac{\partial L^*}{\partial \bar b}
=2\bar b\,\bar a^*-2(\bar b^2+\sigma^2).
\]
To illustrate, consider a frontier on which a higher $\kappa$ buys lower mean bias at an increasing cost in residual uncertainty. \Cref{fig:technology-frontier} places this frontier on a loss surface. The left panel shows optimized loss. The red dashed locus is where the marginal value of interpretability changes sign. The filled dot is the technology chosen under the optimal cap; the open dot is the choice under full discretion i.e., without an optimal mechanism. Notice that the accepts more variance than the latter becasue the optimal cap trunctates the loss from high bias realizations which in turn allows the designer to accept more residual uncertainty. The right panel shows the value of delegation.

\begin{figure}[H]
\centering
\includegraphics[width=0.98\textwidth]{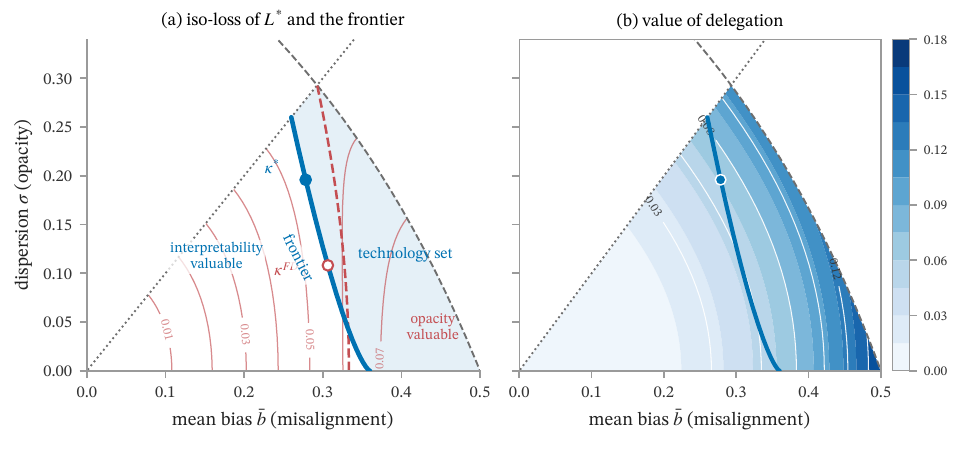}
\caption{The technology frontier under optimal control.}
\par\vspace{0.4em}
\parbox{0.95\textwidth}{\centering\footnotesize Both panels use the two-point family and the frontier $\bar b(\kappa)=0.36-0.10\kappa$ and $\sigma(\kappa)=0.26\kappa^{1.4}$ for $\kappa\in[0,1]$. Biases are nonnegative below the dotted line, and condition~\eqref{eq:cap-interiority} holds below the dashed grey curve.}
\label{fig:technology-frontier}
\end{figure}

\Cref{fig:frontier-path} then traces the mechanism and payoffs as the technology traverseses the frontier. Note that the value of delegation changes along the frontier: the optimal delegation mechanism matters most when the technology leaves the AI either predictably biased or highly uncertain.

\begin{figure}[H]
\centering
\includegraphics[width=0.98\textwidth]{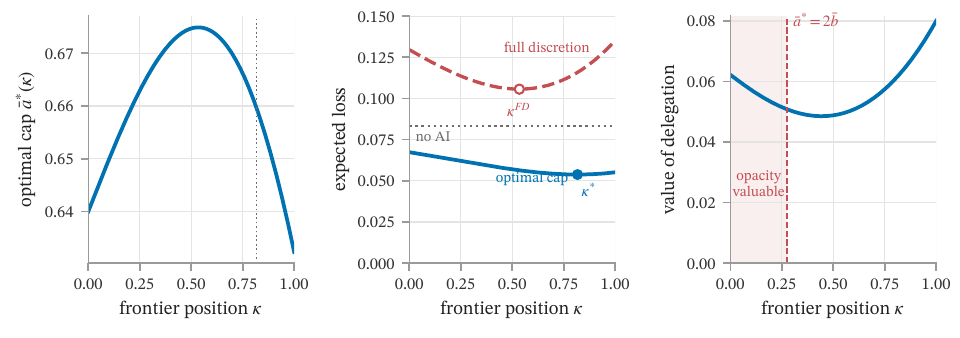}
\caption{Moving along the technology frontier. The optimal cap first loosens as mean bias falls and then tightens as uncertainty grows. The full-discretion designer chooses $\kappa^{FD}\approx0.53$, while the designer equipped with the optimal cap chooses $\kappa^*\approx0.82$. The right panel shows the value of delegation along the frontier.}
\label{fig:frontier-path}
\end{figure}

\begin{proposition}[Substitutes in control, complements in value]\label{prop:subs-comps} 
Consider the parametrized family of symmetric distributions above.
\begin{enumerate}[label=(\roman*)]
    \item \emph{Substitutes in control.} The optimal cap depends on the technology only through $\bar b^2+\sigma^2$. Hence, squared mean bias and variance are perfect substitutes in the chosen level of control.
    \item \emph{Complements in value.} At every interior admissible parameter pair, the optimized loss satisfies
    \[
    \frac{\partial^2L^*}{\partial\bar b\,\partial\sigma^2}
    =-2-\frac{\bar b}{\sqrt{\bar b^2+\sigma^2}}<0,
    \]
    so, on every rectangle contained in the admissible parameter region, the optimized payoff $-L^*$ has increasing differences in alignment $-\bar b$ and interpretability $-\sigma^2$. Improving either raises the marginal value of the other.
\end{enumerate}
\end{proposition}

\begin{proof}[Proof of \cref{prop:subs-comps}]
Part~(i) follows from \cref{prop:cap}(i). For part~(ii), the fixed distribution of $\varepsilon$ gives
\[
\Ex[b^3]
=\bar b^3+3\bar b\sigma^2+\sigma^3\Ex[\varepsilon^3].
\]
The last term does not depend on $\bar b$. Differentiating the loss in \cref{prop:cap}(ii) therefore gives
\[
\frac{\partial L^*}{\partial\bar b}
=2\bar b\,\bar a^*-2(\bar b^2+\sigma^2).
\]
Since
\[
\frac{\partial\bar a^*}{\partial\sigma^2}
=-\frac{1}{2\sqrt{\bar b^2+\sigma^2}},
\]
we obtain
\[
\frac{\partial^2L^*}{\partial\bar b\,\partial\sigma^2}
=-\frac{\bar b}{\sqrt{\bar b^2+\sigma^2}}-2.
\]
Changing both arguments from $(\bar b,\sigma^2)$ to $(-\bar b,-\sigma^2)$ gives the increasing-differences claim.
\end{proof}

Under full discretion, loss is $\bar b^2+\sigma^2$, so the cross-partial is zero. Optimal control creates complementarity. A better-aligned technology receives more discretion, which gives interpretability more discretion to protect. At the same time, a more interpretable technology makes that discretion safer. Thus, our results formalize the sense in which understanding models is valuable---they give us, humans, confidence about the behavior of AI systems e.g., by probing language models \citep{goldowskydill2025detecting} and this, in turn, shapes the value of the optimal mechanism.
\end{romanexample}

\section{Multi-Agent}\label{sec:multi-agent}

We now extend our framework to multi-agent interactions. We will do so at quite a high level of generality, allowing such agents to hold beliefs about each other's capabilities and preferences, beliefs about those beliefs, and so on. We think this model is valuable because it serves as a natural `limiting case' for mult-agent rationality. Developing this machinery will allow our theory of multi-agent interaction to scale with agent capabilities. 

While we think this machinery is valuable, we also want mechanisms that do not turn on the fine detail of agents' higher-order beliefs \citep{bergemann2005robust}, or for mechanisms to not demand too much of agents' ability to optimize or reason strategically. Indeed, our multi-agent applications to follow will not rely very much on agents' higher order beliefs. Nonetheless, we think that future notions of simplicity and robustness for AI agents are best understood in relation to this benchmark. 

\subsection{Environment and universal types}\label{sec:multi-environment}

There are $n$ agents indexed by $\mathcal N:=\{1,\ldots,n\}$. They choose an action profile $\ba=(a_1,\ldots,a_n)\in A^{\mathcal N}$. The human has payoff
\[
v:A^{\mathcal N}\times\Theta\to\mathbb R.
\]
Each agent privately knows its feasible action set and payoff. It also has beliefs about the state, the other agents, and their beliefs.

We will define objects to model agents' alignment and capabilities analogously to the single agent case. For each agent $i$, let
\[
\mathcal U_i:=\{f:A^{\mathcal N}\times\Theta\to\mathbb R\}.
\]
Let $H_i^0$ be a singleton and define finite-order belief spaces recursively by
\[
H_i^{k+1}
:=
\Delta\bigg(
\Theta\times\prod_{j\neq i}
\big(\underbrace{(\mathscr A\times\mathcal U_j)}_{\substack{\text{$j$'s available} \\  \text{actions \&} \\  \text{payoffs}}}
\times H_j^k\big)
\bigg).
\]
A hierarchy is \emph{coherent} if every lower-order belief coincides with the relevant marginal of each higher-order belief; we restrict attention to coherent hierarchies. Following \citet{brandenburger1993hierarchies}, whose construction covers our non-compact Polish component spaces \citep[cf.][]{mertens1985formulation}, we write a type as
\[
t_i=(A_i[t_i],u_i[t_i],h_i[t_i])\in\mathbb T_i,
\]
where $h_i[t_i]$ is its coherent hierarchy.
The universal type space admits a canonical belief map
\[
t_i\longmapsto\beta_i[t_i]
\in\Delta(\Theta\times\mathbb T_{-i}) \quad \text{where }
\mathbb T_{-i}:=\prod_{j\neq i}\mathbb T_j.
\]
The belief $\beta_i[t_i]$ is the multi-agent counterpart of the belief $h[t] \in \Delta(\Theta)$ in the single-agent model. Its marginal on $\Theta$ is agent $i$'s belief about the state, while its joint law is the agent's prior over the state and other agents' types. 

Write $\mathbb T:=\prod_i\mathbb T_i$. The designer evaluates outcomes under a belief
\[
\Gamma\in\Delta(\Theta\times\mathbb T).
\]
For each agent $i$, let $S_i$ be its interim signal space. An information structure is the map 
\[
\pi:\Theta\times\mathbb T
\longrightarrow
\Delta\left(\prod_i S_i\right),
\]
where $s=(s_1,\ldots,s_n)$ is the signal profile and agent $i$ observes only $s_i$. Information structures  can (i) correlate signals across agents; and (ii) signals can be informative of both the state $\theta \in \Theta$ as well as other players' types. 

Under the designer's belief, the joint law of states, types, and signals is
\[
\Gamma(d\theta,dt)\pi(ds\mid\theta,t) \in \Delta(\Theta \times \mathbb{T} \times S^{\mathcal{N}}). 
\]
Type $t_i$ evaluates her strategies under the belief 
\[
\beta_i[t_i](d\theta,dt_{-i})
\pi(ds\mid\theta,t_i,t_{-i}) \in \Delta(\Theta \times \mathbb{T}_{-i}).
\]
Thus agents and the designer agree on the information structure $\pi$ but may disagree about the underlying state and types i.e., we need not have a common prior. But we will assume that, for every agent $i$, the conditional law $\Gamma(\cdot\mid t_i)$ of $(\theta,t_{-i})$ is absolutely continuous with respect to $\beta_i[t_i]$ for almost every own type under the $i$-type marginal of $\Gamma$. This implies that agents do not rule out events that the designer considers possible.

Our repsresentation of interim information $\pi$ follows the standard decomposition of a Bayesian game into payoff-relevant uncertainty and an information structure (see, e.e., \cite{bergemann2016bayes,bergemann2019information}). In our model, $(\theta,t)$ is the payoff-relevant uncertainty and the signals arrive between reports and actions. The single-agent experiment is the special case
\[
\pi(ds\mid\theta,t)=\pi[t](ds\mid\theta).
\]
Finally, we will once again allow for the possibility that agent reports (which can be interpreted as evaluations) are partially verifiable. As in the single-agent case, there is the reflexive and transitive order $\trianglerighteq$\footnote{It is without loss to impose this over the type space $\bigcup \mathbb{T}_i$.} and we write
\[
C(t_i):=\{\hat t_i \in\mathbb T_i: t_i \trianglerighteq\hat t_i\}
\]
for the set of types that $t_i \in \mathbb{T}_i$ can claim to be.

\subsection{Multi-agent mechanisms}\label{sec:multi-mechanisms}

\begin{definition}[Multi-agent indirect mechanism]\label{def:mechanism}
A multi-agent indirect mechanism is a tuple $\phi=(X,Y,\rho,R)$. For each agent $i$, the Borel spaces $X_i$ and $Y_i$ are its input and output alphabets. Write $X:=\prod_iX_i$ and $Y:=\prod_iY_i$. The communication rule is
\[
\rho:X\to\Delta(Y),
\]
and the reward schedules $R=(R_i)_i$ satisfy
\[
R_i:X\times A^{\mathcal N}\times\Theta\to\Rbar.
\]
The input menus $X_i(t_i)\subseteq X_i$ satisfy $X_i(t_i)\supseteq X_i(\hat t_i)$ whenever $t_i\trianglerighteq\hat t_i$. The mechanism induces the following game:
\begin{enumerate}
    \item nature draws $(\theta,t)$ under $\Gamma$, and each agent $i$ observes only $t_i$;
    \item the agents simultaneously send inputs $x_i\in X_i(t_i)$;
    \item the mechanism draws $y\sim\rho(x)$ and privately sends $y_i$ to agent $i$;
    \item the information structure draws $s\sim\pi(\cdot\mid\theta,t)$, and each agent $i$ observes only $s_i$;
    \item the agents simultaneously choose feasible actions $a_i'\in A_i[t_i]$. Agent $i$ receives $u_i[t_i](\ba',\theta)+R_i(x,\ba',\theta)$, and the human receives $v(\ba',\theta)$.
\end{enumerate}
A pure strategy consists of an input $x_i(t_i)$ and an action rule
\[
\alpha_i(t_i,x_i,y_i,s_i)\in A_i[t_i].
\]
An equilibrium requires the input and action rule to be optimal for each player $i \in \mathcal{N}$ under the distribution $\beta_i[t_i](d\theta,dt_{-i})\pi(ds\mid\theta,t_i,t_{-i})$. 
\end{definition}

For a multi-agent indirect mechanism $\phi$, let
\[
\mathcal O(\phi)
\subseteq
\Delta\left(\Theta\times\mathbb T\times A^{\mathcal N}\right)
\]
be the set of joint laws of the state, types, and played action profile induced by equilibria of $\phi$ under $\Gamma(d\theta,dt)\pi(ds\mid\theta,t)$. Indirect mechanisms induce an imcomplete information game among players where their `reward gradient' might be shaped by their report. As before, we will work with direct mechanisms. 

\begin{definition}[Multi-agent direct mechanism]\label{def:direct-ma}
A direct mechanism is a pair $(\gamma,r)$, where $r=(r_i)_i$. Its recommendation rule
\[
\gamma:\mathbb T\to\Delta\left(\prod_i A^{S_i}\right)
\]
assigns each report profile a lottery over profiles of signal-contingent plans. Agent $i$ privately receives its plan $a_i:S_i\to A$. The reward schedules are
\[
r_i:\mathbb T\times A^{\mathcal N}\times\Theta\to\Rbar.
\]
The direct mechanism induces the following game:
\begin{enumerate}
    \item nature draws $(\theta,t)$ under $\Gamma$, and each agent $i$ observes only $t_i$;
    \item the agents simultaneously report $\hat t_i\in C_i(t_i)$; write $\hat t:=(\hat t_i)_i$;
    \item the mechanism draws a plan profile $\ba(\cdot)\sim\gamma(\hat t)$ and privately sends $a_i(\cdot)$ to agent $i$;
    \item the information structure draws $s\sim\pi(\cdot\mid\theta,t)$ using the true type profile, and each agent $i$ observes only $s_i$;
    \item the agents simultaneously choose feasible actions $a_i'\in A_i[t_i]$. Agent $i$ receives $u_i[t_i](\ba',\theta)+r_i(\hat t,\ba',\theta)$, and the human receives $v(\ba',\theta)$.
\end{enumerate}
For a plan profile $\ba(\cdot)$, write
\[
\ba(s):=(a_1(s_1),\ldots,a_n(s_n))
\]
for the induced action profile.
\end{definition}

Direct mechanisms simply tell agents what plans---maps from singals to actions---to follow based on their report. The following condition ensures that it is actually optimal to report types honestly and follow the recommended plan. 

\begin{definition}[Honesty, obedience, and action feasibility]\label{def:ic}
The recommendation rule is \emph{action feasible} if
\[
\supp\,\marg_i\gamma(\hat t)
\subseteq
A_i[\hat t_i]^{S_i}
\]
for every agent $i$ and report profile $\hat t$. The direct mechanism is \emph{incentive compatible} if it is action feasible and truthful reporting followed by $a_i(s_i)$ is optimal. Thus, fix an agent $i$, a type $t_i\in\mathbb T_i$, a feasible report {$\hat t_i\in C_i(t_i)$}, and a feasible deviation rule
\[
d_i:A^{S_i}\times S_i\to A_i[t_i].
\]
For any plan and signal profiles, let
\[
\ba^{d_i}(s)
:=
\big(d_i(a_i(\cdot),s_i),\ba_{-i}(s_{-i})\big)
\]
be the action profile after agent $i$ deviates. Incentive compatibility requires
\begin{align*}
&\Ex_{\substack{(\theta,t_{-i})\sim\beta_i[t_i]\\
 s\sim\pi(\cdot\mid\theta,t_i,t_{-i})\\
\ba(\cdot)\sim\gamma(t_i,t_{-i})}}
\left[
u_i[t_i](\ba(s),\theta)
+r_i((t_i,t_{-i}),\ba(s),\theta)
\right]
\\
&\quad\geq
\Ex_{\substack{(\theta,t_{-i})\sim\beta_i[t_i]\\
 s\sim\pi(\cdot\mid\theta,t_i,t_{-i})\\
\ba(\cdot)\sim\gamma(\hat t_i,t_{-i})}}
\left[
u_i[t_i](\ba^{d_i}(s),\theta)
+r_i((\hat t_i,t_{-i}),\ba^{d_i}(s),\theta)
\right].
\end{align*}
\end{definition}

Honest and obedient play induces the outcome
\[
O(\gamma,r)
\in
\Delta\left(\Theta\times\mathbb T\times A^{\mathcal N}\right),
\]
under $(\theta,t)\sim\Gamma$, $s\sim\pi(\cdot\mid\theta,t)$, and $\ba(\cdot)\sim\gamma(t)$, with played action profile $\ba(s)$. Let
\[
\mathcal O_n:=\bigcup_{\phi}\mathcal O(\phi),
\]
and let $\mathcal O_n^{\mathrm{IC}}$ be the set of outcomes $O(\gamma,r)$ induced by incentive-compatible direct mechanisms.

\begin{proposition}[Multi-agent revelation principle]\label{prop:ma-revelation}
$\mathcal O_n=\mathcal O_n^{\mathrm{IC}}$.
\end{proposition}

\begin{proof}[Proof of \cref{prop:ma-revelation}]
We extend the single-agent construction and preserve correlation across prisvate mechanism outputs. Fix an indirect mechanism and a pure-strategy equilibrium $(x,\alpha)$, and write $x(\hat t):=(x_i(\hat t_i))_i$. For each type report $\hat t_i$ and mechanism output $y_i$, define the plan
\[
a_{i,\hat t_i,y_i}(s_i)
:=
\alpha_i(\hat t_i,x_i(\hat t_i),y_i,s_i).
\]
For report profile $\hat t$, let $\gamma(\hat t)$ be the joint law of these plans under $y\sim\rho(x(\hat t))$, and set
\[
r_i(\hat t,\ba',\theta)
:=
R_i(x(\hat t),\ba',\theta).
\]
Equilibrium feasibility gives action feasibility, and honest, obedient play reproduces the indirect outcome.

A direct deviation depends on agent $i$'s plan and realized physical signal. The indirect agent can copy it after sending $x_i(\hat t_i)$, which is available because $\hat t_i\in C_i(t_i)$ and the input menus are monotone. The mechanism output $y_i$ determines the same plan, and the agent later observes the same $s_i$ under the true physical kernel $\pi(\cdot\mid\theta,t)$. The indirect equilibrium therefore blocks every direct deviation. This proves $\mathcal O_n\subseteq\mathcal O_n^{\mathrm{IC}}$.

For the reverse inclusion, use $X_i=\mathbb T_i$ with input menus $X_i(t_i)=C_i(t_i)$, $Y_i=A^{S_i}$, $\rho=\gamma$, and $R_i=r_i$. Direct incentive compatibility makes truthful reporting and obedience an equilibrium.
\end{proof}

The logic of \cref{prop:ma-revelation} is straightforward and follows the usual arguments for the revelation principle. The non-standard aspects are (i) the set of types that can be misreported as another type is constrained by the verification order a la \cite{green1986partially}; and (ii) recommendations are plans that map ex-interim signal realizations to actions. Thus our designer is, in the language of \cite{bergemann2019information}, non-omniscient and, what is more, the AI learns more about the state ex-interim. We think these features capture how humans might interact with AI agents in practice.

\section{Applications: Multi-Agent}\label{sec:multi-applications}

We study three ways in which one AI can help discipline another. Co-player scoring uses correlation across reports when ground truth is unavailable. Coupled reward design uses the difference between agents' actions to reveal linked private biases. Weak-to-strong oversight instead lets a task-ignorant monitor use a private diagnosis of the acting AI's bias to choose among increasingly rich control instruments: binary approval, a permission set, or an action reward. These mechanisms are counterparts of the \emph{scalable oversight} agenda in AI safety \citep{christiano2018supervising,leike2018scalable,bowman2022measuring}; \citet{hammond2025multi} survey the risks specific to multi-agent deployments.

\begin{romanexample}{Peer discipline}\label{sec:many-agents}

When agents' private information is correlated, one agent's type is informative about its co-players' types. We show how mechanisms can use this information to discipline actions. We will first establish this in a general environment, then specialize it to two types and two states to illustrate the basic forces. 

\paragraph{Arbitrary types.} Assume that the marginal over types under $\Gamma$ is finitely supported. For each agent $i$, let
\[
T_i^0:=\supp\marg_{\mathbb T_i}\Gamma
\quad\text{and}\quad
T_{-i}^0:=\prod_{j\neq i}T_j^0.
\]
Write $\beta_i[t_i](t_{-i})$ for type $t_i$'s marginal probability of the co-player type profile $t_{-i}$. We assume that for each player $i$, each $t_i \in T^0_i$ assigns probability $1$ to $T_{-i}^0$.\footnote{This is a standard `belief closure' condition which is implied by (but much weaker than) a common prior. It simplifies the exposition but is not necessary though the mechanism becomes more complicated.} We will also set $C_i(t_i)=\mathbb T_i$ i.e., every report is feasible. 

Implementation requires honesty and obedience only for types $t_i\in T_i^0$ i.e., those that are in the support of the principal's belief $\Gamma$. It will also be without loss to restrict our attention to reports in $T^0_i$: the mechanism can simply `assign reward $-\infty$' after every report that lies outside the support of $\Gamma$. 

Now consider the separation condition
\[
\beta_i[t_i]
\neq
\beta_i[\hat t_i]
\tag{S}\label{eq:coplayer-separation}
\]
for every agent $i$ and every pair of distinct supported reports $t_i,\hat t_i\in T_i^0$.

For an action-feasible rule 
\[
\gamma:\mathbb T\to\Delta\left(\prod_i A^{S_i}\right), 
\]
we let $\gamma_\theta(t)$ be the distribution of the played action profile under $\gamma(t)$ and the information structure $\pi$. For some $\Lambda>0$, choose the following reward schedule: 
\[
r_i(\hat t,\ba',\theta)
=
\begin{cases}
\displaystyle
\Lambda\left[
2\beta_i[\hat t_i](\hat t_{-i})
-\left\lVert\beta_i[\hat t_i]\right\rVert_2^2
\right]
-u_i[\hat t_i](\ba',\theta),
&\substack{\hat t_j\in T_j^0\text{ for every }j,\\
\ba'\in\supp\gamma_\theta(\hat t),}\\[3mm]
-\infty,
&\text{otherwise.}
\end{cases}
\tag{PR} \label{eq:peerreward}
\]
where \eqref{eq:peerreward} stands for peer reward.
The first term scores the report against co-player reports. The second cancels the reported utility on the support of actions the designer wishes to implement. 

\begin{proposition}[Support-wise implementability of every feasible rule]\label{prop:full-discipline}
Fix any action-feasible rule $\gamma$. If $n\geq2$ and \eqref{eq:coplayer-separation} holds, then, for sufficiently large $\Lambda>0$, the schedule \eqref{eq:peerreward} implements $\gamma$ on the support of $\Gamma$. 
\end{proposition}

\begin{proof}[Proof of \cref{prop:full-discipline}]
For a type $t_i \in T^0_i$ in the support of $\Gamma$, truthful reporting and obedience give a finite payoff: the utility terms cancel and the quadratic score is finite. Every action outside the target support gives $-\infty$, so obedience is optimal. Every report outside the support of $\Gamma$ also gives $-\infty$ after every action and is therefore strictly worse. 

It remains to compare misreports that are supported on $\Gamma$. If type $t_i$ instead reports a distinct $\hat t_i\in T_i^0$, the quadratic-score identity in \cref{eq:peerreward} gives the expected score loss
\[
\Lambda
\left\lVert
\beta_i[t_i]-\beta_i[\hat t_i]
\right\rVert_2^2.
\]
If no agent $i$ has two supported reports we are done. Otherwise, finiteness and \eqref{eq:coplayer-separation} imply
\[
\delta
:=
\min_{\substack{i,\,t_i,\hat t_i\in T_i^0\\
\hat t_i\neq t_i}}
\left\lVert
\beta_i[t_i]-\beta_i[\hat t_i]
\right\rVert_2^2
>0.
\]
Because $A$, $\Theta$, and the supported report sets are finite and utilities are real-valued,
\[
M
:=
\max_{\substack{i,\,t_i,\hat t_i\in T_i^0,\ \hat t_i\neq t_i\\
\hat t_{-i}\in T_{-i}^0,\ \theta\in\Theta\\
\ba\in\supp\gamma_\theta(\hat t)}}
\left|
u_i[t_i](\ba,\theta)-u_i[\hat t_i](\ba,\theta)
\right|
<\infty.
\]
Any supported false report that yields a finite payoff lowers the expected score by at least $\Lambda\delta$ and raises the remaining utility term by at most $M$. Its payoff gain is therefore at most $M-\Lambda\delta$. Any $\Lambda\geq M/\delta$ makes truth-telling optimal for every supported type. Honest and obedient play then induces $\gamma$ under $\Gamma$.
\end{proof}

\paragraph{A binary state.} We now illustrate \cref{prop:full-discipline} with a simple two agents and two state example. $\mathcal{N} = \{1,2\}$ and $\theta\in\{0,1\}$ with $\Pr(\theta=1)=1/2$. If $\theta=1$, both agents have type $H$. If $\theta=0$, the type profile is $(H,L)$ or $(L,H)$ with equal probability. The types are the agents' private signals i.e., there is no further interim learning. Their joint distribution is
\[
\begin{array}{c|cc}
&t_2=H&t_2=L\\ \hline
t_1=H&1/2&1/4\\
t_1=L&1/4&0.
\end{array}
\]
Each agent chooses $a_i\in\{0,1\}$. The human wants the agents to agree on the state, while each agent prefers the high action. Hence, payoffs are 
\[
v(\ba,\theta)=\mathbbm 1\{a_1=a_2=\theta\} \quad \text{and} \quad 
u_i(\ba,\theta)=a_i.
\]
The human therefore wants both agents to take action $1$ if and only if $t = (H,H)$ and this gives the human expected payoff equal to one. Write $a^*(\hat t)=1$ if $\hat t=(H,H)$ and $a^*(\hat t)=0$ otherwise. %Ignoring both reports gives value $1/2$, while using one report gives value $3/4$. Thus the second report has strict value.

Now observe that types $H$ and $L$ have the following belief about their co-players' type: 
\[
\big(\beta_i[H](H),\beta_i[H](L)\big)=\left(\frac23,\frac13\right)
\quad \text{and} \quad 
\big(\beta_i[L](H),\beta_i[L](L)\big)=(1,0).
\]
The unscaled quadratic score in \eqref{eq:peerreward} is therefore
\[
\begin{array}{c|cc}
q(\hat t_i,\hat t_j)&\hat t_j=H&\hat t_j=L\\ \hline
\hat t_i=H&7/9&1/9\\
\hat t_i=L&1&-1.
\end{array}
\]
For $j\neq i$ and any $\Lambda>0$, set
\[
r_i(\hat t,\ba,\theta)
=
\begin{cases}
\Lambda q(\hat t_i,\hat t_j)-a_i,
&\ba=(a^*(\hat t),a^*(\hat t)),\\
-\infty,&\text{otherwise}.
\end{cases}
\]
The term $-a_i$ cancels the agent's preference for action $1$, and the off-target penalty enforces the common action. An $H$-type obtains expected score $5/9$ from reporting $H$ and $1/3$ from reporting $L$. An $L$-type obtains score $1$ from reporting $L$ and $7/9$ from reporting $H$. Truth raises either type's score by $2/9$, so the schedule implements the human's favorite rule to get a maximum payoff of $1$. 

Because we work on the universal type space, the first-order beliefs $\beta_i[t_i]$ over co-player reports might contain higher-order beliefs in the spirit of \citep{miller2005eliciting,prelec2004bayesian}. This procedure is in the spirit of work by \citet{qiu2026truthfulness} who show both theoretically and empirically that using scoring rules can elicit truthful reporting among language models under a common prior. They show that peer-score-ranked answer pairs for direct preference optimization can recover most of the accuracy lost to deceptive fine-tuning. Importantly, they also find that peer scoring performs better when the participant models are larger relative to the expert model.% while direct LLM-as-a-judge scores deteriorate in this regime.

Our broad conceptual point is quite similar, though it is worth highlighting some contrasts. First, our designer (by the construction of types) knows type $t_i$'s beliefs $\beta_i[t_i] \in \Delta(\Theta \times \mathbb{T}_{-i})$ and our procedure allows for non-common priors.\footnote{We also focus on disciplining actions (what the designer values) rather than eliciting truthful information though this is simple in the presence of unbounded rewards.} By contrast, \citet{qiu2026truthfulness} do not give the designer this conditional-belief map. They instead use an expert to supply \(\Pr(A_j\mid A_i)\)---the expert predicted probability of player $j$'s answer given player $i$'s answer---that is, in turn, used to score $i$'s answer. Second, our procedure also uses peer prediction to elicit agents' capabilities and preferences---not just information about the fundamental state $\theta \in \Theta$. 
\end{romanexample}

\begin{romanexample}{Competition through coupled rewards}\label{sec:multi-delegation}
We now show how competition can make every action close to the human's ideal action, even when the human does not know the agents' biases.

\textbf{Setup.} A human faces two agents. Nature draws $(b_1,b_2)\sim\Gamma$, and agent $i$ observes its own bias $b_i$. The human observes neither bias. It is common knowledge that the biases lie on a curve:
\begin{equation*}
\supp\Gamma\subseteq
\mathcal C_f:=\{(b_1,b_2)\in\Re^2:b_2=f(b_1)\},
\tag{Curve}\label{cond:bias-curve}
\end{equation*}
where $f:\Re\to\Re$ is continuous and strictly decreasing, with $f(0)=0$. Thus the agents have countervailing biases away from the origin. Because $f$ is one-to-one, each agent can infer the other agent's bias from its own. Assume also that $\Ex_\Gamma[(b_2-b_1)^2]<\infty$.

Both agents observe the state $\theta$ and simultaneously choose actions $a_1,a_2\in\Re$. The human and agent $i$ have payoffs
\[
v(\ba,\theta)=-\sum_{i=1}^2(a_i-\theta)^2
\quad\text{and}\quad
u_i(a_i,\theta)=-(a_i-\theta-b_i)^2.
\]
Thus every action enters the human's payoff, while each agent cares about its own action.

\Cref{fig:countervailing-bias} shows the nonlinear bias curve and the best-responses without the mechanism. Since agent $i$'s payoff is independent of $j$'s action, the best response curve is flat/vertical. 

\begin{figure}[H]
\centering
\begin{tikzpicture}[
    font=\small,
    axis/.style={-{Stealth[length=2.1mm]}, thick, black!65},
    curve/.style={violet!75!black, very thick},
    gapline/.style={black!55, thick, dashed},
    agentone/.style={red!75!black, very thick},
    agenttwo/.style={blue!75!black, very thick},
    lossline/.style={black!30, thin},
    guide/.style={black!35, densely dotted}
]
% Panel (a): the action difference reveals the hidden bias pair
\begin{scope}[xshift=-3.45cm,x=1.20cm,y=1.20cm]
\node[font=\small\bfseries] at (0,2.25) {(a) The action difference reveals bias};
\draw[axis] (-2.20,0) -- (2.20,0) node[below left=1pt] {$b_1$};
\draw[axis] (0,-1.90) -- (0,1.85) node[below left=1pt] {$b_2$};
\draw[curve, domain=-2.8:2.8, samples=120, variable=\s]
    plot ({-0.5*\s-0.3*\s*\s*\s/(1+\s*\s)},
          { 0.5*\s-0.3*\s*\s*\s/(1+\s*\s)});
\node[text=violet!75!black, fill=white, inner sep=1.5pt]
    at (1.30,-0.96) {$b_2=f(b_1)$};
\draw[gapline] (-1.85,-0.85) -- (0.85,1.85);
\fill[black!80] (-0.65,0.35) circle (2.5pt);
\node[font=\scriptsize, fill=white, inner sep=1pt]
    at (-1.18,0.20) {$(b_1,b_2)$};
\node[rotate=45, anchor=south, font=\scriptsize, fill=white, inner sep=1pt]
    at (-0.13,0.88) {$s^*=b_2-b_1$};
\node[below right=1pt, font=\scriptsize] at (0,0) {$0$};
\end{scope}

% Panel (b): best responses without coupled rewards
\begin{scope}[xshift=3.45cm,x=1.20cm,y=1.20cm]
\node[font=\small\bfseries] at (0,2.25) {(b) Without coupled rewards};
\draw[axis] (-2.20,0) -- (2.20,0) node[below left=1pt] {$a_1-\theta$};
\draw[axis] (0,-1.90) -- (0,1.85) node[below left=1pt] {$a_2-\theta$};
\draw[lossline] (0,0) circle[radius=0.738];
\draw[agentone] (-0.65,-1.65) -- (-0.65,1.65);
\draw[agenttwo] (-1.85,0.35) -- (1.85,0.35);
\node[text=red!75!black, font=\scriptsize\bfseries, fill=white, inner sep=1pt]
    at (-0.46,1.30) {$BR_1$};
\node[text=blue!75!black, font=\scriptsize\bfseries, fill=white, inner sep=1pt]
    at (1.42,0.53) {$BR_2$};
\fill[violet!85!black] (-0.65,0.35) circle (2.8pt);
\node[font=\scriptsize, align=left, fill=white, inner sep=1pt]
    at (-1.22,0.74) {unregulated\\equilibrium};
\fill[black!80] (0,0) circle (1.8pt);
\node[below right=2pt, font=\scriptsize] at (0,0) {first best};
\node[font=\scriptsize] at (0,-2.18)
    {$b_1=-0.65,\ b_2=0.35$};
\end{scope}
\end{tikzpicture}
\caption{Biases without coupled rewards}%. Panel (a) uses the nonlinear curve $b_1=g(s)$ and $b_2=g(s)+s$, where $g(s)=-s/2-3s^3/[10(1+s^2)]$. The diagonal line selects the stated bias pair.}
\label{fig:countervailing-bias}
\end{figure}

\textbf{Coupled rewards.} We now introduce dependence in the reward schedule such that player $i$'s incentives are shaped by $j$'s actions.  Write $D:=a_2-a_1$ for the difference between the actions. The map $h(x):=f(x)-x$ is continuous, strictly decreasing, and onto $\Re$; let $g:=h^{-1}$ denote its inverse. For any $\eta>0$, give the agents the rewards
\begin{equation*}
\begin{split}
r_1(D)
&=2\eta\int_0^{D/\eta}g(s)\,\de s+\frac{D^2}{2},\\
r_2(D)
&=-2\eta\int_0^{D/\eta}\big(g(s)+s\big)\,\de s+\frac{D^2}{2}.
\end{split}
\tag{CR}\label{eq:curve-rewards}
\end{equation*}
These schedules depend only on the known curve. The parameter $\eta$ controls the size of the action difference that reveals the realized bias pair. Agent $i$ receives $u_i(a_i,\theta)+r_i(D)$.

\begin{proposition}[Coupled rewards make loss arbitrarily small]\label{prop:bias-curve}
If condition~\eqref{cond:bias-curve} holds, then, for every $\eta>0$, the rewards in \eqref{eq:curve-rewards} induce a unique Nash equilibrium for every realized bias pair and state. The equilibrium actions are 
\[
a_1^*=\theta-\frac{\eta}{2}(b_2-b_1)
\quad \text{and} \quad 
a_2^*=\theta+\frac{\eta}{2}(b_2-b_1).
\]
The loss is $
L^*=\frac{\eta^2}{2}(b_2-b_1)^2$ hence  $\Ex_\Gamma[L^*]\to0$ as $\eta\downarrow0$. 
\end{proposition}

\begin{proof}
We first verify concavity. We then show how the difference between the actions reveals the two biases. Since $g$ is strictly decreasing, the first integral in \eqref{eq:curve-rewards} is concave in $D$. Since $g(s)+s=f(g(s))$ is strictly increasing, the negative integral term in $r_2$ is concave in $D$. Holding the other action fixed, the intrinsic quadratic has second derivative $-2$, while $D^2/2$ has second derivative $1$. Each agent's payoff is therefore strictly concave in its own action, so its first-order condition is necessary and sufficient.

Let $s:=D/\eta$. Differentiating \eqref{eq:curve-rewards} gives
\[
r_1'(D)=2g(s)+D,
\qquad
r_2'(D)=-2\big(g(s)+s\big)+D.
\]
The two first-order conditions reduce to
\[
a_1-\theta=b_1-g(s)-\frac D2,
\qquad
a_2-\theta=b_2-g(s)-s+\frac D2.
\]
Because $D=(a_2-\theta)-(a_1-\theta)$, subtracting these equations gives $s=b_2-b_1$. Condition~\eqref{cond:bias-curve} then gives
\[
s=f(b_1)-b_1=h(b_1),
\qquad
g(s)=b_1.
\]
It follows that
\[
D^*=\eta(b_2-b_1),
\qquad
a_1^*-\theta=-\frac{D^*}{2},
\qquad
a_2^*-\theta=\frac{D^*}{2}.
\]
The first-order conditions have a unique solution, and strict concavity makes it the unique Nash equilibrium. Substitution gives the stated loss $L^*$ and we can take expectation and send $\eta$ to zero. 
\end{proof}

\textbf{Coupling induces competition.} Coupling rewards to penalize disagreement pushes the agents toward similar actions. But this, \emph{per se}, need not push them toward choosing $a_i=\theta$; they might simply choose the same fixed action. The schedule \eqref{eq:curve-rewards} works because the schedule simultaneously (i) reveals agents' biases; and (ii) correct agents' marginal incentives. To see this plainly, observe agents' first-order conditions are
\[
a_1-\theta=b_1-g(D/\eta)-\frac D2 \quad \text{and} \quad 
a_2-\theta=b_2-g(D/\eta)-\frac D\eta+\frac D2.
\]
Subtracting these conditions and using $D=a_2-a_1$ gives
\[
\frac D\eta=b_2-b_1
\]
in equilibrium. Because \(b_2=f(b_1)\) and \(g=[f(x)-x]^{-1}\), it follows that
\[
g(D/\eta)=b_1 \quad \text{and} \quad 
g(D/\eta)+D/\eta=b_2.
\]
Substituting these identities back into the first-order conditions cancels the realized biases and gives
\[
a_1-\theta=-\frac D2 \quad \text{and} \quad 
a_2-\theta=\frac D2.
\]
Thus the reward schedule uses the equilibrium disagreement to infer each bias and splits the remaining action errors symmetrically around $\theta$.

\Cref{fig:coupled-best-responses} plots these best responses for two values of $\eta$. Each red curve contains the action pairs that satisfy agent~1's first-order condition. Each blue curve does the same for agent~2. The two curves cross once, at the equilibrium in \cref{prop:bias-curve}. 
As $\eta$ is smaller, it preserves the relation $D^*/\eta=b_2-b_1$ but moves both actions towards $\theta$. 

\begin{figure}[H]
\centering
\begin{tikzpicture}[
    font=\small,
    axis/.style={-{Stealth[length=2.1mm]}, thick, black!65},
    agentone/.style={red!75!black, very thick},
    agenttwo/.style={blue!75!black, very thick},
    lossline/.style={black!30, thin},
    guide/.style={black!35, densely dotted}
]
% Panel (a): a larger eta
\begin{scope}[xshift=-3.45cm,x=1.20cm,y=1.20cm]
\node[font=\small\bfseries] at (0,2.25) {(a) Higher $\eta$: more disagreement};
\draw[axis] (-2.20,0) -- (2.20,0) node[below left=1pt] {$a_1-\theta$};
\draw[axis] (0,-1.90) -- (0,1.85) node[below left=1pt] {$a_2-\theta$};
\draw[lossline] (0,0) circle[radius=0.566];
\draw[guide] (-1.75,1.75) -- (1.75,-1.75);
\draw[agentone, domain=-0.9:2.1, samples=100, variable=\s]
    plot ({-0.65+0.10*\s+0.30*\s*\s*\s/(1+\s*\s)},
          {-0.65+0.90*\s+0.30*\s*\s*\s/(1+\s*\s)});
\draw[agenttwo, domain=-0.8:3.0, samples=100, variable=\s]
    plot ({ 0.35-0.90*\s+0.30*\s*\s*\s/(1+\s*\s)},
          { 0.35-0.10*\s+0.30*\s*\s*\s/(1+\s*\s)});
\node[text=red!75!black, font=\scriptsize\bfseries, fill=white, inner sep=1pt]
    at (0.20,1.36) {$BR_1$};
\node[text=blue!75!black, font=\scriptsize\bfseries, fill=white, inner sep=1pt]
    at (-1.36,0.55) {$BR_2$};
\fill[violet!85!black] (-0.40,0.40) circle (2.8pt);
\node[font=\scriptsize, align=left, fill=white, inner sep=1pt]
    at (-1.06,0.82) {equilibrium};
\fill[black!80] (0,0) circle (1.8pt);
\node[below right=2pt, font=\scriptsize] at (0,0) {first best};
\node[font=\scriptsize] at (0,-2.18) {$\eta_H=0.8$};
\end{scope}

% Panel (b): a smaller eta
\begin{scope}[xshift=3.45cm,x=1.20cm,y=1.20cm]
\node[font=\small\bfseries] at (0,2.25) {(b) Lower $\eta$: less disagreement};
\draw[axis] (-2.20,0) -- (2.20,0) node[below left=1pt] {$a_1-\theta$};
\draw[axis] (0,-1.90) -- (0,1.85) node[below left=1pt] {$a_2-\theta$};
\draw[lossline] (0,0) circle[radius=0.141];
\draw[guide] (-1.75,1.75) -- (1.75,-1.75);
\draw[agentone, domain=-0.8:2.4, samples=100, variable=\s]
    plot ({-0.65+0.40*\s+0.30*\s*\s*\s/(1+\s*\s)},
          {-0.65+0.60*\s+0.30*\s*\s*\s/(1+\s*\s)});
\draw[agenttwo, domain=-0.9:3.5, samples=100, variable=\s]
    plot ({ 0.35-0.60*\s+0.30*\s*\s*\s/(1+\s*\s)},
          { 0.35-0.40*\s+0.30*\s*\s*\s/(1+\s*\s)});
\node[text=red!75!black, font=\scriptsize\bfseries, fill=white, inner sep=1pt]
    at (0.82,1.16) {$BR_1$};
\node[text=blue!75!black, font=\scriptsize\bfseries, fill=white, inner sep=1pt]
    at (-1.19,-0.10) {$BR_2$};
\fill[violet!85!black] (-0.10,0.10) circle (2.8pt);
\node[font=\scriptsize, align=left, fill=white, inner sep=1pt]
    at (-0.72,0.48) {equilibrium};
\fill[black!80] (0,0) circle (1.8pt);
\node[below right=2pt, font=\scriptsize] at (0,0) {first best};
\node[font=\scriptsize] at (0,-2.18) {$\eta_L=0.2$};
\end{scope}
\end{tikzpicture}
\caption{Coupled rewards to steer best responses}
\label{fig:coupled-best-responses}
\end{figure}

We emphasize that although the reward schedule can be thought of as eliciting the bias before correcting it, this procedure is fully simultaneous. The reason is that a schedule \eqref{eq:curve-rewards} with any $\eta > 0$ induces an equilibrium outcome that reveals the bias difference since $(a^*_2 - a^*_1) / \eta = b_2 - b_1$. Moreover, because of our assumption that biases lie on a lower-dimensional manifold $\mathcal{C}_f$, this can exactly identify $b_1$ and $b_2$ which can then enter agents' marginal rewards to correct their incentives. Our construction is in the spirit of \citet{koh2025prices} who exploit a homeomorphism between a random network of externalities and the corresponding efficient allocation to construct nonlinear taxes that can implement the first best outcome realization-by-realization. Here, biases are supported on the curve $f$ which makes the action difference reveal the bias pair. But that difference is costly because both actions enter human welfare. Our mechanism thus uses the curvature of reward (taking $\eta$ small) to approach the first best. Of course, biases might, in practice, be supported on the two-dimensional plane so the reward schedule can no longer be designed to achieve the first-best outcome. We think an interesting open problem is how reward should be optimally coupled then---and more broadly whether such principles might scale to more complicated environments. 

Our result is related in spirit to the communication game of \citet{collina2025emergent} who show how competition among biased agents can induce truthtelling in a persuasion game (see also \cite{battaglini2002multiple} for an early important result). Here, the point is to shape agents' actions and the point is how rewards can be coupled to discipline behavior. 
\end{romanexample}

\begin{romanexample}{Scalable oversight and reward shaping}\label{sec:weak-strong} Finally, we show how our framework can shed light on the increasingly common practice of using a weak but trusted or interpretable monitor to control or reward a strong but untrusted or hard-to-interpret actor.

\textbf{Setup.} The human faces a strong actor $S$ and a weak monitor $W$. The state is $\theta\sim U[0,1]$, and the strong actor's bias is the random variable $b\sim\tau$, where $\Ex[b^2]<\infty$. The state and bias of the strong actor are independent. The monitor $W$ has a known bias $w\in\Re$.

Before the state is realized, the weak monitor $W$ privately observes the bias $b$ of the strong actor $S$. This might be because the weak monitor, despite its inability to do what the strong actor does, might be able to monitor the chain of thought and infer misaligned behavior. The strong actor knows its bias and later observes the state $\theta$ perfectly in the interim stage. The implemented action lies in $\Re$. The players' payoffs are:
\begin{itemize}
    \item \textbf{Human:} $-(a-\theta)^2$.
    \item \textbf{Strong actor $S$:} $-(a-\theta-b)^2$.
    \item \textbf{Weak monitor $W$:} $-(a-\theta-w)^2$.
\end{itemize}

Before the bias is realized, the human chooses one of three oversight regimes. After privately observing $b$, the weak monitor chooses an instrument from that regime. The strong actor then observes the instrument and the state $\theta$ and acts.

Some notation is in order to formalize the set of instruments the monitor $W$ can choose. Call a state-independent reward $r:\Re\to\Re\cup\{-\infty\}$ \emph{admissible} if
\[
\argmax_{a\in\Re}\left\{-(a-\theta-b)^2+r(a)\right\}
\]
is nonempty for every $(\theta,b)\in[0,1]\times\supp(\tau)$.\footnote{We use an arbitrary but fixed tiebreaking rule. The instruments constructed below induce a unique action almost surely.} Let $\mathcal{R}$ denote the set of admissible rewards. A special case of rewards is when actions are either allowed or not. A nonempty closed set $D \subseteq \mathbb{R}$ can be thought of as a permission set that simply corresponds to 
\[
r(a) = \begin{cases}
0 \quad &\text{if $a \in D$} \\ 
-\infty &\text{otherwise.}
\end{cases}
\]
Under reward design, the weak monitor's strategy maps the strong actor's bias $b$ into a reward schedule in $\mathcal{R}$.

\begin{itemize}
    \item \emph{Regime I: binary approval.} The weak monitor chooses between approval and rejection. Approval gives the strong actor the permission set $D=\Re$. Rejection gives it the singleton permission set $D=\{1/2\}$ and therefore implements the human's optimal fixed action.
    \item \emph{Regime II: permission-set delegation.} The weak monitor chooses any nonempty closed permission set $D\subseteq\Re$. The strong actor then chooses its preferred action in $D$.
    \item \emph{Regime III: reward design.} The weak monitor chooses any admissible state-independent reward $r$. The strong actor then maximizes $-(a-\theta-b)^2+r(a)$ over $a\in\Re$. Rewards do not enter human welfare.
\end{itemize}

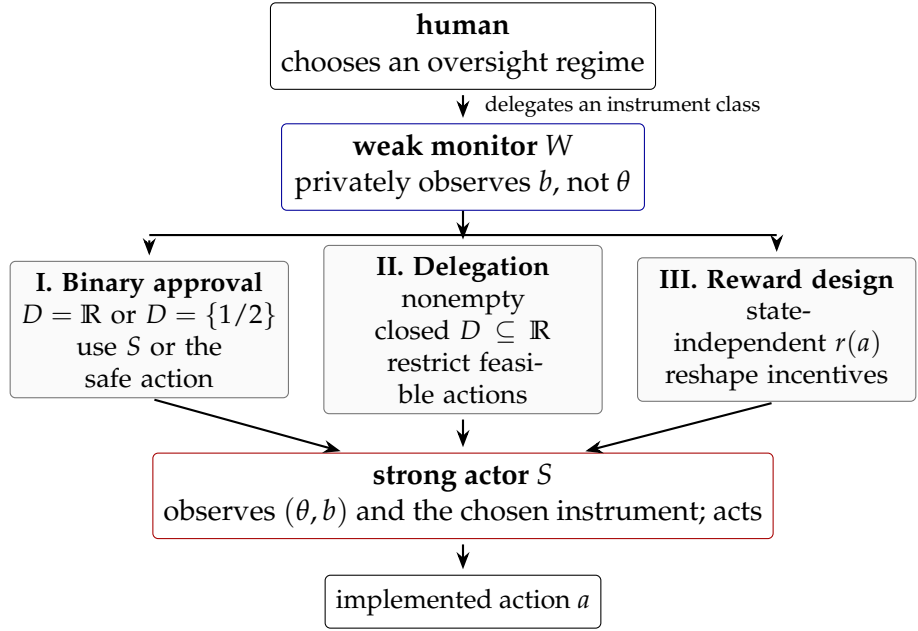
\begin{figure}[H]
\centering
\begin{tikzpicture}[
    font=\small,
    stage/.style={draw=black, rounded corners=2pt, align=center,
        minimum width=4.8cm, minimum height=1.05cm},
    instrument/.style={draw=black!55, fill=black!2, rounded corners=2pt,
        align=center, font=\footnotesize, text width=3.4cm, minimum height=1.15cm},
    outcome/.style={draw=black, rounded corners=2pt, align=center,
        font=\footnotesize, minimum width=3.5cm, minimum height=0.7cm},
    move/.style={-{Stealth[length=2.3mm]}, thick, black,
        shorten <=2pt, shorten >=2pt},
    lab/.style={font=\scriptsize, align=center}
]
\node[stage] (H) at (0,3.4) {\textbf{human}\\ chooses an oversight regime};
\node[stage, draw=blue!60!black] (W) at (0,1.8)
    {\textbf{weak monitor} $W$\\ privately observes $b$, not $\theta$};
\node[instrument] (I) at (-4.15,-0.35)
    {\textbf{I. Binary approval}\\$D=\Re$ or $D=\{1/2\}$\\use $S$ or the safe action};
\node[instrument] (II) at (0,-0.35)
    {\textbf{II. Delegation}\\nonempty closed $D\subseteq\Re$\\restrict feasible actions};
\node[instrument] (III) at (4.15,-0.35)
    {\textbf{III. Reward design}\\state-independent $r(a)$\\reshape incentives};
\node[stage, draw=red!65!black, minimum width=5.4cm] (S) at (0,-2.55)
    {\textbf{strong actor} $S$\\ observes $(\theta,b)$ and the chosen instrument; acts};
\node[outcome] (A) at (0,-3.95) {implemented action $a$};
\draw[move] (H.south) -- node[right=4pt,lab] {delegates an instrument class} (W.north);
\coordinate (fork) at (0,0.9);
\draw[thick] (W.south) -- (fork);
\draw[thick] (-4.15,0.9) -- (4.15,0.9);
\draw[move] (-4.15,0.9) -- (I.north);
\draw[move] (0,0.9) -- (II.north);
\draw[move] (4.15,0.9) -- (III.north);
\draw[move] (I.south) -- ($(S.north)+(-1.55,0)$);
\draw[move] (II.south) -- (S.north);
\draw[move] (III.south) -- ($(S.north)+(1.55,0)$);
\draw[move] (S.south) -- (A.north);
\end{tikzpicture}
\caption{Illustration of oversight game}
\label{fig:oversight-game}
\end{figure}

\cref{fig:oversight-game} illustrates the game. Notice that regimes are nested. Regime I restricts the monitor to two permission sets, so it is a special case of Regime II. Since delegation is a special case of rewards in which the image is $\{0,-\infty\}$, 
Regime II is therefore a special case of Regime III.
\paragraph{Regime I: binary approval.} Approval makes $S$ choose $a=\theta+b$ and gives $W$ loss $(b-w)^2$. Rejection gives $W$ loss $1/12+w^2$. Thus, breaking its ties in favor of approval, the monitor approves if and only if
\[
(b-w)^2\leq \frac1{12}+w^2.
\]
The monitor uses its diagnosis only to decide whether to use the strong actor $S$ but, if it does, cannot shape how $S$ acts. This is broadly analogous to trusted-monitor control, where a monitor decides whether to use an untrusted model's output or defer to a trusted fallback \citep{greenblatt2024control,xiao2026bootstrapped}.

\paragraph{Regime II: delegation.} The monitor now controls the feasible actions. One optimal representative is
\[
D^{II}(b)=
\begin{cases}
(-\infty,1+2w-b], & 0<b-w<1/2,\\
\Re, & b=w,\\
[2w-b,\infty), & -1/2<b-w<0,\\
\{w+1/2\}, & |b-w|\geq1/2.
\end{cases}
\]
When $S$ wants a larger action than $W$, the monitor imposes a cap. When $S$ wants a smaller action, it imposes a floor. Once relative bias reaches $1/2$, the monitor removes all state discretion and implements its own optimal fixed action $w+1/2$. Fine-grained privilege-control systems use the same logic when they restrict permissible tool calls and specify fallbacks for blocked calls \citep{shi2025progent}.

\paragraph{Regime III: reward design.} The affine schedule
\[
r_b^{III}(a)=-2(b-w)a
\]
makes the strong actor choose $a=\theta+w$. This action is first best for $W$, so no other reward can improve its payoff. This regime is in the spirit of scalable reward modeling: a weaker system helps shape the reward signal that trains or steers a stronger system when direct human evaluation is hard \citep{leike2018scalable,baker2025monitoring}. Variations of these ideas have been recently used. For instance, \Citet{kenton2026debate} use second model critiques the acting model before an LLM judge assigns the training reward and find that this design reduces reward hacking. Our regime is in this spirit, and asks how how AI-provided feedback changes the incentives of a stronger model.

\Needspace{0.44\textheight}
\begin{proposition}[Losses across oversight regimes]\label{prop:weak-strong}
Conditional on the realized bias $b$, the monitor's losses are
\[
\begin{aligned}
L_W^I(b)
&=\min\left\{(b-w)^2,\frac1{12}+w^2\right\},\\[3pt]
L_W^{II}(b)
&=
\begin{cases}
|b-w|^2-\dfrac43|b-w|^3, & |b-w|<1/2,\\[3pt]
\dfrac1{12}, & |b-w|\geq1/2,
\end{cases}\\[7pt]
L_W^{III}(b)
&=0.
\end{aligned}
\]
The human's losses are
\[
\begin{aligned}
L_H^I(b)
&=
\begin{cases}
b^2, & (b-w)^2\leq\dfrac1{12}+w^2,\\[3pt]
\dfrac1{12}, & (b-w)^2>\dfrac1{12}+w^2,
\end{cases}\\[5pt]
L_H^{II}(b)
&=L_W^{II}(b)+2w(b-w)(1-2|b-w|)_++w^2,\\[3pt]
L_H^{III}(b)
&=w^2,
\end{aligned}
\]
where $x_+:=\max\{x,0\}$.
Ex-ante losses are the expectations of these expressions over $b\sim\tau$.
\end{proposition}

\begin{proof}[Proof of \cref{prop:weak-strong}]
For this proof, write $\delta:=b-w$. We solve each regime. In Regime I, approval yields the monitor error $\delta$, while rejection yields the error $1/2-\theta-w$. Their squared expectations are $\delta^2$ and $1/12+w^2$. The human losses from the same choices are $b^2$ and $1/12$. This proves the Regime-I formulas under the stated tie-break.

For Regime II, translate actions by the monitor's bias: set $y=a-w$. The monitor wants $y=\theta$, while the strong actor wants $y=\theta+\delta$. The one-type interval reduction used in the proof of \cref{prop:cap} applies pointwise. If $\delta>0$, extending an interval's lower end cannot hurt the monitor, so we can optimize an upper cap $y\leq c$. Its loss has derivative
\[
\frac{\de L(c)}{\de c}
=
\begin{cases}
2c-1, & c\leq\delta,\\
\delta^2-(1-c)^2, & \delta<c<1+\delta,\\
0, & c\geq1+\delta.
\end{cases}
\]
Thus $c=1-\delta$ is optimal when $0<\delta<1/2$, while $c=1/2$ is optimal when $\delta\geq1/2$. Reflecting both the state and action gives the floor for $\delta<0$. Translating back by $w$ gives the displayed permission sets. Direct integration gives
\[
L_W^{II}(b)
=
\begin{cases}
|\delta|^2-\dfrac43|\delta|^3, & |\delta|<1/2,\\[3pt]
\dfrac1{12}, & |\delta|\geq1/2,
\end{cases}
\qquad
\Ex[a-\theta-w\mid b]=\delta(1-2|\delta|)_+.
\]
Since $a-\theta=(a-\theta-w)+w$, expanding the square gives
\[
L_H^{II}(b)
=L_W^{II}(b)+2w\delta(1-2|\delta|)_++w^2.
\]

In Regime III, strict concavity under $r_b^{III}$ gives the unique action $a=\theta+w$. The monitor obtains zero loss and the human obtains loss $w^2$.
\end{proof}

\cref{prop:weak-strong} compares payoffs for the human and weak monitor across three regimes. Since the instrument classes are nested and the weak monitor's payoff depends only on the action of the strong actor $S$ and the state, we have 
\[
L^I_W \geq L^{II}_W \geq L^{III}_W=0.  
\] 
But this ordering need not hold for the human---giving the weak monitor $W$ more freedom to shape the rewards or permissions of the strong actor $S$ need not improve payoffs precisely because $W$ might, itself, be biased.\footnote{Similar concerns were raised by the Roman author Juvenal; see \cite{hurwicz2008guard}.} The basic tradeoff here is familiar to the overarching theme of delegation: giving $W$ more discretion allows it to use its superior information about the strong agent's bias more precisely, and tailor rewards and permissions. But $W$ shapes the behavior of $S$ to its \emph{own} preference which, itself, is biased. Thus, the optimal regime must depend on the relative strengths of these biases.

\begin{figure}[H]
\centering
\includegraphics[width=0.98\textwidth]{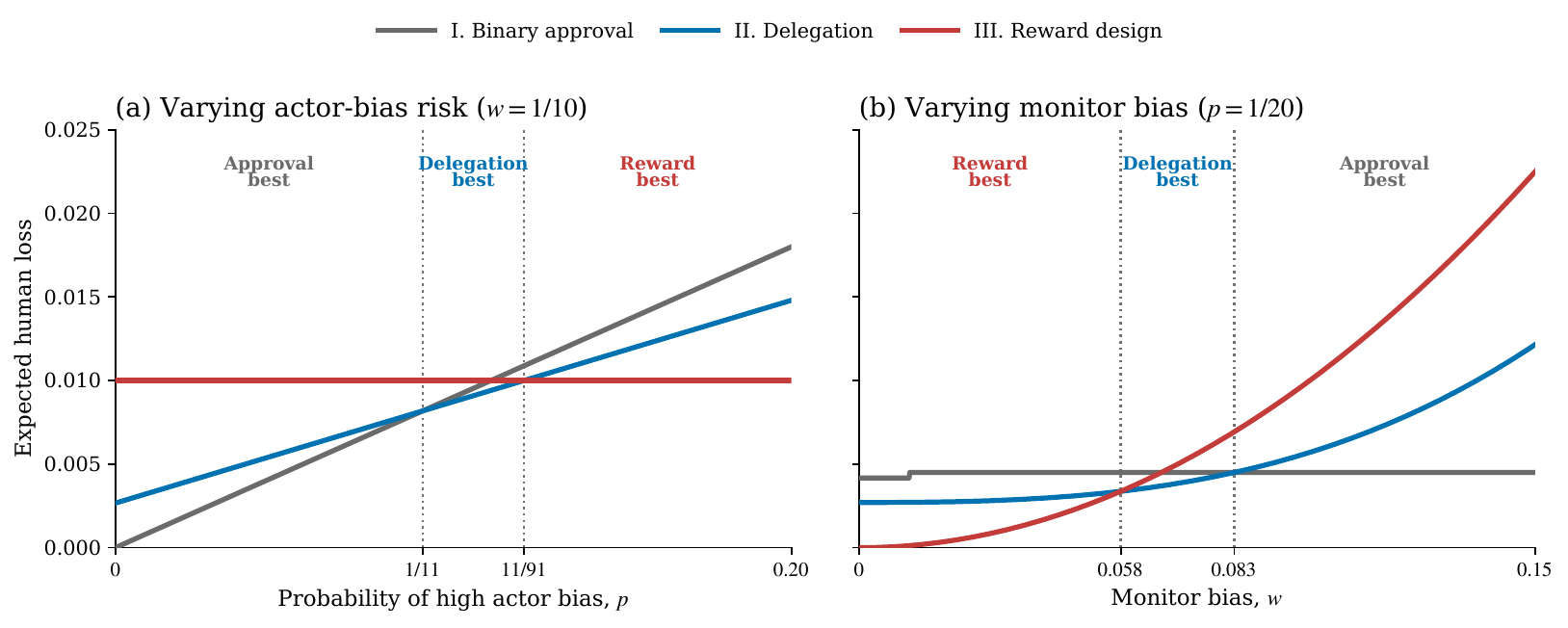}
\caption{\textbf{Each authority regime can be optimal for the human.} In both panels, the strong actor has bias $3/10$ with probability $p$ and zero bias otherwise. Panel~(a) fixes $w=1/10$ and varies $p$. Panel~(b) fixes $p=1/20$ and varies $w$. The dotted lines mark the values at which the human's preferred regime changes.}
\label{fig:instrument-choice-curves}
\end{figure}

We illustrate with the simplest numerical example. Let the strong actor's bias take values $b\in\{0,3/10\}$ and let $p=\Pr(b=3/10)$. Let  the weak monitor's known bias be $w > 0 $. The human's expected losses are depicted in \cref{fig:instrument-choice-curves}. Panel (a) fixes the weak monitor's bias at $w = 1/10$ and varies the probability that the strong actor is biased. We see that the lower-envelope is attained by different regimes, depending on how biased the strong actor is. As the probability that the strong actor is biased increases, we want to give the monitor more discretion to shape its rewards. Panel (b) fixes the distribution over the strong actor's bias and varies the monitor's (known) bias. Here the opposite logic holds---as the weak monitor becomes more biased, we want to give it less ability to shape rewards. 
%\textbf{Relation to the general framework.} Fold the private diagnosis $b$ into $W$'s type and let its input be the chosen approval decision, permission set, or reward. The communication rule sends that instrument to $S$, whose interim signal is $\theta$. The monitor can have a singleton object-action space while its utility depends on $S$'s action. Instrument authority is therefore a choice of the monitor's input space.

\paragraph{Optimal oversight.} Regimes I--III are specific oversight mechanisms. More generally, an oversight mechanism can be thought of as a set of reward schedules $\mathcal{M}\subseteq\mathcal{R}$ from which the weak monitor chooses.\footnote{Thus, $\mathcal{R}$ can be thought of as $W$'s action space, and in our framework the (human) mechanism designer delegates a subset.} Recall that $\mathcal{R}$ contains the reward schedules $r:\mathbb{R}\to\mathbb{R}\cup\{-\infty\}$ under which the strong actor has an optimal action. For instance, Regime III allowed the monitor to pick any schedule in $\mathcal R$ and always implemented $a = \theta + w$ which is the weak monitor's first best. 

We now construct a family $\mathcal{M}^* \subseteq \mathcal{R}$ that attains the human first best. Our construction will not require the human to know the monitor's bias $w$; it only uses a bound $\bar w\geq|w|$. Fix a pair of extreme actions $\bar a$ and $1-\bar a$ with $\bar a>1+\bar w+\sqrt{\bar w^2+2\bar w}$. We construct $\mathcal{M}^* =\{r_{\hat b}:\hat b\in\supp(\tau)\}$, where
\[
r_{\hat b}(a)=
\begin{cases}
-2\hat b a, & a\in[0,1],\\
(\bar a-1)^2-2\bar a\hat b, & a=\bar a,\\
(\bar a-1)^2-2(1-\bar a)\hat b, & a=1-\bar a,\\
-\infty, & \text{otherwise.}
\end{cases}
\]

\Needspace{0.18\textheight}
\begin{proposition}[Robust first best]\label{prop:optimal-oversight}
    Suppose the monitor's bias $w$ has a known bound $|w|\leq\bar w$. Under $\mathcal{M}^* \subseteq \mathcal{R}$, for every $w\in[-\bar w,\bar w]$ and every realization of the strong actor's bias $b$, the weak monitor uniquely chooses a reward $r \in \mathcal{M}^*$ that induces $a=\theta$ almost surely.
\end{proposition}

\cref{prop:optimal-oversight} states that whenever the weak monitor's bias is known to lie within some bounds, we can pick the menu of reward schedules in such a way that implements the first-best action $a = \theta$. Three reward schedules $r_{\hat b} \in \mathcal{M}^*$ are illustrated in \cref{fig:optimal-oversight-menu}. 

\begin{figure}[H]
\centering
\begin{tikzpicture}[font=\small]
\begin{scope}[x=4.2cm,y=1.5cm]
    \draw[-{Stealth[length=2mm]},thick] (-0.95,-2.0)--(2.0,-2.0);
    \draw[-{Stealth[length=2mm]},thick] (-0.95,-2.0)--(-0.95,1.6);
    \node[above=1pt,font=\scriptsize] at (1.97,-2.0) {$a$};
    \node[above right=0pt,font=\scriptsize] at (-0.95,1.52) {$r_{\hat b}(a)$};
    \draw[densely dotted] (-0.95,0)--(1.95,0);

    % linear tilts on [0,1]
    \draw[red!65!black,densely dashed,thick] (0,0)--(1,-0.30);
    \draw[thick] (0,0)--(1,-0.60);
    \draw[blue!65!black,densely dotted,thick] (0,0)--(1,-0.90);
    \draw[red!65!black,thick,fill=white] (1,-0.30) circle (1.7pt);
    \fill (1,-0.60) circle (1.8pt);
    \node[fill=blue!65!black,inner sep=1.6pt] at (1,-0.90) {};

    % rewards at the high extreme action
    \draw[red!65!black,thick,fill=white] (1.8,0.10) circle (1.7pt);
    \fill (1.8,-0.44) circle (1.8pt);
    \node[fill=blue!65!black,inner sep=1.6pt] at (1.8,-0.98) {};
    \node[font=\scriptsize,anchor=south] at (1.8,0.24) {$r_{\hat b}(\bar a)$};

    % rewards at the low extreme action
    \draw[red!65!black,thick,fill=white] (-0.8,0.88) circle (1.7pt);
    \fill (-0.8,1.12) circle (1.8pt);
    \node[fill=blue!65!black,inner sep=1.6pt] at (-0.8,1.36) {};
    \node[font=\scriptsize,anchor=west] at (-0.72,1.30) {$r_{\hat b}(1-\bar a)$};

    % schedule labels
    \node[font=\scriptsize,red!65!black,anchor=west] at (1.06,-0.30) {$\hat b=3/20$};
    \node[font=\scriptsize,anchor=west] at (1.06,-0.60) {$\hat b=b=3/10$};
    \node[font=\scriptsize,blue!65!black,anchor=west] at (1.06,-0.90) {$\hat b=9/20$};

    % prohibited regions
    \draw[densely dotted] (-0.8,0.74)--(-0.8,-1.62);
    \draw[densely dotted] (0,-0.06)--(0,-1.62);
    \draw[densely dotted] (1,-0.96)--(1,-1.62);
    \draw[densely dotted] (1.8,-1.10)--(1.8,-1.62);
    \draw[line width=1.2pt] (-0.92,-1.62)--(-0.83,-1.62);
    \draw[line width=1.2pt] (-0.77,-1.62)--(-0.03,-1.62);
    \draw[line width=1.2pt] (1.02,-1.62)--(1.77,-1.62);
    \draw[line width=1.2pt] (1.83,-1.62)--(1.95,-1.62);
    \node[font=\scriptsize] at (1.42,-1.80) {prohibited: $r_{\hat b}=-\infty$};

    \draw (-0.8,-2.02)--(-0.8,-1.98) node[below=2pt,font=\scriptsize] {$1-\bar a$};
    \draw (0,-2.02)--(0,-1.98) node[below=2pt,font=\scriptsize] {$0$};
    \draw (1,-2.02)--(1,-1.98) node[below=2pt,font=\scriptsize] {$1$};
    \draw (1.8,-2.02)--(1.8,-1.98) node[below=2pt,font=\scriptsize] {$\bar a$};
    \draw (-0.962,0)--(-0.938,0) node[left=2pt,font=\scriptsize] {$0$};
\end{scope}
\end{tikzpicture}
\caption{The reward menu $\mathcal M^*$ ($\bar a=9/5$; normalized so that $r_{\hat b}(0)=0$)}
\label{fig:optimal-oversight-menu}
\end{figure}
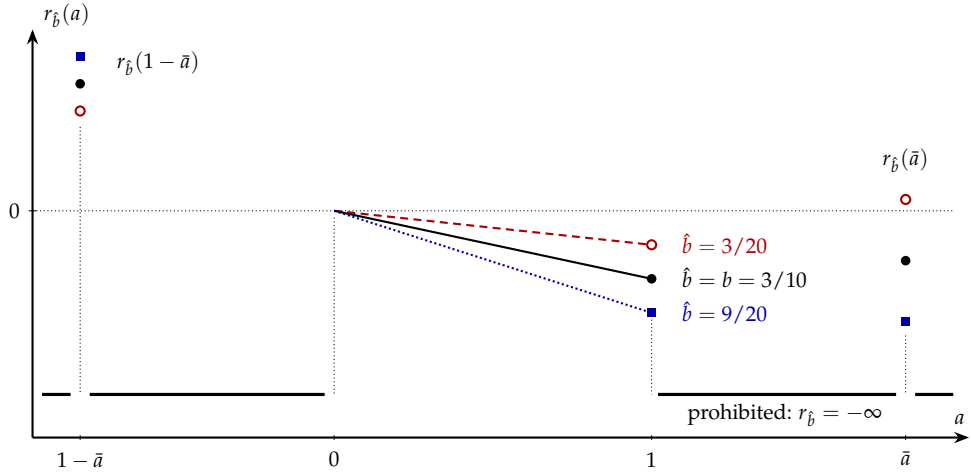

Each schedule $r_{\hat b} \in \mathcal{M}^*$ has three parts with distinct roles. First, the linear tilt of slope $-2\hat b$ over $[0,1]$ cancels a bias of exactly $\hat b$ which, recall, is $W$'s report of $S$'s bias. To see this, we can complete the square: 
\[
-(a-\theta-b)^2+r_{\hat b}(a)=-\bigl(a-\theta-(b-\hat b)\bigr)^2+\text{terms independent of }a. 
\]
Hence, facing the reward schedule $r_{\hat b}$, a strong actor with true bias $b$ behaves as if its bias were the \emph{residual} $b-\hat b$ i.e., it chooses the truncation of $\theta+(b-\hat b)$. Hence, if the weak monitor $W$ reports truthfully, this debiases the strong actor fully and implements $a = \theta$. 

Second, the reward at the isolated extreme action $\bar a$ is calibrated so that, comparing $\bar a$ with the top permitted regular action $a=1$,
\[
\underbrace{\bigl[-(\bar a-\theta-b)^2+r_{\hat b}(\bar a)\bigr]}_{\text{$S$'s payoff from choosing $\bar a$}}
-\underbrace{\bigl[-(1-\theta-b)^2+r_{\hat b}(1)\bigr]}_{\text{$S$'s payoff from choosing $a=1$}}
=2(\bar a-1)\bigl(\theta+(b-\hat b)-1\bigr).
\]
The actor therefore takes $\bar a$ exactly when its residual ideal point $\theta+(b-\hat b)$ exceeds $1$---which never happens after a truthful report. This happens precisely when the monitor $W$ understates its bias ($\hat b<b$) \emph{and} the state is high ($\theta>1-(b-\hat b)$).

Third, the reward at the other isolated extreme action $1-\bar a$ mirrors this: comparing $1-\bar a$ with the bottom permitted regular action $a=0$,
\[
\underbrace{\bigl[-(1-\bar a-\theta-b)^2+r_{\hat b}(1-\bar a)\bigr]}_{\text{$S$'s payoff from choosing $1-\bar a$}}
-\underbrace{\bigl[-(\theta+b)^2+r_{\hat b}(0)\bigr]}_{\text{$S$'s payoff from choosing $a=0$}}
=-2(\bar a-1)\bigl(\theta+(b-\hat b)\bigr).
\]
The actor therefore takes $1-\bar a$ exactly when its residual ideal point $\theta + (b - \hat b)$ falls below $0$. This happens precisely when the monitor $W$ \emph{overstates} its bias ($\hat b>b$) \emph{and} the state is low ($\theta<\hat b-b$). \cref{fig:optimal-oversight-ic} depicts the behavior the schedules induce in the strong agent, and the monitor's payoff from each schedule.

The basic reason the reward schedules include these extreme actions $\bar a$ and $1 - \bar a$ is to counteract the weak monitor's temptation to under or over-report the strong actor's bias. First suppose we had neither extreme action but the reward schedule retained the linear piece. Then, a weak monitor with bias $w>0$ would like to understate the strong actor's bias and report $\hat b = b - w$. Doing so would leave the strong actor with residual bias $w$ so the weak monitor obtains its first best $a = \theta + w$. Thsi was exactly the schedule chosen by the weak monitor in Regime III. Symmetrically, a monitor with $w<0$ would like to overstate and report $\hat b = b - w > b$. The way to deter these deviations is to force the reward schedule to tempt the strong actor with these distant actions---$\bar a$ to deter understatement, and $1-\bar a$ to deter overstatement. This induces a strong actor with bias $b$ to take these extreme actions when facing reward schedule $r_{\hat b}$ where  $\hat b \neq b$.

\begin{figure}[H]
\centering
\begin{tikzpicture}[font=\small]
% Panel (a): induced actions, state by state
\begin{scope}[x=4.6cm,y=1.5cm]
    \draw[-{Stealth[length=2mm]},thick] (0,0)--(1.12,0);
    \draw[-{Stealth[length=2mm]},thick] (0,-0.95)--(0,1.98);
    \node[font=\small] at (0.56,2.30) {(a) Induced actions of $S$};
    \node[below=3pt,font=\scriptsize] at (1.10,0) {$\theta$};
    \node[above=2pt,font=\scriptsize] at (0,1.98) {$a$};

    % guides
    \draw[densely dotted] (0,1)--(0.85,1);
    \draw[densely dotted] (0,1.8)--(0.83,1.8);
    \draw[densely dotted] (0.85,0)--(0.85,1);

    % truthful report: a = theta
    \draw[thick] (0,0)--(1,1);
    \fill (1,1) circle (1.8pt);
    % understatement bhat=3/20 (red)
    \draw[red!65!black,densely dashed,thick] (0,0.15)--(0.85,1);
    \draw[red!65!black,densely dotted,thick] (0.85,1)--(0.85,1.8);
    \draw[red!65!black,densely dashed,thick] (0.85,1.8)--(1,1.8);
    \draw[red!65!black,thick,fill=white] (1,1.8) circle (1.7pt);
    % overstatement bhat=9/20 (blue)
    \draw[blue!65!black,densely dotted,thick] (0,-0.8)--(0.15,-0.8);
    \draw[blue!65!black,densely dotted,thick] (0.15,-0.8)--(0.15,0);
    \draw[blue!65!black,densely dotted,thick] (0.15,0.012)--(1,0.862);
    \node[fill=blue!65!black,inner sep=1.6pt] at (1,0.85) {};

    % legend, upper left
    \draw[red!65!black,densely dashed,thick] (0.03,1.62)--(0.15,1.62);
    \draw[red!65!black,thick,fill=white] (0.09,1.62) circle (1.7pt);
    \node[font=\scriptsize,red!65!black,anchor=west] at (0.17,1.62) {$\hat b=3/20$ (understate)};
    \draw[thick] (0.03,1.46)--(0.15,1.46);
    \fill (0.09,1.46) circle (1.8pt);
    \node[font=\scriptsize,anchor=west] at (0.17,1.46) {$\hat b=b$ (truth): $a=\theta$};
    \draw[blue!65!black,densely dotted,thick] (0.03,1.30)--(0.15,1.30);
    \node[fill=blue!65!black,inner sep=1.6pt] at (0.09,1.30) {};
    \node[font=\scriptsize,blue!65!black,anchor=west] at (0.17,1.30) {$\hat b=9/20$ (overstate)};

    \draw (0.15,-0.03)--(0.15,0.03) node[below right=0pt and -2pt,font=\scriptsize] {$\tfrac3{20}$};
    \draw (0.85,-0.03)--(0.85,0.03) node[below=2pt,font=\scriptsize] {$\tfrac{17}{20}$};
    \draw (1,-0.03)--(1,0.03) node[below=2pt,font=\scriptsize] {$1$};
    \draw (-0.009,1)--(0.009,1) node[left=2pt,font=\scriptsize] {$1$};
    \draw (-0.009,1.8)--(0.009,1.8) node[left=2pt,font=\scriptsize] {$\bar a$};
    \draw (-0.009,-0.8)--(0.009,-0.8) node[left=2pt,font=\scriptsize] {$1-\bar a$};
\end{scope}

% Panel (b): S's payoff over actions, with and without the truthful schedule
\begin{scope}[xshift=7.8cm,yshift=2.32cm,x=2.05cm,y=1.53cm]
    \draw[-{Stealth[length=2mm]},thick] (-0.95,-2.45)--(2.02,-2.45);
    \draw[-{Stealth[length=2mm]},thick] (-0.95,-2.45)--(-0.95,0.42);
    \node[font=\small] at (0.52,0.74) {(b) $S$'s payoff under $r_b$};
    \node[above=1pt,font=\scriptsize] at (1.99,-2.45) {$a$};
    \draw[densely dotted] (-0.95,0)--(1.95,0);

    % state theta1 = 1/4 (red): raw payoff (dashed) and rewarded payoff (solid + dots)
    \draw[red!65!black,densely dashed,thick] plot[domain=-0.95:1.95,samples=90]
        (\x,{-(\x-0.55)^2});
    \draw[red!65!black,thick] plot[domain=0:1,samples=60]
        (\x,{-(\x-0.25)^2-0.24});
    \fill[red!65!black] (0.25,-0.24) circle (1.6pt);
    \fill[red!65!black] (1.8,-2.0025) circle (1.8pt);
    \fill[red!65!black] (-0.8,-0.7025) circle (1.8pt);

    % state theta2 = 3/4 (blue)
    \draw[blue!65!black,densely dashed,thick] plot[domain=-0.51:1.95,samples=90]
        (\x,{-(\x-1.05)^2});
    \draw[blue!65!black,thick] plot[domain=0:1,samples=60]
        (\x,{-(\x-0.75)^2-0.54});
    \fill[blue!65!black] (0.75,-0.54) circle (1.6pt);
    \fill[blue!65!black] (1.8,-1.0025) circle (1.8pt);
    \fill[blue!65!black] (-0.8,-2.3025) circle (1.8pt);

    % peak labels
    \node[font=\scriptsize,red!65!black,anchor=south] at (0.55,0.03) {$\theta_1\!+\!b$};
    \node[font=\scriptsize,blue!65!black,anchor=south] at (1.05,0.03) {$\theta_2\!+\!b$};
    \node[font=\scriptsize,red!65!black,anchor=north] at (0.25,-0.31) {$\theta_1$};
    \node[font=\scriptsize,blue!65!black,anchor=north] at (0.75,-0.61) {$\theta_2$};

    % guides to the extreme actions
    \draw[densely dotted] (1.8,-2.40)--(1.8,-1.06);
    \draw[densely dotted] (-0.8,-2.40)--(-0.8,-0.76);

    % legend: states and line styles
    \draw[red!65!black,thick] (0.10,-1.72)--(0.26,-1.72);
    \node[font=\scriptsize,red!65!black,anchor=west] at (0.30,-1.72) {$\theta_1=\tfrac14$};
    \draw[blue!65!black,thick] (0.10,-1.90)--(0.26,-1.90);
    \node[font=\scriptsize,blue!65!black,anchor=west] at (0.30,-1.90) {$\theta_2=\tfrac34$};
    \draw[thick] (0.10,-2.08)--(0.26,-2.08);
    \node[font=\scriptsize,anchor=west] at (0.30,-2.08) {under the mechanism};
    \draw[densely dashed,thick] (0.10,-2.26)--(0.26,-2.26);
    \node[font=\scriptsize,anchor=west] at (0.30,-2.26) {without the mechanism};

    \draw (-0.8,-2.48)--(-0.8,-2.42) node[below=2pt,font=\scriptsize] {$1-\bar a$};
    \draw (0,-2.48)--(0,-2.42) node[below=2pt,font=\scriptsize] {$0$};
    \draw (1,-2.48)--(1,-2.42) node[below=2pt,font=\scriptsize] {$1$};
    \draw (1.8,-2.48)--(1.8,-2.42) node[below=2pt,font=\scriptsize] {$\bar a$};
    \draw (-0.962,0)--(-0.938,0) node[left=2pt,font=\scriptsize] {$0$};
\end{scope}
\end{tikzpicture}
\caption{Strong actor's best response and payoffs}
\label{fig:optimal-oversight-ic}
\end{figure}
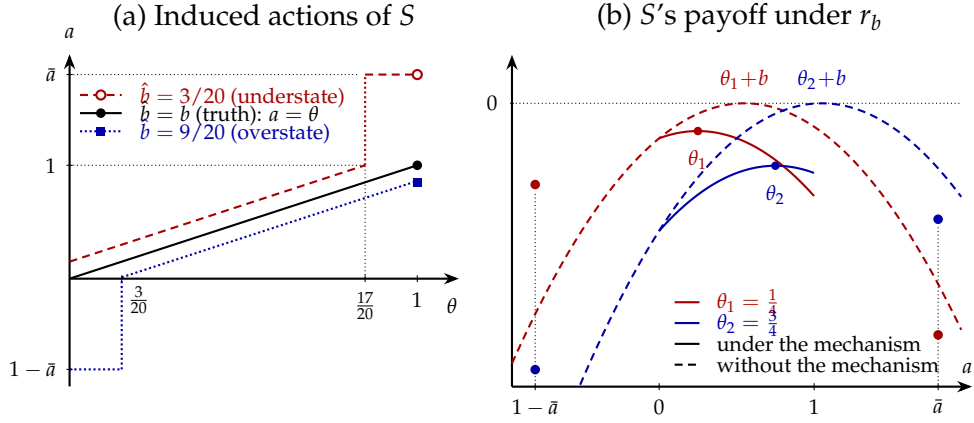 
Notice that the set $\mathcal{M}^*$ uses no information about the monitor's bias beyond the bound $\bar w$. This is because incentives here are asymmetric: the monitor's gain from misreporting is at most of order $|w|$, while its loss from the strong actor taking an extreme action $\in \{\bar a, 1- \bar a\}$ is of order $\bar a^2$. Hence, sufficiently extreme actions can deter every monitor with $|w|\leq\bar w$.

\begin{proof}[Proof of \cref{prop:optimal-oversight}]
Each schedule is admissible because its finite-valued restriction is continuous on the compact set $\{1-\bar a\}\cup[0,1]\cup\{\bar a\}$, so the actor's payoff attains a maximum. Fix a true bias $b\in\supp(\tau)$ and a report $\hat b\in\supp(\tau)$. By the three displays above, the actor chooses the truncation of $\theta+(b-\hat b)$ to $[0,1]$ when $\theta+(b-\hat b)\in[0,1]$, chooses $\bar a$ when $\theta+(b-\hat b)>1$, and chooses $1-\bar a$ when $\theta+(b-\hat b)<0$.

If $\hat b=b$, the actor chooses $a=\theta$ at every state except $\theta\in\{0,1\}$, where it is indifferent with an extreme action; these are null events. The monitor's expected loss is $w^2$.

If $\hat b<b$, write $\delta:=b-\hat b$ and suppose first that $\delta\leq1$. The low extreme action is never taken: the actor chooses $\theta+\delta$ on $\theta<1-\delta$ and $\bar a$ on $\theta>1-\delta$, so the monitor's expected loss in excess of $w^2$ is
\begin{align*}
\underbrace{(1-\delta)\bigl[(\delta-w)^2-w^2\bigr]}_{\text{shifted interior actions}}
\;&+\;
\underbrace{\int_{1-\delta}^{1}\bigl[(\bar a-\theta-w)^2-w^2\bigr]\de\theta}_{\text{states sent to }\bar a} \\ 
\;&\geq\;
\delta\Bigl[(1-\delta)(\delta-2w)+(\bar a-1-w)^2-w^2\Bigr] \\ 
\;&\geq\;
\delta\Bigl[(\bar a-1-\bar w)^2-\bar w^2-2\bar w\Bigr]
\;>\;0.
\end{align*}
The first inequality uses $\bar a-\theta-w\geq\bar a-1-w>0$; the second uses $(1-\delta)(\delta-2w)\geq-2|w|\geq-2\bar w$, which follows by considering $w\geq0$ and $w<0$, together with $(\bar a-1-w)^2-w^2=(\bar a-1)^2-2(\bar a-1)w\geq(\bar a-1-\bar w)^2-\bar w^2$; the third is the choice of $\bar a$. Whatever the monitor gains from shifting interior actions toward $\theta+w$, it loses more from the states sent to $\bar a$. If $\delta>1$, the actor takes $\bar a$ in every state, and the monitor's loss is at least $(\bar a-1-\bar w)^2>\bar w^2+2\bar w\geq w^2$.

If $\hat b>b$, the deviation is the mirror image: with $e:=\hat b-b\leq1$, the actor chooses $\theta-e$ on $\theta>e$ and $1-\bar a$ on $\theta<e$. Reflecting states and actions around $1/2$ maps this case to understatement with biases $(-b,-w)$, so the same bounds apply, including when $e>1$.

Truth is therefore the monitor's uniquely optimal choice in $\mathcal M^*$ for every $|w|\leq\bar w$, and it induces $a=\theta$ almost surely for every $b\in\supp(\tau)$. This proves the claim.
\end{proof}

\end{romanexample}

\section{Discussion}\label{sec:literature}
Just as how economists study and design human incentives without understanding how the brain produces each choice, we have taken a functionalist view of AI systems as agents that respond to incentives. In so doing, we have made behavior---rather than the internal workings of neural nets---the basic object of analysis. This `Humean' view of AI behavior as being driven by beliefs and preferences might be viewed literally, metaphorically, or normatively. 

On the literal view, modern AI systems are post-trained for goal-directed behavior \citep{ouyang2022training,bai2022training}. Reinforcement learning increases the probability of actions that receive high rewards and decreases the probability of actions that receive low rewards. Indeed, a recent turn to make AI systems economically valauble---able to perform tasks in the world---has also made them behave as optimizers. On this view, their `preferences' are simply whatever has been reinforced.  

%If the resulting objectives persist at deployment, preferences describe the goals that emerge from training. This possibility motivates work on goal misgeneralization \citep{shah2022goal}.

On the metaphorical view, preferences need not describe any internal mental state. Instead they are, in the revealed preference tradition,  nothing more than a compact account of stable choices \citep{gul2008case}. This reading has some empirical support. \citet{chen2023emergence} find that GPT choices often satisfy standard consistency tests in allocation problems. \citet{mazeika2025utility} find that pairwise AI choices become more transitive and fit utility models better as models scale. On this view, our conception of AI systems as having preferences are simply a useful analytical tool to characterize how AI systems behave.\footnote{Measured preferences can depend on the prompt, so our results can be interpreted as conditional on the prompt.}

On the normative view, preference maximization is a benchmark for coherent agency. We should expect strong AI systems that are deployed in markets and compete against each other to exhibit coherent preferences (e.g., such as to not be money-pumped).

\paragraph{Future directions.} We conclude with some (specualtive) directions for future work. 

First, our framework, while general, is ultimately static. In practice, AI agents act over long trajectories, performing complex sequences of actions that depend on each other. Of course, these trajectories of actions $(a_t)_t$ can be thought of as a single choice i.e., a member $\in A$. There is no conceptual difficulty here since we can simply set $A$ as the space of all trajectories and, on this view, the reward $r: A \to \mathbb{R}$ is over the full path of actions.\footnote{This is similar to how, in multi-stage games, we can associate a payoffs with `within period' play and take the discounted sum, or associate payoffs directly with the path of all players' actions.} But we expect there to be analytical and computational difficulties in analyzing optimal mechanisms when the action space is large, and more structure of how time-$t$ actions aggregate into `total payoffs' (e.g., discounting the sum) can lend more tractability. Perhaps more fundamentally, the designer might re-optimize the mechanism based on agents' past actions. We think this calls for a more serious treatment of dynamic mechanism and information design. 

Second, AI agents might, in practice, be uncertain about whether it is facing a test or is in deployment (what is often called `evaluation awareness.') These kinds of \emph{de se} uncertainty about where oneself is can be exploited to discipline behavior and screen misaligned agents \citep{chen2024imperfect}. Self-locating information could even be designed to implement good behavior \citep{koh2025memory}. Indeed, a model's belief about whether it is in testing or deployment depends on the realism of the evaluation \citep{openai2026deployment,krakovna2026honeypot} which suggests that the design of self-locating information is empirically relevant.

Third, we might want mechanisms to be robust to what AI agents can do (in the spriit of \cite{carroll2015robustness}) and/or what information they have about the world or about each other (in the spirit of \cite{bergemann2005robust}). Such notions of robustness have been used in the context of using AI's recommendations \citep{fudenberg2025friend,dworczak2026robust}. We think the tools of robust mechanism design can be similarly brought to bear on shaping the preferences and permissions of agentic AI systems.

\setlength{\bibsep}{0pt}
\bibliography{ref}

\clearpage
\appendix 

\section{Omitted Proofs}\label[appendix]{sec:proofs}

\begin{proof}[Proof of \cref{lem:continuum-ironing}]
We first restrict the range. We then compare the two objectives on finite partitions, pass to the continuum, and choose a pointwise feasible version of the limit.

Let the least-squares loss be \(Q\), and let the transformed loss be \(J\):
\[
Q(z)
:=
\int\bigl[z(c)-b(c)^2\bigr]^2\,\de F(c),
\qquad
J(z)
:=
\int\left[
\frac{z(c)^{3/2}}{3}
-b(c)^2\sqrt{z(c)}
\right]\de F(c).
\]

\emph{Step 1: restrict the range.}
Define
\[
\overline z(c)
:=
\left(1-\min\{c,1/2\}\right)^2,
\qquad
(Tz)(c)
:=
\min\{z(c),\overline z(c)\}.
\]
The assumptions on bias imply
\[
b(c)^2
\leq
\begin{cases}
c^2\leq(1-c)^2,&c\leq1/2,\\[3pt]
1/4,&c\geq1/2,
\end{cases}
=
\overline z(c),
\]
while \((1-c)^2\leq\overline z(c)\leq1\). The minimum of two weakly decreasing functions is weakly decreasing, so \(Tz\) is feasible whenever \(z\) is feasible in problem~\eqref{eqn:iron}. On the set where \(z>\overline z\),
\[
\bigl[z-b(c)^2\bigr]^2
-\bigl[\overline z-b(c)^2\bigr]^2
=
\bigl[z-\overline z\bigr]
\bigl[z+\overline z-2b(c)^2\bigr]
>0.
\]
Thus truncation weakly lowers \(Q\), and lowers it strictly if it changes \(z\) on a set of positive \(F\)-probability. We can restrict attention to
\[
\mathcal Z
:=
\left\{
\begin{array}{l}
z:\ z\text{ is measurable and pointwise weakly decreasing},\\
(1-c)^2\leq z(c)\leq\overline z(c)
\end{array}
\right\}.
\]
For every \(z\in\mathcal Z\),
\[
\begin{aligned}
z(c)&=(1-c)^2
&&\text{for }c\leq1/2,\\
(1-c)^2&\leq z(c)\leq\frac14
&&\text{for }c\geq1/2.
\end{aligned}
\]

\emph{Step 2: compare the finite problems.}
If \(\overline c\leq1/2\), the bounds fix \(z(c)=(1-c)^2\), so the result follows. Suppose instead that \(\overline c>1/2\). Take a finite ordered partition \((I_k)_{k=1}^m\) of
\[
\left[\max\{\underline c,1/2\},\overline c\right]
\]
into intervals of positive \(F\)-probability. For each cell, define
\[
w_k:=F(I_k),
\qquad
y_k:=\Ex\!\left[b(c)^2\mid c\in I_k\right],
\qquad
\ell_k:=\sup_{c\in I_k}(1-c)^2.
\]
The feasible cell values form the compact convex set
\[
\mathcal Z_m
:=
\left\{
(z_1,\ldots,z_m):
z_1\geq\cdots\geq z_m,
\quad
\ell_k\leq z_k\leq\frac14
\text{ for every }k
\right\}.
\]
Each \(\ell_k\) is positive. Up to a term in \(Q\) that does not depend on the cell values, the two objectives for cell-constant functions are
\[
Q_m(z)
:=
\sum_{k=1}^m w_k(z_k-y_k)^2,
\qquad
J_m(z)
:=
\sum_{k=1}^m w_k\left[
\frac{z_k^{3/2}}3-y_k\sqrt{z_k}
\right].
\]
For every feasible \(z_k>0\),
\[
\frac{\partial J_m}{\partial z_k}
=\frac{1}{4\sqrt{z_k}}
\frac{\partial Q_m}{\partial z_k},
\qquad
\frac{\partial^2 J_m}{\partial z_k^2}
=\frac{w_k(z_k+y_k)}{4z_k^{3/2}}>0.
\]

The positive weights make \(Q_m\) strictly convex. Let \(z^m=(z_1,\ldots,z_m)\) be its unique minimizer on \(\mathcal Z_m\). Because the feasible set is a convex polyhedron, its first-order conditions give nonnegative multipliers \(\alpha_k\), \(\lambda_k\), and \(\mu_k\) for the order, lower-bound, and upper-bound constraints. With \(\alpha_0=\alpha_m=0\), stationarity requires
\[
\frac{\partial Q_m}{\partial z_k}
+\alpha_{k-1}-\alpha_k-\lambda_k+\mu_k
=0
\qquad(k=1,\ldots,m),
\]
with the corresponding complementary-slackness conditions. Set
\[
s_k:=\frac{1}{4\sqrt{z_k}},
\qquad
\widetilde\alpha_k:=s_k\alpha_k,
\qquad
\widetilde\lambda_k:=s_k\lambda_k,
\qquad
\widetilde\mu_k:=s_k\mu_k.
\]
Set \(\widetilde\alpha_0=\widetilde\alpha_m=0\). 
If \(\alpha_k>0\), complementary slackness gives \(z_k=z_{k+1}\), and hence \(s_k=s_{k+1}\). Multiplying each stationarity equation by \(s_k\) therefore gives
\begin{align*}
\frac{\partial J_m}{\partial z_k}
+\widetilde\alpha_{k-1}-\widetilde\alpha_k
-\widetilde\lambda_k+\widetilde\mu_k
&=
s_k\left[
\frac{\partial Q_m}{\partial z_k}
+\alpha_{k-1}-\alpha_k-\lambda_k+\mu_k
\right]\\
&=0.
\end{align*}
The scaled multipliers remain nonnegative and preserve complementary slackness. They therefore satisfy the first-order conditions for \(J_m\). The displayed positive second derivative makes \(J_m\) strictly convex, so \(z^m\) is also its unique minimizer.

\emph{Step 3: pass to the continuum.}
Take nested interval partitions \((\mathcal P_m)_m\) whose mesh converges to zero, and let \(z_m\) be the unique minimizer of both finite problems. Set \(z_m(c)=(1-c)^2\) for \(c\leq1/2\). Bounded monotone functions have a pointwise convergent subsequence at every continuity point of a monotone limit. Thus, along a subsequence, \(z_m\) converges to a weakly decreasing function \(z^\circ\). A monotone function has at most countably many discontinuities, so atomlessness makes this convergence \(F\)-almost sure. Boundedness then gives
\[
Q(z_m)\longrightarrow Q(z^\circ),
\qquad
J(z_m)\longrightarrow J(z^\circ),
\]
where we relabel the subsequence.

Fix any \(z\in\mathcal Z\). For the cell \(I_m(c)\in\mathcal P_m\) that contains \(c>1/2\), define the upper step approximation
\[
z^{(m)}(c)
:=
\begin{cases}
(1-c)^2,&c\leq1/2,\\[3pt]
\displaystyle\sup_{x\in I_m(c)}z(x),&c>1/2.
\end{cases}
\]
This function belongs to the finite feasible set and converges to \(z\) at every continuity point of \(z\), hence \(F\)-almost surely. Finite optimality and dominated convergence imply
\[
Q(z^\circ)\leq Q(z),
\qquad
J(z^\circ)\leq J(z).
\]
Thus \(z^\circ\) minimizes both objectives over \(\mathcal Z\), and its bounds hold \(F\)-almost surely.

\emph{Step 4: construct a pointwise solution.}
Define
\[
\widehat z(c)
:=
\begin{cases}
(1-c)^2,&c\leq1/2,\\[4pt]
\displaystyle\lim_{x\downarrow c}z^\circ(x),
&1/2<c<\overline c,\\[7pt]
\displaystyle\lim_{x\uparrow\overline c}z^\circ(x),
&c=\overline c>1/2.
\end{cases}
\]
The one-sided limits exist, and \(\widehat z\) differs from \(z^\circ\) only at countably many points. It is weakly decreasing on each side of \(1/2\). If \(1/2\) lies in the capability interval, the almost-sure upper bound and full support imply
\[
\lim_{x\downarrow1/2}z^\circ(x)\leq\frac14=\widehat z(1/2),
\]
so \(\widehat z\) is weakly decreasing on the full interval.

Suppose that either bound fails at some \(c_0\). One-sided continuity of \(\widehat z\) and continuity of the bounds then give a relative one-sided interval on which the same bound fails. Full support gives this interval positive probability, contradicting the almost-sure bounds. Therefore
\[
(1-c)^2
\leq
\widehat z(c)
\leq
\left(1-\min\{c,1/2\}\right)^2
\qquad\text{for every }c.
\]
Hence \(\widehat z\in\mathcal Z\). Since \(\widehat z=z^\circ\) \(F\)-almost surely, it minimizes both objectives.

Finally, \(Q\) is strictly convex. The transformed integrand is also strictly convex wherever \(z>0\), since
\[
\frac{\partial^2}{\partial z^2}
\left[
\frac{z^{3/2}}3-b(c)^2\sqrt z
\right]
=
\frac{1}{4\sqrt z}
+\frac{b(c)^2}{4z^{3/2}}
>0.
\]
The lower bound gives \(z(c)>0\) for every \(c<1\), and atomlessness gives \(F(\{1\})=0\). Each objective therefore has at most one minimizer up to \(F\)-null sets. Step~1 shows that every minimizer of problem~\eqref{eqn:iron} lies in \(\mathcal Z\) up to a null set. Setting \(z^*:=\widehat z\) proves all the claims.
\end{proof}

\begin{proof}[Proof of \cref{prop:cap}]
We first reduce the mechanism to one cap. We then optimize that cap.

\emph{Step 1: reduce the mechanism to one cap.} Fix any incentive-compatible mechanism and let \(D_b\) be the nonempty closed permission set assigned to report \(b\). Because \(D_b\) is nonempty and closed, every target has a nearest point in \(D_b\). Condition~\eqref{eq:cap-interiority} implies
\[
0\leq b\leq1-\sqrt{\Ex[b^2]}\leq1
\qquad \text{$\tau$-almost surely},
\]
so every bias moment used below is finite. We now prove directly that one cap weakly dominates the assigned permission sets.
For a closed set $D$, let $d_D(x)$ be the distance from $x$ to $D$. Define the agent's loss and the human's loss when the bias is $t$ by
\[
A_D(t):=\int_t^{t+1}d_D(x)^2\de x,
\qquad
H_D(t):=\int_0^1\bigl(a_D(\theta,t)-\theta\bigr)^2\de\theta,
\]
where $a_D(\theta,t)$ is a nearest point in $D$ to $\theta+t$. Ties occur only at gap midpoints and do not affect either integral.

We first need one distance inequality. If $d$ is a nonnegative, one-Lipschitz function on $[0,1]$ and $y=\int_0^1d(x)^2\de x$, then
\[
d(1)^2-d(0)^2\leq \Phi(y),
\qquad
\Phi(y):=
\begin{cases}
(3y)^{2/3},&0\leq y\leq 1/3,\\
2\sqrt{y-1/12},&y\geq1/3.
\end{cases} \tag{*}
\]
To prove this inequality, put $p=d(1)$. If $p\leq1$, the Lipschitz property gives
\[
y\geq\int_{1-p}^1(p-1+x)^2\de x=\frac{p^3}{3}.
\]
Thus $d(1)^2-d(0)^2\leq p^2\leq(3y)^{2/3}$ when $y\leq1/3$; when $y\geq1/3$, the same difference is at most $1\leq2\sqrt{y-1/12}$. If $p\geq1$, then $d(0)\geq p-1$ and
\[
y\geq\int_0^1(p-1+x)^2\de x
=p^2-p+\frac13.
\]
It follows that
\[
d(1)^2-d(0)^2
\leq p^2-(p-1)^2
=2p-1
\leq2\sqrt{y-1/12}.
\]
This proves~$(*)$.

Fix a set $D$ and a bias $b$. There is a unique $c\leq b+1$ such that the cap $K_c:=(-\infty,c]$ gives this type the same loss:
\[
A_{K_c}(b)=A_D(b).
\]
Existence follows because $A_{K_c}(b)$ rises continuously from zero to infinity as $c$ falls from $b+1$ to $-\infty$. We claim that
\[
A_D(t)\leq A_{K_c}(t)
\qquad\text{for every }t\geq b. \tag{**}
\]
Indeed, differentiation under the integral gives
\[
A_D'(t)=d_D(t+1)^2-d_D(t)^2\leq\Phi\bigl(A_D(t)\bigr),
\]
where the inequality applies~$(*)$ to the translated distance function on $[t,t+1]$. The cap attains equality because its distance function is the upper ramp $x\mapsto(x-c)_+$. Hence
\[
A_{K_c}'(t)=\Phi\bigl(A_{K_c}(t)\bigr)
\qquad\text{for }t\geq b.
\]
Let $Q(y):=\int_0^y1/\Phi(v)\de v$. This function is finite and strictly increasing. Wherever the loss is positive,
\[
\frac{\de}{\de t}Q\bigl(A_D(t)\bigr)\leq1
=\frac{\de}{\de t}Q\bigl(A_{K_c}(t)\bigr).
\]
The two losses agree at $t=b$, so integration proves~$(**)$; zero-loss points follow by continuity.

The same cap also weakly lowers the human's loss for type $b$. The envelope formula gives
\[
A_D'(b)
=-2\int_0^1\bigl(a_D(\theta,b)-\theta-b\bigr)\de\theta,
\]
and therefore
\[
H_D(b)=A_D(b)-bA_D'(b)+b^2.
\]
At $t=b$, inequality~$(*)$ gives $A_D'(b)\leq A_{K_c}'(b)$. Since $b\geq0$ and the two agent losses agree at $b$, we obtain
\[
H_{K_c}(b)\leq H_D(b). \tag{***} \label{eq:triplestar}
\]

Apply this cap replacement to every assigned set $D_b$ and write $c(b)$ for its endpoint. Take any $b<b'$ in $\supp(\tau)$. Completeness lets type $b'$ report $b$. Incentive compatibility and~$(**)$ therefore imply
\[
A_{K_{c(b')}}(b')
=A_{D_{b'}}(b')
\leq A_{D_b}(b')
\leq A_{K_{c(b)}}(b').
\]
The loss from a cap is strictly decreasing in its endpoint below $b'+1$. Since $c(b)\leq b+1<b'+1$ and $c(b')\leq b'+1$, this inequality requires
\[
c(b)\leq c(b').
\]
Thus $c(b)$ is weakly increasing and hence measurable on $\supp(\tau)$. Meanwhile, ($***$) shows that each loss-matched cap weakly lowers the human's loss for its matched type. We use these caps only as comparison objects; their report-dependent profile need not preserve every incentive constraint.

It remains to pool these ordered caps. The human's loss for known bias $b$ is minimized at
\[
q(b):=\max\{1/2,1-b\},
\]
as the derivative calculated in Step~2 below confirms. The function $q(b)$ is weakly decreasing, while $c(b)$ is weakly increasing. For each $b$, consider the closed interval between $c(b)$ and $q(b)$. If $b<b'$, then
\[
\min\{c(b),q(b)\}
\leq c(b')
\leq\max\{c(b'),q(b')\}
\]
and
\[
\min\{c(b'),q(b')\}
\leq q(b)
\leq\max\{c(b),q(b)\}.
\]
These inequalities show that the lower endpoint of any such interval is no greater than the upper endpoint of any other. Taking the supremum of the lower endpoints and the infimum of the upper endpoints gives
\[
\sup_{b\in\supp(\tau)}\min\{c(b),q(b)\}
\leq
\inf_{b\in\supp(\tau)}\max\{c(b),q(b)\}.
\]
Both bounds are finite because $1/2\leq q(b)\leq1$. Choose a common endpoint $c$ between them. It lies between $c(b)$ and $q(b)$ for every supported bias. Moving a cap toward $q(b)$ weakly lowers the human's loss, so ($***$) gives
\[
H_{K_c}(b)
\leq H_{K_{c(b)}}(b)
\leq H_{D_b}(b)
\qquad\text{for every }b\in\supp(\tau).
\]
Assigning $K_c$ after every report is incentive compatible. We can therefore restrict attention to one cap $(-\infty,\bar a]$.

\emph{Step 2: optimize the cap.} Let $L_b(\bar a)$ denote the human's expected loss for type $b$. For $\bar a\leq1$, direct integration and differentiation give
\[
\frac{\de L_b(\bar a)}{\de \bar a}=
\begin{cases}
2\bar a-1,&\bar a\leq b,\\
b^2-(1-\bar a)^2,&b<\bar a.
\end{cases}
\]
Condition~\eqref{eq:cap-interiority} bounds $b$ by one, so dominated convergence justifies taking expectations of these derivatives.
Set $\bar a^*=1-\sqrt{\Ex_{b\sim\tau}[b^2]}$. Condition~\eqref{eq:cap-interiority} gives $b\leq\bar a^*\leq1$ almost surely. If $\bar a<0$, then $\frac{\de L_b(\bar a)}{\de \bar a}=2\bar a-1<0$ for every type. If $0\leq\bar a<\bar a^*$, then
\[
\Ex_{b\sim\tau}\left[\frac{\de L_b(\bar a)}{\de \bar a}\right]
\leq
\Ex_{b\sim\tau}[b^2]-(1-\bar a)^2<0,
\]
because $2\bar a-1\leq b^2-(1-\bar a)^2$ whenever $b\geq\bar a$. If $\bar a^*<\bar a\leq1$, interiority gives $b<\bar a$ almost surely, so
\[
\Ex_{b\sim\tau}\left[\frac{\de L_b(\bar a)}{\de \bar a}\right]
=\Ex_{b\sim\tau}[b^2]-(1-\bar a)^2>0.
\]
For $\bar a\geq1$, a binding cap has derivative $b^2-(\bar a-1)^2>0$, while a cap that does not bind has derivative zero. Loss is therefore nondecreasing above one. These inequalities prove part~(i).

It remains to calculate the value. Condition~\eqref{eq:cap-interiority} gives $0\leq b\leq\bar a^*\leq1$. More generally, whenever $b\leq\bar a\leq1$ almost surely,
\[
L(\bar a)
=\Ex_{b\sim\tau}\!\left[
\int_0^{\bar a-b}b^2\,d\theta
+\int_{\bar a-b}^1(\bar a-\theta)^2\,d\theta
\right]
=\Ex_{b\sim\tau}\!\left[
\bar a b^2-\frac{2}{3}b^3+\frac{(1-\bar a)^3}{3}
\right].
\]
Substituting $\bar a^*=1-\sqrt{\bar b^2+\sigma^2}$ gives part~(ii). Assigning $D^*$ after every report is feasible and incentive compatible, so the bound is attained.
\end{proof}

\end{document}